\documentclass[sts]{imsart}
\usepackage{amsmath, amsthm, amssymb, latexsym,algorithm,multirow,enumerate,url,booktabs,xcolor,bbm,soul, natbib, fullpage, caption}
\usepackage{float} 
\usepackage{booktabs} 
\usepackage{graphicx} 
\usepackage{mathtools}
\usepackage{natbib, url ,amsmath, amssymb, latexsym, appendix}
\usepackage{subcaption}
\usepackage{algorithm}
\usepackage{algpseudocode}
\usepackage{graphicx}
\RequirePackage[colorlinks,citecolor=blue,urlcolor=blue]{hyperref}
\algnewcommand\Input{\item[\textbf{Input:}]} 

\newtheorem{theorem}{Theorem}

\newtheorem{lemma}{Lemma}

\allowdisplaybreaks
\def\E {\text{E}}
\def\vech {\text{vech}}

\def\tr {\text{tr}}
\def\cov {\text{cov}}

\DeclareMathOperator{\logit}{logit}

\def\e{{\rm e}}                        
\def\tr{\text{\rm tr}}
\def\vec{\text{\rm vec}}

\def\cov{\text{\rm cov}}
\def\Var{\text{\rm Var}}
\def\logit{\text{\rm logit}}
\def\KL{\text{\rm KL}}

\def\F{\text{F}}
\def\S{\text{S}}

\def\E{\text{\mathbb E}}
\def\tr{\text{\rm tr}}

\def\diag{\text{\rm diag}}

\def\dg{\text{\rm dg}}
\def\vech{\text{\rm vech}}

\def\L{{\mathcal{L}}}

\def\N{{\text{N}}}
\def\IG{{\text{IG}}}
\def\E{{\text{E}}}
\def\SN{{\text{SN}}}

\def\bfone{{\bf 1}}

\graphicspath{{img/}}
\startlocaldefs

\endlocaldefs

\begin{document}

\begin{frontmatter}
\title{Weighted Fisher divergence for 
high-dimensional Gaussian \\
variational inference}
\runtitle{Weighted Fisher divergence for 
high-dimensional Gaussian
variational inference}

\begin{aug}
\author[A]{\fnms{Aoxiang}~\snm{Chen}\ead[label=e1]{e0572388@u.nus.edu}},
\author[B]{\fnms{David J.}~\snm{Nott}\ead[label=e2]{standj@nus.edu.sg}}
\and
\author[C]{\fnms{Linda S. L.}~\snm{Tan}\ead[label=e3]{statsll@nus.edu.sg}}

\address[A]{Department of Statistics and Data Science, National University of Singapore\printead[presep={\ }]{e1}.}
\address[B]{Department of Statistics and Data Science, National University of Singapore\printead[presep={\ }]{e2}.}
\address[C]{Department of Statistics and Data Science, National University of Singapore\printead[presep={\ }]{e3}.}
\end{aug}

\begin{abstract}
Bayesian inference has many advantages for complex models, but standard Monte Carlo methods for summarizing the posterior  can be computationally demanding, and it is attractive to consider optimization-based variational methods. Our work considers Gaussian approximations with sparse precision matrices which are tractable to optimize in high-dimensions. The optimal Gaussian approximation is usually defined as being closest to the posterior in Kullback-Leibler divergence, but it is useful to consider other divergences when the Gaussian assumption is crude, to capture important posterior features for given applications. Our work studies the weighted Fisher divergence, which focuses on gradient differences between the target posterior and its approximation, with the Fisher and score-based divergences as special cases. We make three main contributions. First, we compare approximations for weighted Fisher divergences under mean-field assumptions for Gaussian and non-Gaussian targets with Kullback-Leibler approximations. Second, we go beyond mean-field and consider approximations with sparse precision matrices reflecting posterior conditional independence structure for hierarchical models. Using stochastic gradient descent to enforce sparsity, we develop two approaches to minimize the Fisher and score-based divergences, based on the reparametrization trick and a batch approximation of the objective. Finally, we study the performances of our methods using logistic regression, generalized linear mixed models and stochastic volatility models. 
\end{abstract}

\end{frontmatter}

\section{Introduction}
Bayesian inference is a powerful tool for quantifying uncertainty, but it is demanding to implement for two reasons. First, specifying a full probabilistic model for all unknowns and observables requires careful thought, and components of the model need to be checked against the data. Second, Bayesian computations are difficult, requiring approximation of high-dimensional integrals. For many Bayesian models, exact posterior inference is infeasible, and a variety of numerical methods for summarizing the posterior are used in practice, such as Markov chain Monte Carlo (MCMC) and variational inference (VI). MCMC is often asymptotically unbiased, in that we can estimate posterior quantities as precisely as we wish with a large enough number of iterations, although certain variants (e.g. non-reversible methods) may incur a small bias. While MCMC is often treated as the gold standard for posterior estimation, its computational cost can be prohibitively high for large datasets or complex models \citep{Robert1999, Maclaurin2015}. On the other hand, VI reformulates posterior approximation into an optimization problem by minimizing a divergence between the true posterior and a simpler variational distribution. This enables faster and more scalable inference, leveraging advances in optimization algorithms \citep{Blei2017}. As a result, VI is increasingly popular for its computational efficiency in large-scale problems.

The performance of VI is largely determined by the family of variational approximations chosen, optimization technique, and divergence characterizing discrepancy between the true posterior and variational density. Much of the VI literature has focused on improving expressiveness of the variational family and enhancing optimization methods, often using Kullback-Leibler divergence (KLD) as a measure of approximation quality. To better capture dependence structure among variables, which can be especially strong in hierarchical models, partially factorized VI \citep{Goplerud2025} or structured variational approximations that mimic the true dependency structure \citep{Hoffman2015, Tan2018, Durante2019, Tan2021} can be employed. More recently, flow-based methods which transform an initial simple distribution into more flexible forms through a series of invertible transformations have been introduced \citep{Rezende2014, Dinh2017, Agrawal2024}. These approaches allow VI to capture highly complex posterior distributions, significantly enhancing the flexibility of the inference.

Despite the popularity of KLD, studying alternatives is important, particularly when using simple variational families which may be employed for tractability in high-dimensional problems. These approximations may not be capable of matching the posterior closely, and choosing an appropriate divergence can help to capture the most important features of the posterior for a given application. A family of divergences including KLD as a special case is the R\'{e}nyi's $\alpha$ family \citep{Li2016}, where $\alpha$ can be adjusted to give Hellinger distance ($\alpha =0.5$), $\chi^2$-divergence ($\alpha =2$) and KLD ($\alpha = 1$). While $\alpha$ can help to balance between mode-seeking and mass-covering behavior, the most practical methods for optimizing the variational R\'{e}nyi bound use biased stochastic gradients when $\alpha\neq 1$. Stein divergence has also emerged as a powerful objective for VI. \cite{Ranganath2016} introduced operator variational inference, a minimax approach that optimizes Stein discrepancies by constructing variational objectives based on Stein operators. \cite{Liu2016stein} developed Stein variational gradient descent, which uses kernelized Stein discrepancies to iteratively transform particles toward the posterior. In this article, we explore use of the weighted Fisher divergence in Gaussian VI, focusing on the Fisher and score-based divergences as special cases. The definitions and motivations for studying these divergences are presented below.

\subsection{Weighted Fisher divergence} \label{sec_weighted_FD}
Let $p(y|\theta)$ be the likelihood of observed data $y$, where $\theta \in \mathbb{R}^d$ is an unknown model parameter. Consider Bayesian inference with a prior density $p(\theta)$. In classical variational inference \citep{Ormerod2010, Blei2017}, the true posterior $p(\theta|y) = p(y|\theta) p(\theta)/p(y)$ is approximated with a more tractable density $q(\theta)$ by minimizing the KLD between them, where
\[
\KL(q\|p) = \int q(\theta) \log \frac{q(\theta)}{p(\theta|y)} d\theta.
\]
Let $\E_q$ denote expectation with respect to $q(\theta)$. As $\log p(y) = \KL(q\|p) + \mathcal{L}$, where $\mathcal{L} = \E_q \{  \log p(y, \theta) - \log q(\theta) \}$,  minimizing 
the KLD is equivalent to maximizing an evidence lower bound $\mathcal{L}$ on $\log p(y)$, which does not depend on normalizing constant of the true posterior.

Score matching \citep{Hyvarinen2005} focuses instead on closeness between gradients of the log densities with respect to the variable $\theta$, although the score function refers conventionally to gradient of the log-likelihood with respect to the parameter. A form of such discrepancy is the weighted Fisher divergence \citep{Barp2022}, defined as 
\begin{equation*}
\begin{aligned}
S_M(q\|p) &= \int q(\theta) \bigg \| \nabla_\theta \log \frac{q(\theta)}{p(\theta|y) } \bigg\|_M^2  d\theta,
\end{aligned}
\end{equation*}
where $\| \cdot \|_M$ is the $M$-weighted vector norm defined as $\| z \|_M=\sqrt{z^{\top}M z}$ and $M$ is a positive semi-definite matrix. Like KLD, $S_M(q\|p)$ is asymmetric, non-negative, and vanishes when $q(\theta) = p(\theta|y)$. Let $h(\theta) = p(y|\theta)p(\theta)$. Then $ \nabla_\theta \log p(\theta|y) = \nabla_\theta \log h(\theta)$, which is independent of the unknown normalizing constant $p(y)$.  Similarly, if $q(\theta)$ contains an unknown normalizing constant, this is not required to evaluate the weighted Fisher divergence. Unlike the evidence lower bound, the weighted Fisher divergence provides a direct measure of the distance between the true posterior and variational density.

When $M$ is the identity matrix $I$, $S_{I}(q\|p)$ is known as {\em Fisher divergence} \citep[FD, ][]{Hyvarinen2005}, denoted hereafter as $F(q\|p)$. When $q(\theta)$ is $\N(\mu, \Sigma)$ and $M$ is its covariance matrix $\Sigma$, $S_{\Sigma}(q\|p)$ is known as {\em score-based divergence} (SD) in \citet{Cai2024a}, denoted as $S(q\|p)$ henceforth. \citet{Cai2024a} derived closed-form updates for Gaussian variational parameters in a batch and match (BaM) algorithm based on the SD, and showed that $S(q\|p)$ is affine invariant while $F(q\|p)$ is not. This means that  $S(\Tilde{q} \| \Tilde{p})=S(q\|p)$ if $\Tilde{p}$ and $\Tilde{q}$ denote the densities of $p$ and $q$ respectively after an affine transformation of $\theta$.

In sliced score matching \citep{Song2020}, the scores are projected onto randomly generated vectors $v$ before comparison for dimension reduction, and the weight matrix $M = \E(vv^\top)$. \cite{Liu2022} applied the weighted Fisher divergence in estimating the parameters of truncated densities, whose normalizing constants are intractable, and the weight function is the shortest distance between a data point and the boundary of the domain. The weighted Fisher divergence is also widely used in training score-based generative models \citep{Song2021}, where a forward diffusion and reverse-time process are defined through stochastic differential equations (SDEs). The scores are estimated via neural networks and trained using a time integrated weighted Fisher divergence, where the weight matrix depends on a function of time specified in the SDE \citep{Huang2021, Lu2022}. The above choices of $M$ are not directly applicable or lack intrinsic motivation in our setting, and hence we focus primarily on the FD and SD, as they represent natural and widely studied choices in VI. However, our results in Section \ref{sec:ord_div_Gaussian} also consider general constant weight matrices $M$, besides the FD and SD.

In recent years, there is increasing interest in use of the weighted Fisher divergence in VI. \cite{Huggins2018} showed that the Fisher divergence defined in terms  of the generalized $\ell_p$ norm is an upper bound to the $p$-Wasserstein distance, and its optimization ensures closeness of the variational density to the true posterior in terms of important point estimates and uncertainties. \cite{Yang2019} derived an iteratively reweighted least squares algorithm for minimizing the FD in exponential family based variational approximations, while \citet{Elkhalil2021} employed the factorizable polynomial exponential family as variational approximation in their Fisher autoencoder framework. \cite{Modi2023} developed Gaussian score matching variational inference with closed form updates, by minimizing the KLD between a target and Gaussian variational density subject to a matching score function constraint. For implicit variational families structured hierarchically, \citet{Yu2023} used the FD to reformulate the optimization objective into a minimax problem. \cite{Cai2024b} proposed a variational family built on orthogonal function expansions, and transformed the optimization objective into a minimum eigenvalue problem using the FD.

Our contributions in this article are fourfold. First, we study behavior of the weighted Fisher divergence in mean-field Gaussian VI for Gaussian and non-Gaussian targets, showing its tendency to underestimate the posterior variance more severely than KLD. Second, we develop Gaussian VI for high-dimensional hierarchical models for which posterior conditional independence structure is captured via a sparse precision matrix. Sparsity is enforced by using stochastic gradient descent (SGD), and two distinct approaches are proposed for minimizing the FD and SD. Algorithms based on unbiased gradients computed using the reparameterization trick \citep{Kingma2014} are denoted as FDr and SDr (``r'' for reparameterization trick), while algorithms that rely on a batch approximation of the objective at each iteration are denoted by FDb and SDb (``b" for batch approximation). Third, we study the variance of unbiased gradient estimates computed using the reparametrization trick, and limiting behavior of the batch approximated FD and SD under mean-field. Finally, we present extensive experiments demonstrating that methods based on the reparameterization trick (FDr and SDr) suffer from high variations in gradients and perform poorly relative to baselines such as KLD and BaM. In contrast, methods based on the batch approximation (FDb and SDb) converge more rapidly and scale more efficiently to high-dimensional models.

This article is organized as follows. We study the quality of posterior mean, mode and variance approximations for Gaussian and non-Gaussian targets in Sections \ref{sec:ord_div_Gaussian} and \ref{sec:ord_div_non_Gaussian} respectively, when using the weighted Fisher divergence in VI. Section \ref{sec_sparse_GVA} introduces Gaussian VI for hierarchical models by capturing posterior conditional independence via a sparse precision matrix. Two SGD approaches for minimizing the weighted Fisher divergence are proposed in Sections \ref{sec_sgd_reparametrization trick} and \ref{sec_sgd_batch_approximation}, based respectively on the reparametrization trick and batch approximation. Experimental results are discussed in Section \ref{sec_applications} with applications to logistic regression, generalized linear mixed models (GLMMs) and stochastic volatility models. Section \ref{sec_conclusion} concludes the paper with a discussion.

\section{Ordering of divergences for Gaussian target}  \label{sec:ord_div_Gaussian}
Accurate estimation of the posterior variance is important in VI, as it affects uncertainty quantification in Bayesian inference. Here, we establish an ordering of the weighted Fisher and KL divergences according to the estimated posterior variance when the target $p(\theta|y)$ is $\N(\nu, \Lambda^{-1})$ with a precision matrix $\Lambda$. All divergences considered can recover the true mean $\nu$ and precision matrix $\Lambda$ when the variational family is also Gaussian with a full covariance matrix. However, the computation cost of optimizing a full-rank Gaussian variational approximation can be prohibitive for high-dimensional models. A widely used alternative is the mean-field Gaussian variational approximation, $q(\theta) = \N(\mu, \Sigma)$, with a diagonal covariance matrix $\Sigma$. The mean-field assumption simplifies the optimization but tends to underestimate the true posterior variance under KLD \citep{Blei2017, Tan2018, Giordano2018}. Here, we examine the severity of posterior variance underestimation under the weighted Fisher divergence compared to KLD.

Our results in this section generalize similar results in \cite{Margossian2024} from SD to the general class of weighted Fisher divergences. For KLD under the mean-field assumption, \cite{Margossian2024} showed that the posterior mean can be recovered ($\hat{\mu}=\nu$) and the optimal variance parameter is 
\begin{equation*}
\hat{\Sigma}^{\text{KL}}_{ii} = 1/\Lambda_{ii} \quad \text{for} \quad i=1, \dots, d.
\end{equation*}
Thus, the precision is matched by the variational density, but the variance is underestimated. Lemma \ref{Lem - weighted Fisher div betw Gaussians} presents the weighted Fisher divergence for a general  weight matrix $M$, which is $I_d$ in FD and $\Sigma$ in SD.

\begin{lemma} \label{Lem - weighted Fisher div betw Gaussians}
The $M$-weighted Fisher divergence between a Gaussian target $p(\theta|y) = \N(\theta|\nu, \Lambda^{-1})$ and Gaussian variational approximation $q(\theta) = \N(\theta|\mu, \Sigma)$ is 
\begin{align*}
S_M(q \| p) &= \tr(\Sigma^{-1} M) + \tr(\Lambda M \Lambda \Sigma) - 2\tr(M\Lambda)\\ &\quad + (\mu - \nu)^{\top}\Lambda M \Lambda (\mu - \nu).
\end{align*}
If $\Sigma$ is a diagonal matrix, then 
\begin{align*}
S_M(q \| p) &= \sum_{i=1}^d \{\Sigma_{ii}^{-1} M_{ii} + (\Lambda M \Lambda)_{ii} \Sigma_{ii} \} - 2 \tr(M\Lambda) \\
&\quad + (\mu - \nu)^{\top}\Lambda M \Lambda (\mu - \nu).
\end{align*}
\end{lemma}

From Lemma \ref{Lem - weighted Fisher div betw Gaussians}, $\nabla_\mu S_M(q \| p) = 2\Lambda M \Lambda (\mu - \nu)$. Thus, $\nabla_\mu S_M(q \| p) = 0$ implies $\hat{\mu} = \nu$, and the true posterior mean is recovered for any $M$-weighted Fisher divergence where $M$ is independent of $\mu$. Under the mean-field assumption, at this optimal value $\hat{\mu}$,
\[
S_M(q \| p) = \sum_{i=1}^d \{\Sigma_{ii}^{-1} M_{ii} + (\Lambda M \Lambda)_{ii} \Sigma_{ii} \} - 2 \tr(M\Lambda).
\]
If the weight $M$ is independent of $\Sigma$, then
$\nabla_{\Sigma_{ii}} S_M(q \| p) = (\Lambda M \Lambda)_{ii} - M_{ii} / \Sigma_{ii}^2 = 0$ implies 
\[
\hat{\Sigma}_{ii} = \sqrt{M_{ii} / (\Lambda M \Lambda)_{ii}} \quad \text{for} \quad i=1, \dots, d.
\]
Thus a closed form solution exists for any $M$ independent of $\Sigma$. Moreover, if $M$ is a diagonal matrix, then 
\begin{equation} \label{cov-M}
\hat{\Sigma}_{ii} = \sqrt{\frac{M_{ii}}{\sum_{j=1}^d M_{jj} \Lambda_{ij}^2}} \quad \text{for} \quad i=1, \dots, d.
\end{equation}
When $M_{ii}=1$ $\forall$ $i$, we recover the FD for which the optimal variance parameters are
\begin{equation}\label{cov-Fisher}
\hat{\Sigma}^{\text{F}}_{ii} = \frac{1}{ \sqrt{\sum_{j=1}^d \Lambda_{ij}^2 } }  \quad \text{for} \quad i=1, \dots, d.
\end{equation}

Optimal variational parameters for SD under the mean-field assumption have been presented in \cite{Margossian2024}, and a discussion is included here for completeness. Plugging $M=\Sigma$ in Lemma \ref{Lem - weighted Fisher div betw Gaussians},  
\begin{equation*}
\begin{aligned}
S(q\|p) &=d + \sum_{i=1}^d \sum_{j=1}^d \Sigma_{ii} \Sigma_{jj} \Lambda_{ij}^2 - 2 \sum_{i=1}^d \Sigma_{ii} \Lambda_{ii} ,
\end{aligned}
\end{equation*}
at the optimal value $\hat{\mu}$. Let $\mathbf{s} = (s_1, \dots, s_d)^\top $ such that $s_i = \Sigma_{ii}\Lambda_{ii} \geq 0$, and $H$ be a $d \times d$ symmetric matrix with $H_{ij} = \Lambda_{ij}^2/(\Lambda_{ii}\Lambda_{jj})$. Then $S(q\|p) = d + 2 F(\mathbf{s})$, where 
\begin{equation}\label{cov-score}
F(\mathbf{s}) = \frac{1}{2} \mathbf{s}^{\top} H \mathbf{s} - \mathbf{1}^{\top} \mathbf{s}.
\end{equation}
Thus the optimal  $\hat{\Sigma}^{\text{S}}_{ii}$ that minimizes $S(q \| p)$ can be obtained by solving a non-negative quadratic program (NQP) for $\mathbf{s}$. NQP is the problem of minimizing the quadratic objective function in \eqref{cov-score} subject to the constraint $s_i \geq 0$ $\forall$ $i$. 
Since $\Lambda$ is positive definite, 
\[
x^\top H x = \sum_{i=1}^d \sum_{j=1}^d (x_i/\Lambda_{ii}) \Lambda_{ij}^2 (x_j/\Lambda_{jj}) = y^\top \Lambda y > 0
\]
for any $x = (x_1, \dots, x_d)^\top \in \mathbb{R}^d$ and $y = (x_1/\Lambda_{11}, \dots, x_d/\Lambda_{dd})^\top$. Thus $H$ is symmetric positive definite, which implies that $F(\mathbf{s})$ is bounded below and its optimization is convex. However, there is no analytic solution for the global minimum due to the non-negativity constraints and iterative solutions are required \citep{Sha2002}. The Karush-Kuhn-Tucker (KKT) conditions are first derivative tests that can be used to check whether a solution returned by an iterative solver is indeed a local optimum. For the NQP in \eqref{cov-score}, the KKT conditions state that $\forall$ $i=1, \dots, d$, either (a) $s_i = 0$ and $(Hs)_i > 1$ or (b) $s_i > 0$ and $(Hs)_i = 1$. Note that $\nabla_{ \mathbf{s}} F( \mathbf{s}) = H \mathbf{s} - \bfone$. These conditions correspond to cases where the constraint is active or inactive at the optimum. Case (a) implies $\hat{\Sigma}_{ii}^{\text{S}} = 0$, meaning that the variational density collapses to a point estimate in the $i$th dimension. Note that KLD and FD do not face this issue of ``variational collapse". Case (b) implies 
\begin{equation} \label{Hsi=1}
\begin{aligned}
(H\mathbf{s})_i =& \sum_{j=1}^d H_{ij} s_j = \sum_{j=1}^d \frac{\Lambda_{ij}^2}{\Lambda_{ii} \Lambda_{jj}} \Sigma_{jj} \Lambda_{jj} = 1\\
&\implies  \sum_{j=1}^d \Lambda_{ij}^2 \hat{\Sigma}_{jj}^{\text{S}} = \Lambda_{ii}.
\end{aligned}
\end{equation}

\begin{figure*}[htb!]
\centering
\includegraphics[width=6.597in]{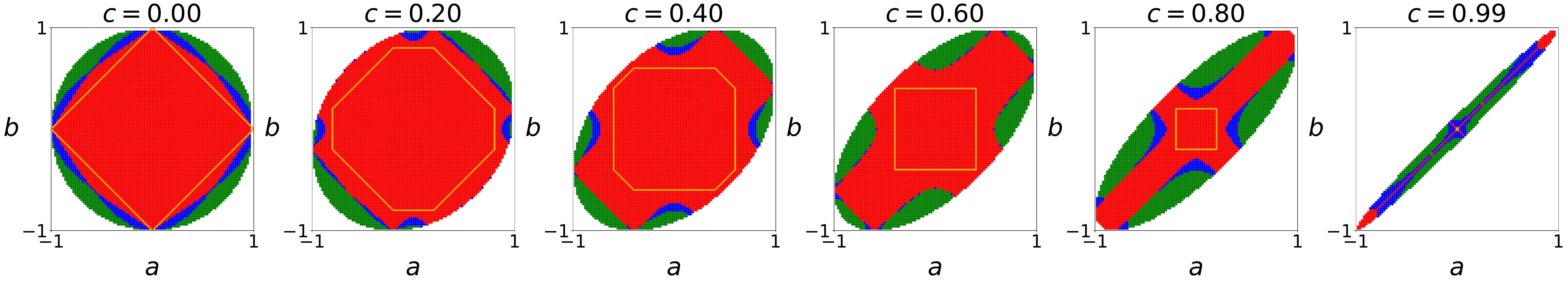}
\caption{Variance parameter comparisons for FD and SD. The red, blue and green regions indicate where $\Sigma_{ii}^{S} \leq \Sigma_{ii}^{F}$ for all cases, only two cases and only one case respectively. The orange-bordered region indicates where $\Lambda$ is diagonally dominant.}\label{com_fisher_score}
\end{figure*}

Next, we investigate how the variance parameters $\{\Sigma_{ii}\}$ obtained by minimizing the weighted Fisher divergence compare to those obtained by minimizing the KLD.

\begin{theorem} \label{thm1}
Suppose the target is a multivariate Gaussian with  precision matrix $\Lambda$, and the variational family is Gaussian with diagonal covariance matrix $\Sigma$. Let $\hat{\Sigma}^{\text{KL}}_{ii}$, $\hat{\Sigma}^{\text{M}}_{ii}$ and $\hat{\Sigma}^{\text{S}}_{ii}$ denote the optimal value of the $i$th diagonal element of $\Sigma$ obtained by minimizing the KL, $M$-weighted Fisher and score-based divergences respectively, where $M$ is a positive definite diagonal matrix independent of $\Sigma$. Then
\begin{equation*}
\begin{aligned}
\hat{\Sigma}^{\text{M}}_{ii} \leq\hat{\Sigma}^{\text{KL}}_{ii} \quad \text{and} \quad 
\hat{\Sigma}^{\text{S}}_{ii} \leq\hat{\Sigma}^{\text{KL}}_{ii} \quad \text{for} \quad i=1, \dots, d,
\end{aligned}
\end{equation*}
and $\exists \, i \in \{1, \dots, d\}$ such that $\hat{\Sigma}_{ii}^{M}<\hat{\Sigma}_{ii}^{\text{KL}}$ and $\hat{\Sigma}_{ii}^{S}<\hat{\Sigma}_{ii}^{\text{KL}}$.
\end{theorem}
\begin{proof}
We first prove $\hat{\Sigma}^{\text{M}}_{ii}\leq\hat{\Sigma}^{\text{KL}}_{ii}$ $\forall \, i$. From \eqref{cov-M}, 
\begin{equation}\label{deri-F}
\begin{aligned}
\hat{\Sigma}^{\text{M}}_{ii} = \sqrt{\frac{M_{ii}}{\sum_{j=1}^d M_{jj} \Lambda_{ij}^2}}
\leq \sqrt{\frac{M_{ii}}{M_{ii} \Lambda_{ii}^2}} = \frac{1}{\Lambda_{ii}} = \hat{\Sigma}_{ii}^{\KL}.
\end{aligned}
\end{equation}
Since $\Lambda$ has at least one nonzero off-diagonal entry, $\exists \, i \in \{1, \dots, d\}$  such that the inequality in \eqref{deri-F} is strict. The proof for $\hat{\Sigma}^{\text{S}}_{ii} \leq \hat{\Sigma}^{\text{KL}}_{ii}$ is given in \cite{Margossian2024} and we include it here for entirety. From the KKT conditions discussed earlier, if case (a) applies, then $\hat{\Sigma}^{\text{S}}_{ii} = 0 < \hat{\Sigma}^{\text{KL}}_{ii}$. Otherwise, case (b) applies and \eqref{Hsi=1} implies that 
\begin{equation} \label{SB<KL}
\Lambda_{ii}^2 \hat{\Sigma}^{\text{S}}_{ii}  \leq  \sum_j \Lambda_{ij}^2 \hat{\Sigma}^{\text{S}}_{jj}  = \Lambda_{ii} \implies \hat{\Sigma}^{\text{S}}_{ii} \leq \frac{1}{\Lambda_{ii}} =\hat{\Sigma}^{\text{KL}}_{ii}.
\end{equation}
To obtain the strict inequality, note that if case (a) applies for at least one $i$, then $\hat{\Sigma}_{ii}^{S}<\hat{\Sigma}_{ii}^{KL}$ for such an $i$. Otherwise, case (b) applies $\forall$ $i$. Since $\Lambda$ has at least one nonzero off-diagonal entry, $\exists \, i \in \{1, \dots, d\}$ such that the first inequality in \eqref{SB<KL} is strict.
\end{proof}

From Theorem \ref{thm1}, both the weighted Fisher and score-based divergences tend to underestimate the posterior variance more severely than KLD under mean-field, but the ordering between FD and SD is more nuanced. If case (a) of the KKT conditions apply, then $ \hat{\Sigma}^{\text{S}}_{ii} = 0 < \hat{\Sigma}^{\text{F}}_{ii}$. If case (b) applies, then from  \eqref{cov-Fisher} and \eqref{Hsi=1},
\[
\begin{aligned}
\Lambda_{ii}^2 \hat{\Sigma}^{\text{S}}_{ii} &\leq \sum_{j=1}^d \Lambda_{ij}^2 \hat{\Sigma}^{\text{S}}_{jj} = \Lambda_{ii} \hat{\Sigma}^{\text{F}}_{ii} \sqrt{\sum\nolimits_{j=1}^d \Lambda_{ij}^2 } \\
&\implies \hat{\Sigma}^{\text{S}}_{ii} \leq \hat{\Sigma}^{\text{F}}_{ii} \frac{\sqrt{\sum_{j=1}^d \Lambda_{ij}^2 }}{\Lambda_{ii} }. 
\end{aligned}
\]
Moreover, if $\Lambda$ is a {\em diagonally dominant} matrix such that $\sum_{j\neq i} |\Lambda_{ij}| \leq |\Lambda_{ii}| \; \forall\; i$, then
\begin{equation*}
\begin{aligned}
\hat{\Sigma}^{\text{S}}_{ii} &\leq \hat{\Sigma}^{\text{F}}_{ii} \frac{\sqrt{\Lambda_{ii}^2 + \sum_{j\neq i} \Lambda_{ij}^2 }}{\Lambda_{ii} } 
\leq \hat{\Sigma}^{\text{F}}_{ii} \frac{\sqrt{\Lambda_{ii}^2 + (\sum_{j\neq i} |\Lambda_{ij}|)^2 }}{\Lambda_{ii} }\\ 
&\leq \hat{\Sigma}^{\text{F}}_{ii} \frac{\sqrt{\Lambda_{ii}^2 + \Lambda_{ii}^2 }}{\Lambda_{ii} } 
= \sqrt{2} \hat{\Sigma}^{\text{F}}_{ii}.
\end{aligned}
\end{equation*}  
Thus the ratio of $\hat{\Sigma}^{\text{S}}_{ii} / \hat{\Sigma}^{\text{F}}_{ii}$ is bounded by $\sqrt{2}$ $\forall\; i$ if $\Lambda$ is diagonally dominant.

For a more concrete comparison of posterior variance approximation based on FD and SD, consider a three-dimensional Gaussian target with precision matrix,
\begin{equation*}
\Lambda = 
\begin{bmatrix}
1 & a & b \\
a & 1 & c \\
b & c & 1
\end{bmatrix}.
\end{equation*}
For FD, $\hat{\Sigma}_{ii}^{\text{F}}$ can be obtained from \eqref{cov-Fisher}, while the splitting conic solver \citep[SCS,][]{Odonoghue2016} in the CVXPY Python package is used to solve the NQP in \eqref{cov-score} for SD. SCS is designed for convex optimization problems characterized by conic constraints, such as non-negativity. It decomposes the optimization into subproblems solved iteratively by operator-splitting techniques.

Fig \ref{com_fisher_score} illustrates how variance parameters obtained from FD and SD compare by varying the conditional correlations $a$, $b$ and $c$. Each plot represents a value of $c$. The colored regions represent configurations for which $\Lambda$ is positive definite, and there is no region where $\hat{\Sigma}_{ii}^{S} > \hat{\Sigma}_{ii}^{F}$ $\forall$ $i$. Variance estimates based on SD are more likely to exceed those based on FD when $a$, $b$ or $c$ has a large magnitude. In this example, $\hat{\Sigma}^{\text{S}}_{ii} /\hat{\Sigma}^{\text{F}}_{ii}$ can be bounded more tightly by 1 instead of only $\sqrt{2}$ when $\Lambda$ is diagonally dominant.

\section{Ordering of divergences for non-Gaussian target} \label{sec:ord_div_non_Gaussian}

\begin{figure*}[tb!]
\centering
\includegraphics[width=6.597in]{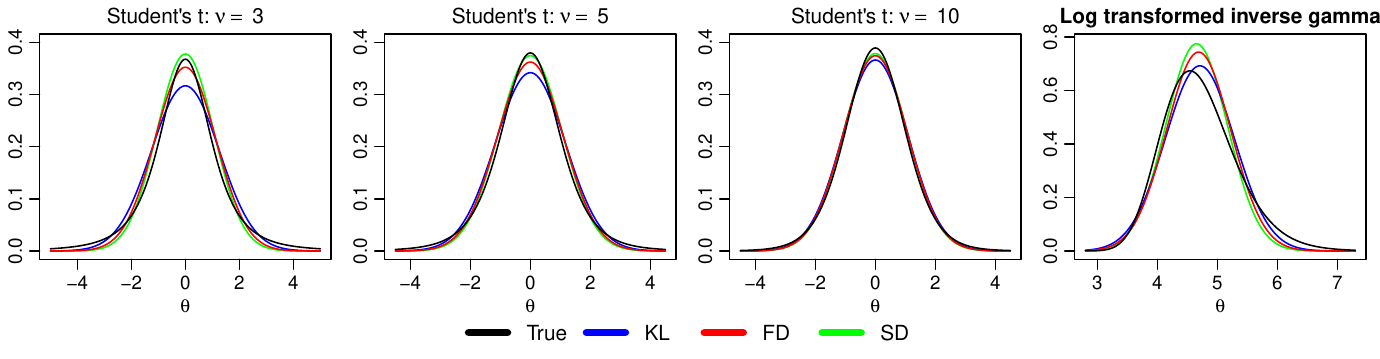}
\caption{Gaussian variational approximations for Student's $t$ and log transformed inverse gamma.}\label{studentt_invgam}
\end{figure*}

Next, we study the ordering of FD, SD and KLD in posterior mean, mode and variance estimation when the target distribution is non-Gaussian while the variational approximation is Gaussian. Theoretical analysis in this setting is complex and numerical methods are often required. We show that the true posterior mean is recoverable across all divergences for the multivariate Student's $t$, and an ordering of the mean, mode and variance estimation is established for the log transformed inverse gamma density. Otherwise, empirical comparisons are made by considering $p(\theta \mid y)$ as some univariate non-Gaussian density, and the variational density $q(\theta)$ as $\N(\mu, \sigma^2)$. Results in this section indicate that KLD estimates the mean most accurately and has the highest accuracy (as defined below) when the target is skewed, but the lowest accuracy when the target is symmetric and has heavy tails. SD captures the mode most accurately if the target density is skewed, but underestimates the posterior variance most severely.

Let $\mu_{*}$, $m_{*}$ and $\sigma^2_{*}$ denote the mean, mode and variance of the target density. To evaluate the performance of different divergences, we use the normalized absolute difference in mean and mode: $|\mu - \mu_{*}|/\sigma_{*}$ and $|\mu - m_{*}|/\sigma_{*}$, variance ratio: $\sigma^2 / \sigma^2_{*}$, and integrated absolute error : $\text{IAE}(q) = \int |q(\theta) - p(\theta|y)| d\theta \in [0, 2]$, which is invariant under monotone transformations of $\theta$. We define $\text{accuracy}(q) = 1 - \text{IAE}(q) / 2$, where a higher value indicates a more accurate approximation of the target. In the examples below, the VI objective function is tractable or can be computed numerically, and variational approximations are optimized using L-BFGS via {\tt optim} in {\tt R}.

\subsection{Student's $t$}
First, consider the multivariate Student's $t$ distribution, $t_\nu(m,S)$ for $\theta \in \mathbb{R}^d$ as the target, where
\[
p(\theta \mid y)=\frac{\Gamma\left(\frac{\nu+d}{2}\right) \left(1 + (\theta-m)^\top S^{-1}(\theta-m) / \nu \right)^{-\frac{\nu+d}{2}} }
{\Gamma\left(\frac{\nu}{2}\right)\,(\nu\pi)^{d/2}\,|S|^{1/2}}.
\]
The Student's $t$ is symmetric but has heavier tails than the Gaussian, and $\nu$, $m$ and $S$ denote the degrees of freedom, location parameter and scale matrix respectively. Theorem \ref{thm-t} shows that the true posterior mean or mode of the Student's $t$ is recoverable by a Gaussian variational approximation under all three divergences.

\begin{theorem}\label{thm-t}
Let $q(\theta)=\N(\mu,\Sigma)$ and $p(\theta \mid y)=t_\nu(m,S)$. Then $\mu = m$ is a stationary point of   the KLD, FD and SD between $q(\theta)$ and $p(\theta \mid y)$. 
\end{theorem}

Next, we consider the univariate Student's $t$ as the target to compare the performances of different divergences in capturing the variance. For $\theta \sim t(\nu)$, $p(\theta|y) = \left( 1+ \theta^2/\nu \right)^{-(\nu+1)/2} \Gamma(\frac{\nu+1}{2})/(\sqrt{\pi\nu} \, \Gamma(\frac{\nu}{2}))$,
where $\nu \in \{3, 5, 10\}$ is the degrees of freedom. All divergences successfully capture mode of the target at 0, verifying Theorem \ref{thm-t}. From Table \ref{table-stu}, SD exhibits the most severe posterior variance underestimation, followed by FD and then KLD. In terms of the IAE, both FD and SD yield approximations with higher accuracy than KLD. Fig \ref{studentt_invgam} (first 3 plots) compares optimal variational densities with the target, and showing that KLD tends to underestimate the mass around the mode more severely than FD and SD.

\begin{table}[tb!]
\centering
\begin{tabular}{c|ccc|ccc}
\hline
& \multicolumn{3}{c}{$\sigma^2 / \sigma^2_{*}$} & \multicolumn{3}{|c}{accuracy} \\ 
$\nu$ & KLD & FD & SD & KLD & FD & SD 
\\ \hline
3 & \textbf{0.529} & 0.428 & 0.372 & 92.18 & \textbf{93.66} & 92.62 
\\ 
5 & \textbf{0.818} & 0.728 & 0.681 & 94.72 & 95.82 & \textbf{95.97} 
\\ 
10 & \textbf{0.950} & 0.909 & 0.889 & 97.01 & 97.55 & \textbf{97.73} 
\\  \hline
\end{tabular}
\caption{Results for Student's $t$ (best values highlighted in bold).}\label{table-stu}
\end{table}

\subsection{Log transformed inverse gamma}
Consider the normal sample model in \citet{Tan2024b}, where $y_i \mid \theta \sim \N(0, \exp(\theta))$ for $i=1, \dots, n$, with prior, $\exp(\theta) \sim \IG (a_0, b_0)$, and $a_0 = b_0 = 0.01$. The true posterior of $\exp(-\theta)$ is $\text{G}(a_1, b_1)$, where $a_1 = a_0 + n/2$ and $b_1 = b_0 + \sum_{i=1}^n y_i^2/2$. The true posterior mode, mean and variance of $\theta$ are $m_* = \log (b_1/a_1)$, $\mu_* = \log b_1 - \psi(a_1)$ and $\sigma_*^2 = \psi_1(a_1)$, where $\psi(\cdot)$ and $\psi_1(\cdot)$ denote the digamma and trigamma functions respectively.

This is a rare example where the FD, SD and evidence lower bound for the KLD can be derived in closed form. Moreover, the optimal variational parameters for all three divergences are available analytically, as given in Table \ref{table_ig_VA_par}. Note that $W_0(\cdot)$ denotes the principal branch of the Lambert W function \citep{Corless1996}. Theorem~\ref{thm-loggamma} shows that SD underestimates the variance most severely, followed by FD and then KLD. Moreover, SD yields the best estimate of the mode, while KLD estimates the mean most accurately, with FD lying in between.

\begin{table}[tb!]
\centering
\begin{tabular}{ccc}
\hline
KLD & $\hat{\sigma}_{\KL}^2 = \frac{1}{a_1}$ 
& $\hat{\mu}_{\KL} = \log \frac{b_1}{a_1} + \frac{1}{2a_1}$ 
\\[2mm]
FD & $\hat{\sigma}^2_{\F}=-2W_0\left(-\frac{1}{2(a_1+1)}\right)$ 
& $\hat{\mu}_{\F} =\log \frac{ b_1}{a_1 + 1} + \frac{3\hat{\sigma}^2_{\F}}{2}$
\\ [2mm]
SD &  
$\hat{\sigma}^2_{\S}=1 - W_0\left(\frac{\e a_1^2}{(a_1+1)^2}\right)$  & 
$\hat{\mu}_{\S} =\log \frac{ b_1}{a_1 + 1} + \frac{3\hat{\sigma}^2_{\S}}{2}$\\ 
\hline
\end{tabular}
\caption{Optimal variational parameters for log transformed inverse gamma.}\label{table_ig_VA_par}
\end{table}

\begin{theorem}\label{thm-loggamma}
Let $\hat{\mu}_{\KL}$, $\hat{\mu}_{\F}$, $\hat{\mu}_{\S}$, $\hat{\sigma}^2_{\KL}$, $\hat{\sigma}^2_{\F}$ and $\hat{\sigma}^2_{\S}$ denote the optimal mean and variance parameters that minimize the KLD, FD and SD respectively, when the target is a log transformed inverse gamma density and the variational approximation is Gaussian. Then 
\begin{gather*}
\hat{\sigma}^2_{\S} < \hat{\sigma}^2_{\F} < \hat{\sigma}^2_{\KL}< \sigma_*^2, 
\\
m_* < \hat\mu_{\S} < \hat\mu_{\F} < \hat\mu_{\KL} < \mu_*,
\end{gather*}
where $\mu_{*}$, $m_{*}$ and $\sigma^2_{*}$ denote the mean, mode and variance of the target.
\end{theorem}

To verify Theorem \ref{thm-loggamma}, we simulate $n=6$ observations by setting $\exp(\theta) = 225$. Table \ref{table_ig} shows that the ordering in mean, mode and variance estimation is consistent with Theorem \ref{thm-loggamma}. Overall, KLD has the highest accuracy followed by FD and then SD. A visualization is given in Fig \ref{studentt_invgam} (last plot).
\begin{table}[htb!]
\centering
\begin{tabular}{cccc}
\hline
& KLD & FD & SD \\ \hline
$|\mu - \mu_{*}|/\sigma_{*}$ & \textbf{0.015} & 0.048 & 0.102 \\ 
$|\mu - m_{*}|/\sigma_{*}$ & 0.265 & 0.231 & \textbf{0.177} \\
$\sigma^2 / \sigma^2_{*}$ & \textbf{0.845} & 0.732 & 0.674 \\
accuracy & \textbf{92.67} & 91.91 & 91.53 \\ 
\hline
\end{tabular}
\caption{Results for log transformed inverse gamma (best values highlighted in bold).}\label{table_ig}
\end{table}

\begin{table*}[htb!]
\centering
\begin{tabular}{c|ccc|ccc|ccc|ccc}
\hline
& \multicolumn{3}{c}{$|\mu - \mu_{*}|/\sigma_{*}$} & \multicolumn{3}{|c|}{$|\mu - m_{*}|/\sigma_{*}$} & \multicolumn{3}{c|}{$\sigma^2 / \sigma^2_{*}$} & \multicolumn{3}{c}{accuracy} \\ 
$(t, \lambda)$ & KLD & FD & SD & KLD & FD & SD & KLD & FD & SD & KLD & FD & SD \\ 
\hline
(1, 1) & \textbf{0.001} & 0.003 & 0.004 & 0.070 & 0.067 & \textbf{0.066} & \textbf{0.992} & 0.984 & 0.979 & 98.27 & 98.31 & \textbf{98.32}\\ 
(1, 2) & \textbf{0.006} & 0.031 & 0.064 & 0.255 & 0.230 & \textbf{0.197}  & \textbf{0.919} & 0.851 & 0.803 & 93.77 & \textbf{93.81} & 93.79 \\   
(1, 5) & \textbf{0.004} & 0.251 & 0.586 & 0.657 & 0.912 & \textbf{0.075} & \textbf{0.677} & 0.642 & 0.248 & \textbf{83.93} & 76.44 & 68.50 \\ 
(5, 1) & \textbf{0.004} & 0.251 & 0.586 & 0.657 & 0.912 & \textbf{0.075} & \textbf{0.677} & 0.642 & 0.248 & \textbf{83.92} & 76.42 & 68.49 \\ 
(5, 2) & \textbf{0.024} & 1.285 & 1.011 & 0.939 & 2.200 & \textbf{0.097} & 0.504 & \textbf{0.757} & 0.054 & \textbf{76.50} & 45.38 & 37.06 \\ 
(5, 5) & \textbf{0.077} & 1.819 & 1.209 & 1.201 & 2.942 & \textbf{0.086} & 0.352 & \textbf{0.644} & 0.008 & \textbf{68.00} & 30.35 & 16.35 \\ 
\hline
\end{tabular}
\caption{Results for skew normal (best values highlighted in bold).}\label{table_skew}
\end{table*}

\begin{figure*}[htb!]
\centering
\includegraphics[width=6.597in]{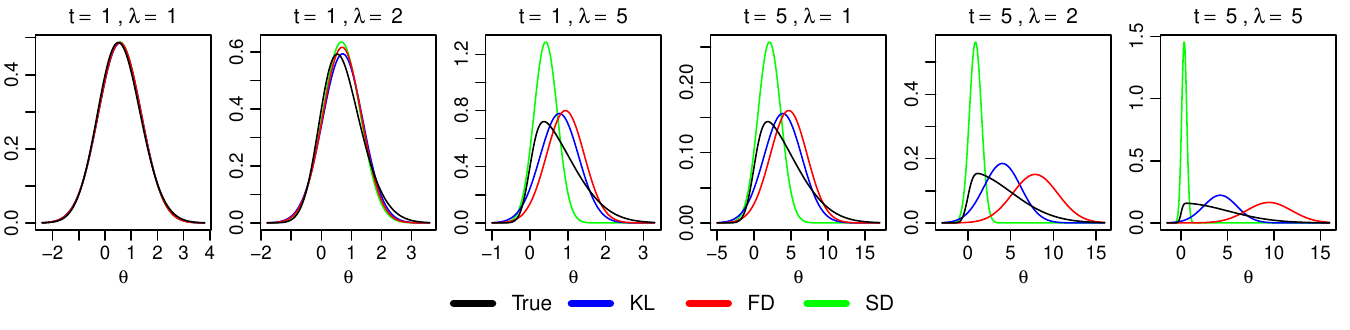}
\caption{Gaussian variational approximations for skew normal.}\label{skewnormal}
\end{figure*}

\subsection{Skew normal}
Finally, let the target be a univariate skew normal, $\theta \sim \SN(m, t, \lambda)$. Then $p(\theta|y) = 2 \, \phi(\theta|m, t^2) \, \Phi \{\lambda(\theta - m) \}$, where $m \in \mathbb{R}$, $t >0$ and $\lambda \in \mathbb{R}$ are the location, scale and skewness parameters respectively, and $\Phi(\cdot)$ is cumulative distribution function of the standard normal.

We set $m = 0$ and let $t\in \{1,5\}$ and $\lambda\in \{1,2,5\}$. From Table \ref{table_skew}, KLD estimates the mean most accurately, while SD captures the mode most accurately.  For FD, estimation of the mode is very poor when both scale and skewness are large. SD underestimates the variance most severely, with the variance estimate collapsing to zero as $t$ and $\lambda$ increase. KLD has higher accuracy than both FD and SD as skewness and scale increase. From Fig \ref{skewnormal}, SD is good at identifying the mode, whereas FD and KLD estimate the variance more accurately. We note that multiple local minimums were detected for SD in this context.

\section{Sparse Gaussian variational approximations} \label{sec_sparse_GVA}
Next, we consider Gaussian VI for hierarchical models and compare performances of the FD, SD and KLD. Given observed data $y = (y_1,\dots, y_n)^\top$, the variable $\theta = (\theta_L^\top,\theta_G^\top)^\top \in \mathbb{R}^d$ of a two-tier hierarchical model can be partitioned into a {\em global} variable $\theta_G$ that is shared among all observations and {\em local} variables $\theta_L =(b_1^\top,\dots, b_n^\top)^\top$, where $b_i$ is specific to the observation $y_i$ for $i=1, \dots, n$. Let the joint likelihood of the model be  
\begin{multline} \label{joint-model}
p(y, \theta) = p(\theta_G)p(b_1,\dots, b_\ell|\theta_G) \times 
\\
\prod_{k=\ell+1}^n p(b_k|b_{k-1},\dots, b_{k-\ell},\theta_G)
\prod_{i=1}^n p(y_i|\theta_G,b_i),
\end{multline}
where $\{y_i\}$ are conditionally independent given $\theta$, and $\{b_i\}$ follow an $\ell$th order Markov model given $\theta_G$. Thus $\{b_i\}$ are conditionally independent of each other a posteriori given the $\ell$ neighboring values and $\theta_G$. In a random effects model, $\{b_i\}$ are the random effects with $\ell = 0$, while for a state space model, $\{b_i\}$ are the latent states with $\ell =1$.

Let $q_\lambda(\theta)$ be $\N(\mu,\Sigma)$, a Gaussian variational approximation of the posterior with mean $\mu \in \mathbb{R}^d$ and covariance matrix $\Sigma\in \mathbb{R}^{d \times d}$. Consider a Cholesky decomposition of the precision matrix $\Omega = \Sigma^{-1} = TT^\top$ where $T$ is a lower triangular matrix, and denote the variational parameters as $\lambda=(\mu^\top,\text{vech}(T)^\top)^\top$, where $\text{vech}(\cdot)$ is an operator that stacks lower triangular elements of a matrix columnwise from left to right into a vector.

In a multivariate Gaussian, conditional independence implies sparse structure in the precision matrix, with $\Omega_{ij}=0$ if $\theta_i$ and $\theta_j$ are conditionally independent given the remaining variables. By Proposition 1 of \cite{Rothman2010}, the Cholesky factor $T$ has the same row-banded structure as $\Omega$. Suppose $T$ is block partitioned according to  $(b_1^\top,\dots, b_n^\top, \theta_G^\top)^\top$, with corresponding blocks $T_{ij}$ for $i,j=1,\dots, n+1$. First, $T_{ij}=0$ if $j>i$ as $T$ is lower-triangular. If we further constrain $T_{ij}=0$ for $1\leq j\leq i-l$, then $\Omega$ reflects the conditional independence structure of the joint likelihood in (7). For instance, for GLMMs with $\ell=0$, $T$ has the sparse block structure,
\begin{equation*}
T =
\begin{bmatrix}
T_{11} & 0 & \ldots & \ldots & 0 \\
0 & T_{22} & \ldots & \ldots & 0 \\
\vdots & \vdots & \ddots & \vdots & \vdots \\
0 & 0 & \dots & T_{n n} & 0 \\
T_{G1} & T_{G2} & \ldots & T_{Gn} & T_{GG} \\
\end{bmatrix}.
\end{equation*}
When $\theta$ is high-dimensional, exploiting the conditional independence structure in the model is essential to making Gaussian VI feasible, as the number of parameters to be optimized in $T$ grows quadratically with $n$. However, after imposing sparsity on $T$, the number of parameters only grows linearly with $n$. Predetermined sparsity in $T$ can be enforced in SGD by updating only the elements in $T$ that are not constrained to zero.

\section{SGD based on reparameterization trick} \label{sec_sgd_reparametrization trick}
In this section, we develop SGD algorithms to minimize the FD and SD based on unbiased gradient estimates derived using the reparameterization trick \citep{Kingma2014}, named FDr and SDr respectively. In this approach, the gradients involve Hessians of the log joint density, which are sparse matrices that can be computed efficiently. However, we demonstrate later in Section \ref{sec variance gradients} that these gradients have much higher variance than corresponding algorithms based on KLD, resulting in slow convergence and suboptimal variational approximations. An alternative approach is thus proposed in Section \ref{sec_sgd_batch_approximation}.

Let $g(\lambda, \theta) = \nabla_\theta \log h(\theta)-\nabla_\theta \log q_\lambda(\theta)$ where $h(\theta)=p(\theta)p(y|\theta)$ is as defined previously. The FD and SD between $q_\lambda(\theta)$ and $p(\theta|y)$ can be written as
\begin{align*}
F(\lambda) &= \E_q[g(\lambda, \theta)^\top g(\lambda, \theta)], \\
S(\lambda) &=\E_q[f(\lambda, \theta)^\top f(\lambda, \theta)],
\end{align*}
respectively, where $f(\lambda, \theta)=T^{-1}g(\lambda, \theta)$. The gradients for minimizing the FD and SD via SGD can be derived by applying the reparametrization trick. Instead of simulating $\theta$ directly from $q_\lambda(\theta)$, we generate $z \sim \N(0, I_d)$ and compute $\theta=\mu+T^{-\top}z$. Thus
\begin{align*}
F(\lambda) &= \E_\phi \left\{ g(\lambda, \mu+T^{-\top}z)^\top  g(\lambda, \mu+T^{-\top}z) \right\}, \\
S(\lambda) &= \E_\phi \left\{ f(\lambda, \mu+T^{-\top}z)^\top  f(\lambda, \mu+T^{-\top}z) \right\}, 
\end{align*}
where $\E_\phi(\cdot)$ denotes expectation with respect to $\phi(z)$, the density function of $\N(0, I_d)$. Note that
\[
\begin{aligned}
g(\lambda, \theta) &= \nabla_\theta \log h(\theta) + TT^\top (\theta-\mu), \\
f(\lambda, \theta) &= T^{-1}\nabla_\theta \log h(\theta) + T^\top (\theta-\mu),
\end{aligned}
\]
both of which depends on $\lambda$ directly as well as through $\theta$. Applying the chain rule, 
\begin{align*}
\nabla_\mu F(\lambda) &= 2 \E_\phi  \{\nabla_\theta^2 \log h(\theta) g(\lambda, \theta)\}, \\
\nabla_\mu S(\lambda) & = 2 \E_\phi  \{\nabla_\theta^2 \log h(\theta) \Sigma g(\lambda, \theta)\}, \\
\nabla_{\vech(T)} F(\lambda) 
& =  2  \E_\phi \vech \left\{ g(\lambda, \theta) z^\top \right. \\
&\left. \quad- T^{-\top} z g(\lambda, \theta)^\top \nabla_\theta^2 \log h(\theta)  T^{-\top} \right\}, \\
\nabla_{\vech(T)} S(\lambda) 
& = - 2  \E_\phi \vech \left\{ \Sigma g(\lambda, \theta) \nabla_\theta \log h(\theta)^{\top}T^{-\top} \right. \\
&\left. \quad + T^{-\top} z g(\lambda, \theta)^\top \Sigma \nabla_\theta^2 \log h(\theta)  T^{-\top} \right\}.
\end{align*}

Unbiased gradient estimates can be obtained by sampling from $\phi(z)$.  All gradient computations can be done efficiently even in high-dimensions, as they only involve sparse matrix multiplications and solutions of sparse triangular linear systems. The Hessian $\nabla_\theta^2 \log h(\theta)$ has the same block sparse structure as $\Omega$, as $b_i$ and $b_j$ only occur in the same factor of (\ref{joint-model}) if $b_j$ is one of the $\ell$ neighboring values of $b_i$. For $\nabla_{\vech(T)} F(\lambda)$ and $\nabla_{\vech(T)} S(\lambda)$, we only need to compute elements corresponding to those in $\vech(T)$ that are not fixed by sparsity. For instance, to compute the second term in $\nabla_{\vech(T)} F(\lambda)$, we just find $u = T^{-\top} z$ and $v = T^{-1} \nabla_\theta^2 h(\theta) g(\lambda, \theta)$, and then form $u_iv_j$ for nonzero elements $(i,j)$ of $T$.

The update for $T$ in SGD does not ensure that its diagonal entries remain positive. Hence, we introduce $T^*$ such that $T^*_{ii}=\log (T_{ii})$ for $i=1,\dots,n,$ and $T^*_{ij}=T_{ij}\;for\;i\neq j$. Let $J$ be a $d\times d$ matrix with diagonal equal to $\diag(T)$ and all off-diagonal entries being 1, and $D$ be a diagonal matrix with the diagonal given by $\vech(J)$. Then $\nabla_{\vech (T^*)}F(\lambda)=D \nabla_{\vech (T)}F(\lambda)$ and updates for $T^*$ are unconstrained.

\begin{algorithm}[tb!]
\caption{SGD based on reparametrization trick}
\label{Algs}
\begin{algorithmic}[1]
\Input Initial $\mu\in\mathbb{R}^d$, initial $T^*\in\mathbb{R}^{d\times d}$, stepsize schedule $\{\rho_t\}$
\Function{Map}{$T^*$}
    \State Construct $T$: $T_{ii}\gets \exp(T^*_{ii})$, $T_{ij}\gets T^*_{ij}$ for $i\neq j$
    \State \Return $T$
\EndFunction
\Function{BuildD}{$T$}
    \State $J \gets \mathbf{1}\mathbf{1}^\top$, set $\mathrm{diag}(J)\gets \mathrm{diag}(T)$
    \State $D \gets \mathrm{diag}(\vech(J))$
    \State \Return $D$
\EndFunction
\State $t \gets 1$
\While{not converged}
    \State $T \gets \Call{Map}{T^*}$,\quad $D \gets \Call{BuildD}{T}$
    \State Sample $z \sim \mathcal{N}(0,I_d)$,\quad$u \gets T^{-\top} z$,\quad $\theta \gets \mu + u$
    \State $g \gets \nabla_\theta \log h(\theta) + Tz$
    \If{\textsc{KLD}}
        \State $\mu \gets \mu + \rho_t\, g$,\quad $v \gets T^{-1} g$,\quad $g_T \gets -\,u\,v^\top$
        \State $T^* \gets T^* + \rho_t\, D\, g_T$
    \ElsIf{\textsc{FD}r \textbf{ or } \textsc{SD}r}
        \If{\textsc{SD}r}
            \State $g \gets T^{-1} g$,\quad $z \gets z - g$,\quad $g \gets T^{-\top} g$
        \EndIf
        \State $w \gets \nabla_\theta^2 \log h(\theta)\, g$,\quad $v \gets T^{-1} w$,\quad$\mu \gets \mu - 2\,\rho_t\, w$
        \State $g_T \gets g\,z^\top - u\,v^\top$,\quad $T^* \gets T^* - 2\,\rho_t\, D\, g_T$
    \EndIf
    \State $t \gets t+1$
\EndWhile
\end{algorithmic}
\end{algorithm}

Algorithm \ref{Algs} outlines the SGD algorithms for updating $(\mu, T)$ by minimizing the FD, SD or KLD \citep[derived in][]{Tan2018}. The stepsize $\rho_t$ is computed elementwise adaptively using Adadelta \citep{Zeiler2012}. All three algorithms compute $g(\lambda, \theta)$, but the KLD based algorithm uses $g(\lambda, \theta)$ to update $\mu$ and $T$ directly, while FDr and SDr premultiply $g(\lambda, \theta)$ by the Hessian $ \nabla_\theta^2 \log h(\theta)$ and are hence more  computationally intensive.

\subsection{Analysis of variance of gradient estimates} \label{sec variance gradients}
Here, we study the variance of unbiased gradient estimates derived by applying the reparametrization trick on the KLD, FD and SD. The variance of these gradients plays a crucial role in stability of the optimization, as large variance can cause a zigzag phenomenon, making convergence difficult. For a closed form analysis, we assume the target $p(\theta|y)$ is $\N(\nu, \Lambda^{-1})$. Then $\nabla_{\theta}\log h(\theta)=-\Lambda(\theta-\nu)$ and $\nabla_{\theta}^2\log h(\theta)=-\Lambda$.

\begin{figure*}[htb!]
\centering
\includegraphics[width=6.597in]{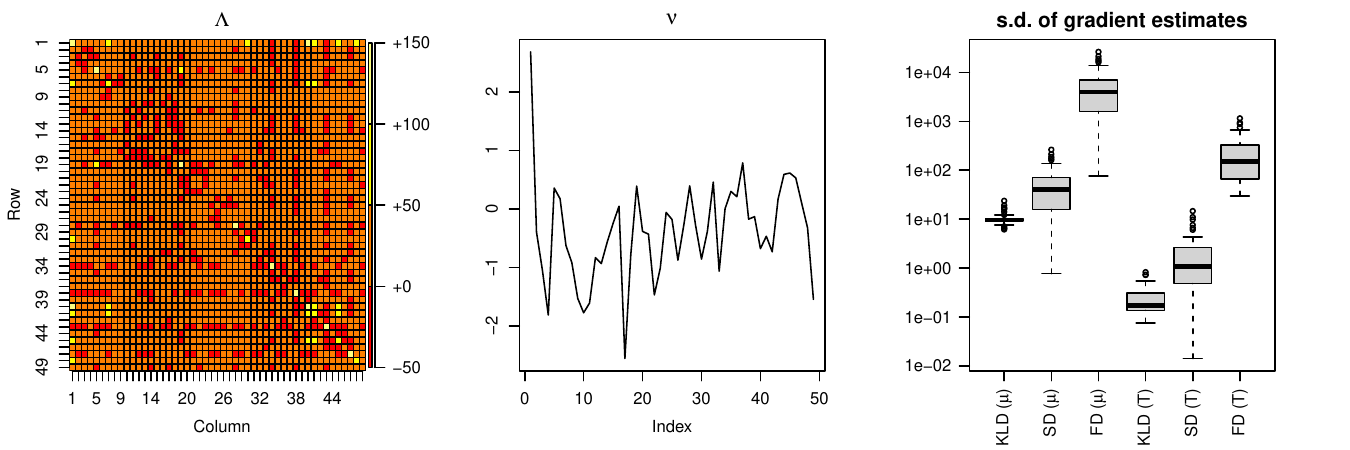}
\caption{First two plots show the true precision matrix $\Lambda$ and mean $\nu$, and the third plot contains boxplots of the standard deviation (s.d.) in gradient estimates for $\{\mu_i\}$ and $\{T_{ii}\}$.}
\label{grad_compare}
\end{figure*}

From Algorithm \ref{Algs}, gradient estimates with respect to $\mu$ for the KLD, FD and SD based on a single sample are 
\[
g_{\mu}^{\KL} = A z - \Lambda(\mu-\nu), \quad
g_{\mu}^{\text{F}} = 2\Lambda g_{\mu}^{\KL},\quad
g_{\mu}^{\text{S}} = 2 \Lambda \Sigma g_{\mu}^{\KL},
\]
where $A = T-\Lambda T^{-\top}$. The stochasticity stems from drawing $z \sim \N(0, I_d)$ and $\Var(g_{\mu}^{\KL}) = A A^\top$, while
\[
\begin{aligned}
\Var(g_{\mu}^{\text{F}})= 4\Lambda \Var(g_{\mu}^{\KL}) \Lambda, \;\;
\Var(g_{\mu}^{\text{S}}) = 4\Lambda \Sigma \Var(g_{\mu}^{\KL}) \Sigma\Lambda.
\end{aligned}
\]
Similarly, from Algorithm \ref{Algs}, the gradient estimates with respect to $T$ are
\begin{align*}
g_T^{\KL} &= T^{-\top} z (\mu -\nu)^\top \Lambda T^{-\top} -  T^{-\top} zz^\top A^\top T^{-\top} ,
\\
g_T^{\text{F}} &= 2\{ \Lambda  (\mu -\nu) z^\top + T^{-\top} z (\mu -\nu)^\top \Lambda^2 T^{-\top}\\
&\quad - A zz^\top - T^{-\top} zz^\top A^\top \Lambda T^{-\top} \},
\\
g_T^{\text{S}} 
&= 2[  \Sigma \Lambda (\mu -\nu) \{z^\top T^{-1} +  (\mu -\nu) ^\top\}  \Lambda T^{-\top} \\
&\quad-\Sigma  A \{z z^\top T^{-1}  + z   (\mu -\nu)^\top\} \Lambda T^{-\top}  \\
& \quad + T^{-\top} \{z  (\mu -\nu)^\top \Lambda - zz^\top A^\top \} \Sigma \Lambda T^{-\top} ].
\end{align*}

The variance of these estimates depends on the mean and precision of the true target, which is fixed, and that of the variational approximation, which changes during SGD. Suppose $\Lambda$ and $T$ are both diagonal matrices, then
\begin{align*}
\Var(g_{\mu_i}^{\KL}) &= T_{ii}^2 - 2 \Lambda_{ii} + \Lambda_{ii}^2 T_{ii}^{-2}, \\
\Var(g_{\mu_i}^{\text{F}})&= 4\Lambda_{ii}^2 \Var(g_{\mu_i}^{\KL}), \\
\Var(g_{\mu_i}^{\text{S}}) &= (4\Lambda_{ii}^2/T_{ii}^4) \,\Var(g_{\mu_i}^{\KL}), \\
\Var(g_{T_{ii}}^{\KL}) &= T_{ii}^{-4} \left\{ \Lambda_{ii}^2 (\mu_i - \nu_i)^2 + 2 \left( T_{ii} - \Lambda_{ii}/T_{ii} \right)^2 \right\}, \\
\Var(g_{T_{ii}}^{\text{F}}) &= 4 (T_{ii}^2 + \Lambda_{ii})^2\Var(g_{T_{ii}}^{\KL}). \\
\Var(g_{T_{ii}}^{\text{S}})
&= 4\Lambda_{ii}^2 T_{ii}^{-8}  \left\{ \left( 3\Lambda_{ii}  - T_{ii}^2 \right)^2  (\mu_i - \nu_i)^2\right. \\
&\quad + \left.8 \left( T_{ii} - \Lambda_{ii}/T_{ii} \right)^2 \right\}.
\end{align*}
It can be verified that these variances are zero at convergence, when $\mu_i=\nu_i$ and $T_{ii}^2 = \Lambda_{ii}$ $\forall$ $i$. The variance of gradients with respect to $\mu$ of FD and SD are larger than that of KLD if $\Lambda_{ii} > 0.5$ and $\Lambda_{ii} / T_{ii}^2 > 0.5 $ respectively. Assuming $\mu_i = \nu_i$ for the SD, the variance of gradients with respect to $T$ of FD and SD are larger than that of KLD if $T_{ii}^2 + \Lambda_{ii} > 0.5$ and $\Lambda_{ii}/T_{ii}^2 > 0.25$ respectively.

In summary, the variance of gradient estimates based on FD is larger than that of KLD once $\Lambda_{ii}>0.5$, regardless of the values of the variational parameters, and variance inflation is larger for $T$ than $\mu$. For SD, the inflation factor involves the ratio $\Lambda_{ii}/T_{ii}^2$, so variance inflation relative to KLD can be reduced if $T_{ii}^2 > \Lambda_{ii}$.

Next, we investigate the variance of gradient estimates for a multivariate Gaussian target with $d=49$ in a real setting. The true precision matrix $\Lambda$ and mean $\nu$, visualized in the first two plots of Fig \ref{grad_compare}, are derived from MCMC samples obtained by fitting a logistic regression model to the German credit data in Section \ref{sec_logreg}. The diagonal entries of $\Lambda$ range from 0.52 to 148.82 with a mean of 45.68. We set $T = 10 I_d$ and $\mu = 0$ to represent an uninformative initialization. The stochastic gradients with respect to $\mu$ and $T$ are computed for each divergence by generating $z \sim \N(0, I_d)$ for 1000 iterations. The standard deviation of these gradient estimates are calculated for $\mu_i$ and $T_{ii}$ for $i=1, \dots, d$, and summarized using boxplots in Fig \ref{grad_compare}. The $y$-axis of the boxplots has a log scale. KLD has the smallest standard deviation, followed by SD, while the standard deviation of FD is much larger than SD and KLD for both $\mu$ and $T$. Although $\Lambda$ is not a diagonal matrix, these findings are consistent with our earlier analysis. This example highlights the difficulty in using SGD to minimize the FD and SD due to the much larger variance in gradient estimates relative to KLD, which motivates an alternative optimization procedure described next.

\section{SGD based on batch approximation}
\label{sec_sgd_batch_approximation}
SGD based on the reparametrization trick faces multiple issues such as increased computational and storage costs due to the Hessian, and high variance in gradient estimates. To address these challenges, we propose  alternative algorithms (named FDb and SDb) in this section, which minimize estimates of the FD and SD computed using a batch of samples randomly simulated from the current variational approximation at each iteration. In this approach, the gradients are biased but they no longer depend on the Hessian, leading to reduced computation costs and improved convergence. In Section \ref{sec_related}, we discuss how this approach iteratively refines the variational approximation by emulating gradients of the target posterior evaluated on each batch of samples. This approach is also more scalable and stable compared to BaM \citep{Cai2024a} for high-dimensional models with conditional independence structure that can be exploited. Section \ref{sec_batch_meanfield} analyzes the behavior of FDb and SDb under the mean-field assumption in the limit of an infinite batchsize.

\begin{algorithm}[tb!]
\caption{SGD based on batch approximation}
\label{BatchAlgs}
\begin{algorithmic}[1]
\Input Initial $\mu\in\mathbb{R}^d$, initial $T^*$, batchsize $B$, stepsize schedule $\{\rho_t\}$
\While{not converged}
    \State $T \gets \Call{Map}{T^*}$,\quad $D \gets \Call{BuildD}{T}$
    \State Sample $z_i \sim \mathcal{N}(0, I_d)$ 
    \State $\theta_i \gets \mu + T^{-\top} z_i$ for $i=1,\dots,B$
    \State Compute $g_h(\theta_i)$ for $i=1,\dots,B$
    \State Compute summary statistics:
        \begin{align*}
        \overline{\theta} &\gets \tfrac{1}{B}\sum_{i=1}^B \theta_i,\quad 
        C_{\theta g} \gets \tfrac{1}{B}\sum_{i=1}^B (\theta_i-\overline{\theta})(g_h(\theta_i)-\overline{g}_h)^\top, \\
        \overline{g}_h &\gets \tfrac{1}{B}\sum_{i=1}^B g_h(\theta_i),\quad
        C_\theta \gets \tfrac{1}{B}\sum_{i=1}^B (\theta_i-\overline{\theta})(\theta_i-\overline{\theta})^\top, \\
        C_g &\gets \tfrac{1}{B}\sum_{i=1}^B (g_h(\theta_i)-\overline{g}_h)(g_h(\theta_i)-\overline{g}_h)^\top
        \end{align*}
    \State $U \gets C_\theta + (\mu - \overline{\theta})(\mu - \overline{\theta})^\top$,\quad $g_\mu \gets 2TT^\top(\mu-\overline{\theta}) - 2\overline{g}_h$
    \If{\textsc{FD}b}
        \State $\mu \gets \mu - \rho_t TT^\top g_\mu$,\quad $W \gets C_{\theta g} - (\mu - \overline{\theta}) \overline{g}_h^\top$
        \State $g_T \gets 2\big(W + W^\top + TT^\top U + U TT^\top\big)T$
    \ElsIf{\textsc{SD}b}
        \State $\mu \gets \mu - \rho_t g_\mu$,\quad $V \gets C_g + \overline{g}_h \overline{g}_h^\top$
        \State $g_T \gets 2\big(UT - T^{-\top} T^{-1} V T^{-\top}\big)$
    \EndIf
    \State $T^* \gets T^* - \rho_t D g_T$
    \State $t \gets t+1$
\EndWhile
\end{algorithmic}
\end{algorithm}

The SD and FD can be written respectively as
\begin{align*}
S_{q_\lambda}(\lambda) &= \E_{q_\lambda}  \| g_h(\theta) + \Sigma^{-1} (\theta - \mu) \|^2_\Sigma \\
&=  \E_{q_\lambda} \left\{ g_h(\theta)^\top \Sigma g_h(\theta) + 2 g_h(\theta)^\top (\theta - \mu) \right.\\
& \left. \quad + (\theta - \mu) ^\top \Sigma^{-1} (\theta - \mu) \right\},\\
F_{q_\lambda}(\lambda)  &= \E_{q_\lambda}  \| g_h(\theta) + \Sigma^{-1} (\theta - \mu) \|^2\\
&= \E_{q_\lambda} \left\{ g_h(\theta)^\top g_h(\theta) + 2 g_h(\theta)^\top\Sigma^{-1} (\theta - \mu) \right. \\
& \left. \quad + (\theta - \mu) ^\top \Sigma^{-2} (\theta - \mu) \right\},
\end{align*}
where $g_h(\theta) = \nabla_\theta \log h(\theta)$ and the subscript $q_\lambda$ emphasizes that expectation is with respect to $q_\lambda(\theta)$. To estimate SD and FD at the $t$-iteration, we generate $B$ samples  $\{\theta_1, \dots, \theta_B\}$ from the current estimate of the variational density $q_t(\theta) = \N(\theta|\mu^{(t)}, \Sigma^{(t)})$. This can be done by generating $z_i \sim \N(0, I_d)$ and computing $\theta_i = \mu^{(t)} +T^{(t)-\top} z_i$ for $i=1, \dots, B$, where $\Sigma^{(t)} = T^{(t)-\top}T^{(t)-1}$. By using the summary statistics computed in step 6 of Algorithm \ref{BatchAlgs}, estimates of SD and FD at iteration $t$ are 
\begin{equation} \label{batch_est}
\begin{aligned}
\hat{S}_{q_t}(\lambda) &= \frac{1}{B} \sum_{i=1}^B \left\{ g_h(\theta_i)^\top \Sigma g_h(\theta_i) + 2 g_h(\theta_i)^\top (\theta_i - \mu) \right. \\
& \left. \quad + (\theta_i - \mu) ^\top \Sigma^{-1} (\theta_i - \mu) \right\} \\
&= \tr(V \Sigma) + \tr(U \Sigma^{-1}) +  2 \tr(W), \\
\hat{F}_{q_t}(\lambda) 
&= \frac{1}{B} \sum_{i=1}^B \left\{ g_h(\theta_i)^\top g_h(\theta_i) + 2 g_h(\theta_i)^\top \Sigma^{-1}(\theta_i - \mu) \right. \\
&\left. \quad + (\theta_i - \mu) ^\top \Sigma^{-2} (\theta_i - \mu) \right\} \\
&= \tr(V) + \tr(U \Sigma^{-2}) + 2\tr(W\Sigma^{-1}),
\end{aligned}
\end{equation}
where  $U = C_\theta + (\mu - \overline{\theta}) (\mu - \overline{\theta})^\top$, $V= C_g + \overline{g}_h \overline{g}_h^\top$, $W= C_{\theta g}  -(\mu - \overline{\theta})\overline{g}_h^\top$ and the subscript $q_t$ indicates that samples are drawn from $q_t$. Differentiating with respect to $\mu$ and $T$, 
\begin{align}
\nabla_\mu \hat{S}_{q_t}(\lambda) &= 2\Sigma^{-1} (\mu - \overline{\theta}) - 2 \overline{g}_h, \label{gradbatch} \\ 
\nabla_\mu \hat{F}_{q_t}(\lambda) &= \Sigma^{-1} \nabla_\mu \hat{S}_{q_t}(\lambda), \nonumber\\
\nabla_{\vech(T)} \hat{S}_{q_t}(\lambda) &= 2\vech( U T  - \Sigma V T^{-\top} ),\nonumber \\
\nabla_{\vech(T)} \hat{F}_{q_t}(\lambda) &= 2\vech \{( W + W^\top  + \Sigma^{-1} U  \nonumber\\
& \quad +U \Sigma^{-1}) T\}. \nonumber
\end{align}

These gradient estimates of SD and FD are biased because the $\theta$'s are replaced by samples $\{\theta_1, \dots, \theta_B\}$ generated from $q_t(\theta) = \N(\theta|\mu^{(t)}, \Sigma^{(t)})$, and are no longer functions of $(\mu, \Sigma)$ when we derive the gradients. On the other hand, the reparametrization trick in Section \ref{sec_sgd_reparametrization trick} produces unbiased estimates because the $\theta$'s are regarded as samples from $q(\theta) = \N(\theta|\mu, \Sigma)$, and remain functions of $(\mu, \Sigma)$ when the chain rule is applied to find the gradients.

With the batch approximation, all gradients are independent of the Hessian, which reduces computation costs significantly and enhances stability during optimization. As before, we only update elements of $\vech(T)$ not fixed by sparsity, and ensure positivity of diagonal entries in $T$ by applying a transformation. SGD algorithms for updating $(\mu, T)$ based on minimizing the batch approximated FD and SD are outlined in Algorithm \ref{BatchAlgs}.

\subsection{Interpretation and related methods} \label{sec_related}
Previously, \cite{Elkhalil2021} designed autoencoders based on minimizing a batch approximation of the Fisher divergence using SGD. \cite{Cai2024a} also proposed a BaM algorithm that derived closed form updates of $(\mu, \Sigma)$ by minimizing the objective,  
\begin{equation*}\label{BAM}
\hat{S}_{q_t}(\lambda) + (2 / \rho_t) \KL(q_t\|q_\lambda),
\end{equation*}
with respect to $\lambda$ at the $t$th iteration, where $\rho_t = Bd/t$ is the learning rate. BaM can be interpreted as a proximal point method that produces a sequence of variational densities $q_0, q_1, \dots$ such that $q_{t+1}$ matches the scores $g_h(\theta)$ at $\{\theta_1, \dots, \theta_B\}$ on average better than $q_t$, while the KLD based penalty ensures stability by preventing $q_{t+1}$ from deviating too much from $q_t$. Similarly, Algorithm \ref{BatchAlgs} can be interpreted as minimizing \begin{equation*}
\hat{S}_{q_t}(\lambda) + (1 / 2\rho_t) \left\|\lambda-\lambda_t \right\|^2
\end{equation*}
with respect to $\lambda$, where an $\ell_2$ penalty is used instead, and a linear approximation of $\hat{S}_{q_t}(\lambda)$ at $\lambda_t$ is considered. Then 
\begin{equation*}
\begin{aligned}
\lambda_{t+1} 
& = \arg\min_{\lambda} \left\{ \hat{S}_{q_t}(\lambda_t) +  \nabla_\lambda \hat{S}_{q_t}(\lambda_t)^\top (\lambda-\lambda_t) \right.\\
&\quad \left.+ \ (1 / 2\rho_t)  \left\| \lambda - \lambda_t \right\|^2 \right\}
= \lambda_t - \rho_t \nabla_\lambda \hat{S}_{q_t}(\lambda_t),
\end{aligned}
\end{equation*}
which corresponds to the SGD update with stepsize $\rho_t$ employed in Algorithm \ref{BatchAlgs}. This discussion extends similarly to the FD.

\begin{figure*}
\centering
\includegraphics[width=4.7645in]{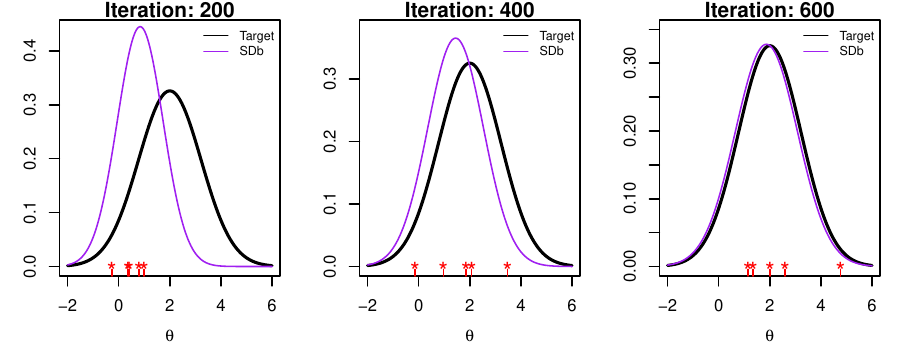}
\caption{Progression of SDb for a Gaussian target where the red $*$'s mark the randomly chosen samples.}\label{SDb_plot}
\end{figure*}

Instead of viewing \eqref{gradbatch} as biased estimates of the gradients of SD and FD, we can consider $\hat{S}_{q_t}(\lambda)$ and $\hat{F}_{q_t}(\lambda)$ as new objectives, which measure the divergence between $q_\lambda(\theta)$ and $p(\theta|y)$ based on their gradients evaluated at randomly selected samples at each iteration. Indeed, $\hat{S}_{q_t}(\lambda)$ and $\hat{F}_{q_t}(\lambda)$ reduce to zero when $q_\lambda(\theta) = p(\theta|y)$, which can be seen from \eqref{batch_est} by plugging in $g_h(\theta) = -\Sigma^{-1}(\theta - \mu)$, as each term in the sums is equal to zero. This supports their use as optimization objectives. At each iteration $t$, $q_{t+1}$ updates $q_t$ so as to reduce the difference in gradients between $q_t$ and the true posterior when evaluated on the randomly selected batch of samples. As $q_t$ converges to $p(\theta|y)$, the samples also shift towards the region where the true posterior has high probability mass. Figure \ref{SDb_plot} illustrates the progression of SDb with batchsize $B=5$ using $\N(2, 1.5)$ as target. As optimization proceeds, SDb gradually refines the variational density by emulating the gradients of the target on the batch of samples and converges steadily toward the target.

Despite the preceding discussion and empirical evidence in Section \ref{sec_applications}, it is important to establish formal convergence for SDb and FDb, as these algorithms can be interpreted as either relying on biased gradients of the SD and FD, or unbiased gradients based on objectives that change at each iteration, and convergence is not guaranteed in either case. In this article, we do not resolve these issues definitively, and we will leave these as open problems for future work. However, we make some contribution in this direction by providing Theorem \ref{thm-convSDb} below, which proves the convergence of SDb in the limit of infinite batchsize with a Gaussian target, where natural gradients \citep{Tan2024a} are used with a constant stepsize. The proof of Theorem \ref{thm-convSDb} given in the supplement follows that of BaM in \cite{Cai2024a} closely under similar settings. There are several differences between the conditions in Theorem \ref{thm-convSDb} and SDb as implemented in Algorithm \ref{BatchAlgs}. In particular, we update $T_t$ instead of $\Sigma_t^{-1}$, using an adaptive instead of constant stepsize, based on Euclidean rather than natural gradients, using a finite instead of infinite batchsize.

\begin{theorem}\label{thm-convSDb}
Suppose the target is $\N(\mu, \Lambda)$ and the variational approximation at iteration $t$ is $\N(\mu_t, \Sigma_t)$, where $\lambda_t = (\mu_t, \Sigma_t)$. Define the normalized errors,
\begin{align*}
\epsilon_t &= \Lambda^{1/2}(\mu_t - \nu), 
\quad 
\Delta_t &= \Lambda^{-1/2} (\Sigma_t^{-1} - \Lambda) \Lambda^{-1/2}. 
\end{align*}
If $\hat{S}_q(\lambda)$ is minimized using SGD using a constant stepsize $0 <\rho<1/4$, based on the natural gradients updates,
\begin{align*}
& \Sigma_{t+1}^{-1} = \Sigma_t^{-1} + 2 \rho \nabla_{\Sigma} \hat{S}_{q_t} (\lambda_t), \\
&\mu_{t+1} = \mu_t - \rho \Sigma_{t+1} \nabla_\mu \hat{S}_{q_t} (\lambda_t), 
\end{align*}
then in the limit of infinite batch size ($B \to \infty$), $\|\epsilon_t\| \to 0$ and $\|\Delta_t\| \to 0$ as $t \to \infty$, where $\| \cdot \|$ denotes the spectral norm.
\end{theorem}

Algorithm \ref{BatchAlgs} differs from BaM in several key aspects. While BaM relies on KLD regularization for stability and has closed-form updates for $(\mu, \Sigma)$, we use an $\ell_2$ penalty and a linearization of the batch approximation leading to SGD. By avoiding linearization of $\hat{S}_{q_t}(\lambda)$ and SGD, the number of iterations required for convergence is reduced in BaM, but each BaM iteration is expensive as the closed form update of $\Sigma$ involves inverting a $d \times d$ matrix with cost of $\mathcal{O}(d^3)$, although this can be reduced to $\mathcal{O}(d^2B + B^3)$ for small batchsize $B$ through low rank solvers. This high cost can result in long runtimes in high dimensions. Moreover, BaM is designed for full covariance matrices and it is not clear how sparsity can be enforced in the precision matrix to take advantage of the posterior conditional independence structure in hierarchical models. BaM can also run into instability and numerical issues with ill-conditioned matrices in practice, which may not be alleviated even with larger batchsizes. On the other hand, SGD allows updating of the Cholesky factor of the precision matrix, where sparse structures can be easily enforced. Smaller batchsizes can also be used, which further reduces the computation and storage burden. While BaM is suited for full covariance Gaussian VI, our approach provides a more scalable and stable alternative for high-dimensional hierarchical models with conditional independence structure.

The batch, match and patch (pBaM) algorithm \citep{Modi2025} extends BaM to higher dimensions through a patch step, that projects the covariance matrix into a family of low rank plus diagonal matrices \citep{Ong2018} such that $\Sigma = \Lambda \Lambda^\top + \Psi$, via an expectation maximization (EM) procedure. The rank $K$ of the low-rank factor $\Lambda$ controls the trade-off between computational efficiency and approximation accuracy, and the approximation accuracy improves with a larger $K$ for a given batchsize $B$. Each EM step has a cost of $\mathcal{O}(dK^{2} + K^{3} + KBd)$ while each BaM step has a cost of $\mathcal{O}(dB^{2} + B^{3} + KBd)$. Hence, pBaM is much more scalable in high-dimensional regimes, where storing and updating dense covariance matrices is impractical. Unlike pBaM, our approach does not require the tuning of additional hyperparameters such as $K$, but is still scalable and able to exploit sparsity.

\subsection{Batch approximated objective under mean-field}
\label{sec_batch_meanfield}
Next, we investigate behavior of the batch approximated FD and SD under the mean-field assumption considered in Section \ref{sec:ord_div_Gaussian}. Suppose the target $p(\theta|y)$ is $\N(\nu, \Lambda^{-1})$ with non-diagonal precision matrix $\Lambda$, and the variational approximation $q(\theta)$ is  $\N(\mu, \Sigma)$ where $\Sigma$ is a diagonal matrix. Using $B > 1$ samples $\{\theta_1, \dots, \theta_B\}$ from an estimate of $q$, $\hat{q}(\theta) = \N(\hat{\mu}, \hat{\Sigma})$, where $\hat{\Sigma}$ is also a diagonal matrix, the batch approximated SD and FD are
\begin{align*}
\hat{S}_{q}(\lambda)
&= \sum_{i=1}^d ( V_{ii} \Sigma_{ii} + U_{ii} \Sigma^{-1}_{ii}) + 2\tr(W), 
\\
\hat{F}_{q}(\lambda)
&= \sum_{i=1}^d (U_{ii} \Sigma^{-2}_{ii} + 2W_{ii} \Sigma^{-1}_{ii}) + \tr(V).
\end{align*}

\begin{lemma}\label{Lem2}
$\hat{S}_{q}(\lambda)$ is minimized at $\Sigma^{\hat{S}}_{ii} = \sqrt{C_{\theta,ii}/C_{g, ii}}$ and $\mu_i^{\hat{S}} = \overline{\theta}_i + \overline{g}_{h,i} \Sigma^{\hat{S}}_{ii}$ for $i=1, \dots, d$. 
If the diagonal entries of $C_{\theta g}$ are all negative, then $\hat{F}_{q}(\lambda)$ is minimized at 
$\Sigma^{\hat{F}}_{ii} = -C_{\theta,ii}/C_{\theta g,ii}$ and $\mu_i^{\hat{F}} = \overline{\theta}_i + \overline{g}_{h,i} \Sigma^{\hat{F}}_{ii}$ for $i=1, \dots, d$.
\end{lemma}

Next, we study limiting behavior of the batch approximated SD and FD as the batchsize $B \to \infty$. Theorem \ref{thm2} relies on the limits of summary statistics step 6 of Algorithm \ref{BatchAlgs} presented in Lemma \ref{Lem3}.

\begin{lemma} \label{Lem3}
Suppose $\{\theta_1, \dots, \theta_B\}$ are samples from $\hat{q}(\theta) = \N(\hat{\mu}, \hat{\Sigma})$ and the target is $p(\theta|y) = \N(\nu, \Lambda^{-1})$. As $B\rightarrow\infty$, 
\begin{equation*} 
\begin{gathered}
\overline{\theta} \xrightarrow{\text{a.s.}} \hat{\mu}, \quad
C_\theta  \xrightarrow{\text{a.s.}} \hat{\Sigma}, \quad 
\overline{g}_h  \xrightarrow{\text{a.s.}} \Lambda(\nu-\hat{\mu}), \\
C_g  \xrightarrow{\text{a.s.}} \Lambda\hat{\Sigma}\Lambda, \quad
C_{\theta g}  \xrightarrow{\text{a.s.}} -\hat{\Sigma}\Lambda.
\end{gathered} 
\end{equation*}
\end{lemma}

\begin{theorem} \label{thm2}
Suppose the target $p(\theta|y)$ is $\N(\nu, \Lambda^{-1})$. Let the variational approximation $q(\theta)$ be $\N(\mu, \Sigma)$, and $\hat{q}(\theta) = \N(\theta|\hat{\mu}, \hat{\Sigma})$ be an estimate of $q(\theta)$, where $\Sigma$ and $\hat{\Sigma}$ are both diagonal matrices. As $B \rightarrow \infty$, $\hat{S}_{q}(\lambda)$ and $\hat{F}_{q}(\lambda)$ are minimized at $(\mu^{\hat{S}}, \Sigma^{\hat{S}})$ and $(\mu^{\hat{F}}, \Sigma^{\hat{F}})$ respectively, where 
\begin{gather*}
\Sigma^{\hat{S}}_{ii} \xrightarrow{\text{a.s.}}  \sqrt{ \frac{\hat{\Sigma}_{ii}}{ \sum_{j=1}^d \hat{\Sigma}_{jj} \Lambda_{ij}^2 } },\qquad
\Sigma^{\hat{F}}_{ii} \xrightarrow{\text{a.s.}}  \frac{1}{\Lambda_{ii}},\\
\mu_i^{\hat{S}} \xrightarrow{\text{a.s.}}  \hat{\mu}_i +  \sqrt{ \frac{\hat{\Sigma}_{ii}}{ \sum_{j=1}^d \hat{\Sigma}_{jj} \Lambda_{ij}^2 } } \sum_{j=1}^d \Lambda_{ij} (\nu_j - \hat{\mu}_j), \\
\mu_i^{\hat{F}} \xrightarrow{\text{a.s.}}  \hat{\mu}_i +  \frac{1}{\Lambda_{ii}}\sum_{j=1}^d \Lambda_{ij} (\nu_j - \hat{\mu}_j).
\end{gather*}
\end{theorem}

From Lemma \ref{Lem3}, $C_{\theta g}$ converges almost surely to $-\hat{\Sigma}\Lambda$, with $i$th diagonal entry $-\hat{\Sigma}_{ii} \Lambda_{ii} < 0$. Thus, diagonal elements of $C_{\theta g}$ are likely negative for a sufficiently large $B$, but may be positive for a small batchsize $B$. In that case, assuming $\mu_i = \bar{\theta}_i + \Sigma_{ii} \bar{g}_{h,i}$, $\nabla_{\Sigma_{ii}} \hat{F}_{q}(\lambda) = -2 \Sigma_{ii}^{-2} (C_{\theta, ii} + C_{\theta g, ii}) < 0$, and $\hat{F}_{q}(\lambda)$ decreases monotonically as $\Sigma_{ii} \rightarrow \infty$. Thus the batch approximated FD faces the issue of ``variance explosion". This is in stark contrast to results in Section \ref{sec:ord_div_Gaussian} where the FD has a closed form solution. On the other hand, the batch approximated SD no longer faces the issue of ``variational collapse", and has a closed form solution for any $B>1$. As $B \rightarrow \infty$, $\Sigma^{\hat{S}}_{ii} \xrightarrow{\text{a.s.}}  
\sqrt{ \hat{\Sigma}_{ii} / (\sum_{j=1}^d \hat{\Sigma}_{jj} \Lambda_{ij}^2) }$, the limit of which is equal to that in \eqref{cov-M} where $M=\hat{\Sigma}$ in the weighted Fisher divergence. It follows from Theorem \ref{thm1} that $\Sigma^{\hat{S}}_{ii} \leq \Sigma^{\hat{F}}_{ii} = \Sigma^{\text{KL}}_{ii}$ in the limit of infinite batchsize. Thus the batch approximated SD underestimates the posterior variance more severely than the batch approximated FD, for which the posterior variance estimate matches that of the KLD, as $B \rightarrow \infty$. However, unlike the FD and SD, the true mean $\nu$ is not recovered by the batch approximated FD and SD even as $B \rightarrow \infty$, unless $\Lambda$ is a diagonal matrix.

\section{Applications}\label{sec_applications}
We evaluate the performances of Algorithms \ref{Algs} and \ref{BatchAlgs} by applying them to logistic regression, GLMMs and stochastic volatility models, and compare their results with BaM, pBaM and MCMC. MCMC sampling is performed using RStan by running 2 chains in parallel, each with 20,000 iterations. The first half is discarded as burn-in, while the remaining 20,000 draws are used to compute kernel density estimates, regarded as the gold standard.

As BaM allows a full covariance matrix, while pBaM uses a more restrictive factor covariance structure and hence may have lower approximation accuracy, we  use BaM whenever it is computationally feasible. pBaM is only used in high-dimensional settings, where BaM is impractical or numerically unstable. The choice of batchsize for FDb, SDb, BaM and pBaM is dependent on the method and model complexity, due to the trade-off between computational efficiency and approximation accuracy. For FDb and SDb, small batchsizes are often sufficient, as only a small step is taken at each iteration due to the reliance on noisy gradient estimates in SGD. In contrast, BaM uses closed form updates that involve matrix inversion and larger batchsizes are necessary to ensure stability and avoid ill-conditioned updates.

To evaluate the multivariate accuracy of variational approximation relative to MCMC, we use maximum mean discrepancy \citep[MMD, ][]{Zhou2023}. We calculate $M^*=-\log (\text{MMD}^2_u+10^{-5})$, where 
\begin{equation*}
\begin{aligned}
\text{MMD}^2_u &= \frac{1}{m(m-1)}\sum_{i \neq j}^m[k(\textbf{x}_v^{(i)},\textbf{x}_v^{(j)})+k(\textbf{x}_g^{(i)},\textbf{x}_g^{(j)})\\
&\quad -k(\textbf{x}_v^{(i)},\textbf{x}_g^{(j)})-k(\textbf{x}_v^{(j)},\textbf{x}_g^{(i)})],
\end{aligned}
\end{equation*}
$\textbf{x}_v^{(1)},\dots,\textbf{x}_v^{(m)}$ and $\textbf{x}_g^{(1)},\dots,\textbf{x}_g^{(m)}$ represent independent samples drawn from the variational approximation and MCMC respectively, $k$ is the radial basis kernel function and $m=1000$. $M^*$ is computed 50 times for each variational approximation and a higher value indicates better multivariate accuracy. In addition, we assess the ability to capture the marginal mean, mode and standard deviation of each variable accurately using the normalized absolute difference ($|\mu - \mu^*|/\sigma^*$, $|\mu - m^*|/\sigma^*$) and standard deviation ratio $\sigma/\sigma^*$, where $\mu$ and $\sigma$ denote the variational mean and standard deviation, and $\mu^*$, $m^*$, $\sigma^*$ denote the mean, mode and standard deviation of each variable based on MCMC samples. The {\em marginal} posterior mode for each variable is reported rather than the joint posterior mode. This distinction arises in high dimensions because MCMC algorithms predominantly explore the typical set, which covers most of the probability mass and is where draws tend to lie in, but this set can lie far from the neighborhood of the global mode. This shell geometry largely vanishes after marginalizing, making the marginal mode a more stable quantity to estimate from MCMC samples \citep{LiuIhler2013, Betancourt2018}.

To assess convergence, we track unbiased estimates of the lower bound, $\hat{\L}$, averaged over every 1000 iterations for SGD methods and 50 iterations for (p)BaM (BaM and pBaM) to reduce noise. Fewer iterations are used for averaging in (p)BaM, as it uses closed form updates, which lead to more stable trajectories. Moreover, (p)BaM usually requires a larger batchsize and converges faster than SGD methods. Each algorithm is terminated when the gradient of a linear regression line fitted to the past five lower bound averages becomes negative (this indicates that the lower bounds have reached a maximum and begun to fluctuate around it), or when the maximum number of iterations is reached. All experiments are performed on a 16GB Apple M1 computer, using R and Julia 1.11.2.

\subsection{Logistic regression} \label{sec_logreg}
\begin{table*}[htb!]
\centering
\begin{tabular}{@{}cccccccc@{}}
\hline
& & KLD & FDr & SDr & FDb & SDb & BaM \\ 
\hline
\multirow{2}{*}{$\frac{|\mu - m^*|}{\sigma^*}$} & \multicolumn{1}{c}{German} & 0.08$\pm$0.06 & 0.57$\pm$0.68 & 0.76$\pm$0.62 & 0.27$\pm$0.25 & \textbf{0.08$\pm$0.05} & 0.08$\pm$0.06  \\ 
& \multicolumn{1}{c}{a4a} & 0.18$\pm$0.15 & 0.45$\pm$0.52 & 0.38$\pm$0.43 & 0.40$\pm$0.43 & 0.15$\pm$0.15 & \textbf{0.14$\pm$0.16}  \\ 
\hline
\multirow{2}{*}{$\frac{|\mu - \mu^*|}{\sigma^*}$} & \multicolumn{1}{c}{German} & 0.02$\pm$0.02 & 0.56$\pm$0.68 & 0.78$\pm$0.61 & 0.25$\pm$0.26 & \textbf{0.01$\pm$0.01} & \textbf{0.01$\pm$0.01}  \\ 
& \multicolumn{1}{c}{a4a} & 0.08$\pm$0.06 & 0.51$\pm$0.55 & 0.45$\pm$0.51 & 0.46$\pm$0.50 & 0.06$\pm$0.06 & \textbf{0.05$\pm$0.06} \\ 
\hline
\multirow{2}{*}{$\frac{\sigma}{\sigma^*}$} & \multicolumn{1}{c}{German} & \textbf{0.99 $\pm$ 0.02} & 0.92  $\pm$ 0.10 & 0.88  $\pm$ 0.17 & \textbf{0.99  $\pm$ 0.02} & \textbf{0.99  $\pm$ 0.02} & \textbf{0.99  $\pm$ 0.02}  \\ 
& \multicolumn{1}{c}{a4a} & 0.61  $\pm$ 0.21 & 0.12  $\pm$ 0.12 & 0.22  $\pm$ 0.31 & 0.54  $\pm$ 0.83 & 0.71  $\pm$ 0.17 & \textbf{0.97  $\pm$ 0.12} \\ 
\hline
\multirow{2}{*}{time} & \multicolumn{1}{c}{German} & 3.6 (45) & 15.0 (60) & 9.8 (39) & 8.5 (32) & 4.7 (16) & 3.4 (0.9) \\ 
& \multicolumn{1}{c}{a4a} & 19.9 (45) & 108.5 (31) & 55.2 (16) & 14.0 (10) & 72.8 (49) & 13.6 (1.05)  \\ 
\hline
\end{tabular}
\caption{Logistic regression. Mean and standard deviation of normalized absolute difference in mode and mean, and standard deviation ratio (best values highlighted in bold). Runtime is in seconds and number of iterations (in thousands) is given in brackets.}
\label{table_germana4a}
\end{table*}

Consider the logistic regression model where $y = (y_1, \dots, y_n)^\top$ represents $n$ independent binary responses. Each $y_i$ follows a Bernoulli distribution with success probability $p_i$, modeled as
\begin{equation*}
\logit (p_{i}) = X_{i}^\top \theta \quad \text{for} \quad i=1, \dots, n.
\end{equation*}
$X_i \in \mathbb{R}^d$ denotes the covariates of the $i$th observation and $\theta \in \mathbb{R}^d$ denotes the unknown coefficients, which is assigned the prior $\N(0, \sigma_0^2I_d)$ with $\sigma_0^2=100$. Here, the precision matrix of the Gaussian variational approximation is not sparse, but a full matrix. The log joint density of the model, gradient and Hessian are given in the supplement.

We fit the logistic regression model to two real datasets from the UCI machine learning repository. The first is the German credit data, which consists of $1000$ individuals classified as having a ``good'' or ``bad'' credit risk, and $20$ attributes. All quantitative predictors are standardized to have mean zero and standard deviation one, while qualitative predictors are encoded using dummy variables. The second is the Adult data with $48,842$ observations, which is used to predict whether an individual's annual income exceeds \$$50,000$ based on $14$ attributes. For MCMC to be feasible, we use the preprocessed a4a data at \url{www.csie.ntu.edu.tw/~cjlin/libsvmtools/datasets/binary.html}, which has $4781$ training samples derived from the Adult data. After preprocessing, $d = 49$ for German credit data and $d = 124$ for a4a data. As the a4a data has a large number of observations, we only generate $10,000$ MCMC samples from two parallel chains, each consisting of $10,000$ iterations. For these datasets, we only use BaM as it already performs  very well. We use a batchsize of $B=3$ for FDb and SDb, and $B=50$ for BaM. The maximum number of iterations is $60,000$.

\begin{figure}[htb!]
\centering
\includegraphics[width=3.518in]{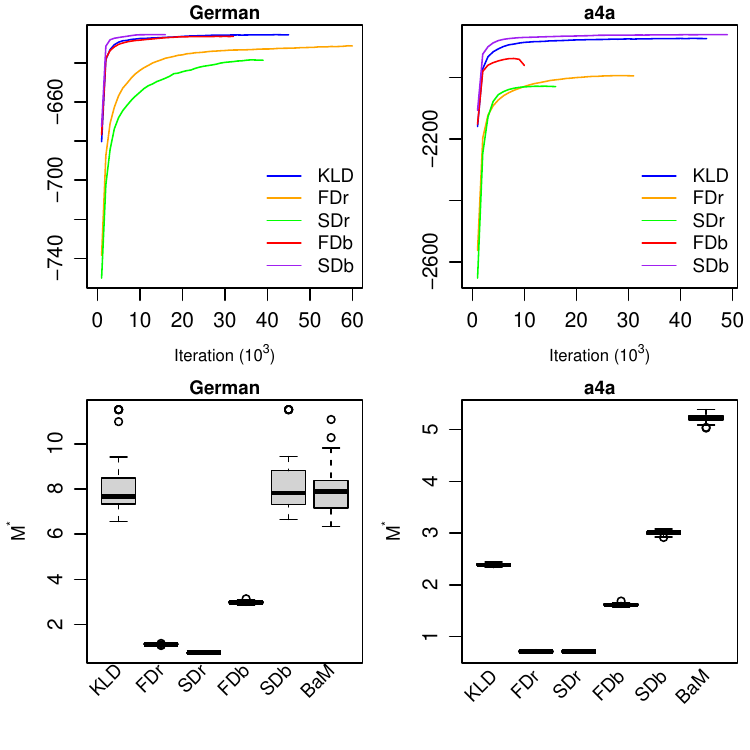}
\caption{Logistic regression. First row contains plots of the lower bound averaged over every 1000 iterations and second row contains boxplots of $M^*$ values.} \label{fig_german}
\end{figure}

Fig \ref{fig_german} shows the progression of the lower bound for SGD methods. FDr and SDr converge very slowly and attain much poorer lower bounds than other methods, likely due to the high variance in their gradient estimates, as discussed in Section \ref{sec variance gradients}. In contrast, SDb converges rapidly and achieves the highest lower bound within the first 1000 iterations, surpassing even KLD. The lower bound achieved by FDb is lower than KLD and SDb although it performs better than FDr and SDr. From the $M^*$ results in Fig \ref{fig_german}, FDr and SDr produce much poorer variational approximations than KLD, while FDb and SDb provide significant improvements over FDr and SDr. In particular, SDb produced better results than KLD.

From Table \ref{table_germana4a}, KLD is the fastest among SGD methods. For the German credit data, FDr, SDr, FDb and SDb each takes $\sim$0.25s per 1000 iterations, but FDb and SDb require fewer iterations to converge. For the a4a data, FDr and SDr require $\sim$3.5s per 1000 iterations compared to  $\sim$1.5s for FDb and SDb, due to the higher cost of computing the Hessian in high dimensions. BaM converges most rapidly, outperforming all SGD methods in runtime. Its $M^*$ values are comparable to KLD for German credit data, and much higher than SGD methods for a4a data.

Overall, $M^*$ values for optimizing FD and SD based on batch approximation are consistently higher than those based on the reparameterization trick. SDb and BaM are better than KLD at estimating the marginal mode and mean accurately for both datasets.

\subsection{Generalized linear mixed model} \label{sec_GLMMs}
Let $y_i=(y_{i1},\dots,y_{in_i})^\top $ denote the $n_i$ observations for the $i$th subject and $y=(y_1^\top,\dots,y_n^\top)^\top$. Each $y_{ij}$ is distributed according to a density in the exponential family, and a smooth invertible link function $g(\cdot)$ relates its mean $\mu_{ij}$ to a linear predictor $\eta_{ij}$ such that
\begin{equation*}
g (\mu_{ij}) = \eta_{ij}=X_{ij}^\top  \beta + Z_{ij}^\top  b_i
\end{equation*}
for $i=1, \dots, n$, $j=1, \dots, n_i$. Here, $\beta \in \mathbb{R}^p$ is the fixed effect, $b_i \in \mathbb{R}^r$ is the random effect characterizing the $i$th subject, and $X_{ij}  \in \mathbb{R}^p$ and $Z_{ij}  \in \mathbb{R}^r$ are the covariates. We assume $b_i \sim \N(0, G^{-1})$ and let $G=WW^{\top}$ be the Cholesky decomposition of precision matrix $G$, where $W$ is a lower triangular matrix with positive diagonal entries. For unconstrained optimization of $W$, we introduce $W^*$ such that $W_{ii}^* = \log(W_{ii})$ and $W_{ij}^* = W_{ij}$ if $i \neq j$, and let $\zeta = \text{vech}(W^*)$. Normal priors, $\beta \sim \N(0, \sigma_0^2I_p)$ and $\zeta \sim \N(0, \sigma_0^2I_{r(r+1)/2})$, where $\sigma_0^2=100$ are assigned. The global variables are $\theta_G = (\beta^{\top},\zeta^{\top})^{\top}$ and the local variables are $\theta_L= (b_1^{\top},\dots,b_n^{\top})^{\top}$. We focus on GLMMs with canonical link functions and responses from the one-parameter exponential family. The gradient $\nabla_\theta \log h(\theta)$ and Hessian $H = \nabla_\theta^2 \log h(\theta)$, which has a sparse structure analogous to that of $\Omega$, are derived in the supplement.

First, consider the epilepsy data \citep{Thall1990} from a clinical trial with $n = 59$ patients, who were randomly assigned to a drug, progabide ($\text{Trt} = 1 $), or a placebo ($\text{Trt} = 0$). The response is the number of seizures experienced by each patient during 4 follow-up periods. Covariates include logarithm of the patient's age at baseline, which is centered by subtracting the mean ($\text{Age}$), logarithm of 1/4 the number of seizures prior to the trial ($\text{Base}$), visit number coded as $-$0.3, $-$0.1, 0.1, 0.3 ($\text{Visit}$), and an indicator of the 4th visit (V4). We consider Poisson mixed models with random intercepts and slopes \citep{Breslow1993},
\begin{multline*}
\text{Epi I:} \log \mu_{ij} = \beta_0+\beta_{\text{Base}} \text{Base}_i+\beta_{\text{Trt}} \text{Trt}_i +\beta_{\text{Age}} \text{Age}_i  
\\
+ \beta_{\text{Base}\text{Trt}} \text{Base}_i \text{Trt}_i +\beta_{\text{V4}} \text{V4}_{ij}   +b_i,
\\
\text{Epi II:}  \log \mu_{ij} = \beta_0+\beta_{\text{Base}} \text{Base}_i+\beta_{\text{Trt}} \text{Trt}_i  +\beta_{\text{Age}} \text{Age}_i 
\\ 
+ \beta_{\text{Base}\text{Trt}} \text{Base}_i \text{Trt}_i  +\beta_{\text{Visit}} \text{Visit}_{ij} +b_{i1} + b_{i2} \text{Visit}_{ij},
\end{multline*}
for $i=1,...,n$, $j=1,...,4$.

\begin{table*}[htb!]
\centering
\begin{tabular}{@{}cccccccc@{}}
\hline
& & KLD & FDr & SDr & FDb & SDb & (p)BaM 
\\ \hline
\multirow{4}{*}{$\tfrac{|\mu - m^*|}{\sigma^*}$} 
& \multicolumn{1}{c}{Epi I}        & \textbf{0.07$\pm$0.05} & 2.06$\pm$1.62 & 2.31$\pm$2.10 & 0.29$\pm$0.19 & \textbf{0.07$\pm$0.05} & 0.08$\pm$0.05 
\\ 
& \multicolumn{1}{c}{Epi II}       & 0.11$\pm$0.09 & 2.13$\pm$1.80 & 2.55$\pm$2.40 & 0.20$\pm$0.15 & \textbf{0.10$\pm$0.08} & 0.10$\pm$0.09 
\\ 
& \multicolumn{1}{c}{Toenail}     & \textbf{0.21 $\pm$ 0.13} & 1.21 $\pm$ 1.30 & 1.58 $\pm$ 2.20 & 0.67 $\pm$ 0.66 & 0.35 $\pm$ 0.21 & 0.33 $\pm$ 0.15 
\\ 
& \multicolumn{1}{c}{Polypharmacy} & \textbf{0.18$\pm$0.11} & 1.00$\pm$1.20 & 1.34$\pm$2.14 & 0.44$\pm$0.39 & 0.22$\pm$0.13 & 0.21$\pm$0.11 
\\ \hline
\multirow{4}{*}{$\tfrac{|\mu - \mu^*|}{\sigma^*}$} 
& \multicolumn{1}{c}{Epi I}        & 0.04$\pm$0.03 & 2.05$\pm$1.63 & 2.31$\pm$2.11 & 0.28$\pm$0.16 & 0.02$\pm$0.02 & \textbf{0.02$\pm$0.01}  \\ 
& \multicolumn{1}{c}{Epi II}       & 0.05$\pm$0.04 & 2.13$\pm$1.79 & 2.56$\pm$2.39 & 0.15$\pm$0.14 & \textbf{0.03$\pm$0.03} & 0.04$\pm$0.04  \\ 
& \multicolumn{1}{c}{Toenail}      & \textbf{0.11 $\pm$ 0.07} & 1.46 $\pm$ 1.20 & 1.83 $\pm$ 2.14 & 0.78 $\pm$ 0.61 & 0.36 $\pm$ 0.23 & 0.28 $\pm$ 0.18 
\\ 
& \multicolumn{1}{c}{Polypharmacy} & \textbf{0.06$\pm$0.04} & 1.16$\pm$1.17 & 1.51$\pm$2.11 & 0.43$\pm$0.39 & 0.17$\pm$0.13 & 0.12$\pm$0.09
\\ \hline
\multirow{4}{*}{$\tfrac{\sigma}{\sigma^*}$} 
& \multicolumn{1}{c}{Epi I}        & 0.95$\pm$0.06 & 0.76$\pm$0.26 & 0.72$\pm$0.34 & 0.81$\pm$0.21 & 0.94$\pm$0.04 & \textbf{0.98$\pm$0.02}  \\ 
& \multicolumn{1}{c}{Epi II}       & 0.96$\pm$0.09 & 0.94$\pm$0.30 & 0.69$\pm$0.28 & 0.88$\pm$0.18 & 0.95$\pm$0.08 & \textbf{0.97$\pm$0.08}  \\ 
& \multicolumn{1}{c}{Toenail}      & \textbf{0.88 $\pm$ 0.05} & 0.35 $\pm$ 0.14 & 0.10 $\pm$ 0.05 & 0.67 $\pm$ 0.14 & 0.76 $\pm$ 0.11 & 0.80 $\pm$ 0.11  
\\ 
& \multicolumn{1}{c}{Polypharmacy} & \textbf{0.94$\pm$0.03} & 0.39$\pm$0.12 & 0.13$\pm$0.05 & 0.80$\pm$0.09 & 0.86$\pm$0.07 & 0.90$\pm$0.06  
\\ \hline
\multirow{4}{*}{time} 
& \multicolumn{1}{c}{Epi I}        & 2.3 (40) & 16.7 (47) & 4.3 (12) & 6.6 (24) & 11.6 (34) & 1.0 (0.4)  \\ 
& \multicolumn{1}{c}{Epi II}       & 5.9 (52) & 73.3 (60) & 26.0 (21) & 16.7 (25) & 39.4 (42) & 17.7 (2.3)  \\ 
& \multicolumn{1}{c}{Toenail}      & 2.8 (11) & 117.4 (30) & 27.7 (7) & 93.7 (30) & 35.0 (10) & 11.3 (1.35)  \\ 
& \multicolumn{1}{c}{Polypharmacy} & 5.1 (10) & 346.1 (30) & 94.7 (7) & 285.8 (30) & 139.2 (12) & 36.7 (1.75)  \\ 
\hline
\end{tabular}
\caption{GLMM. Mean and standard deviation of normalized absolute difference in mode and mean and standard deviation ratio (best values highlighted in bold). Runtime is in seconds and number of iterations (in thousands) is given in brackets. BaM is used for Epi I and Epi II, and pBaM is used for Toenail and Polypharmacy.} \label{table_GLMMs}
\end{table*}

Next, consider the toenail data \citep{DeBacker1998} from a clinical trial comparing two oral antifungal treatments for toenail infections. Each of 294 patients was evaluated for up to 7 visits, resulting in a total of 1908 observations. Patients were randomized to receive 250 mg of terbinafine ($\text{Trt} = 1$) or 200 mg of itraconazole ($\text{Trt} = 0$) per day. The response variable is binary, with 0 indicating no or mild nail separation and 1 for moderate or severe separation. Visit times in months ($t$) are standardized to have mean 0 and variance 1. A logistic random intercept model is fitted to this data, 
\begin{equation*}
\begin{aligned}
\text{logit}(\mu_{ij}) = \beta_0+\beta_{\text{Trt}} \text{Trt}_i+\beta_t t_{ij}
+\beta_{\text{Trt} \times t} \text{Trt}_i \times t_{ij}+b_i, 
\end{aligned}
\end{equation*}
for $i=1,...,294, 1\leq j \leq 7$.

Lastly, we analyze the polypharmacy data \citep{Hosmer2013} which contains 500 subjects, each observed for drug usage over 7 years, resulting in 3500 binary responses. Covariates include Gender (1 for males, 0 for females), Race (0 for whites, 1 otherwise), Age ($\log(\text{age}/10)$) and INPTMHV (0 if there are no inpatient mental health visits and 1 otherwise). The number of outpatient mental health visits (MHV) is coded as $\text{MHV1}=1$ if $1 \leq \text{MHV} \leq 5$, $\text{MHV2}=1$ if $6 \leq \text{MHV} \leq 14$, and $\text{MHV3}=1$ if $\text{MHV} \geq 15$. We consider a logistic random intercept model, 
\begin{align*}
\text{logit}&(\mu_{ij})=\beta_0+\beta_{\text{Gender}}\text{Gender}_i+\beta_{\text{Race}} \text{Race}_i \\
& + \beta_{\text{Age}} \text{Age}_{ij} + \beta_{\text{MHV1}} \text{MHV1}_{ij} + \beta_{\text{MHV2}} \text{MHV2}_{ij}\\
& + \beta_\text{MHV3} \text{MHV3}_{ij}  + \beta_{\text{INPT}} \text{INPTMHV}_{ij} + b_i,
\end{align*}
for $i=1,\dots, 500$, $j=1,\dots,7$.

We set $B=5$ for FDb and SDb. For BAM, $B=100$ for the epilepsy data. For the higher-dimensional toenail and polypharmacy data, BaM is prone to ill-conditioned updates and converges very slowly with smaller batchsizes. Hence, we use pBaM for these two datasets with $B=32$, as recommended by \citet{Modi2025}. We set the rank $K=32$ for toenail data and $K=64$ for polypharmacy data, which is higher in dimension, so a larger $K$ is used. The maximum number of iterations for epilepsy data is $60,000$, which is reduced to $30,000$ for the larger toenail and polypharmacy data.

\begin{figure}[tb!]
\centering
\includegraphics[width=3.518in]{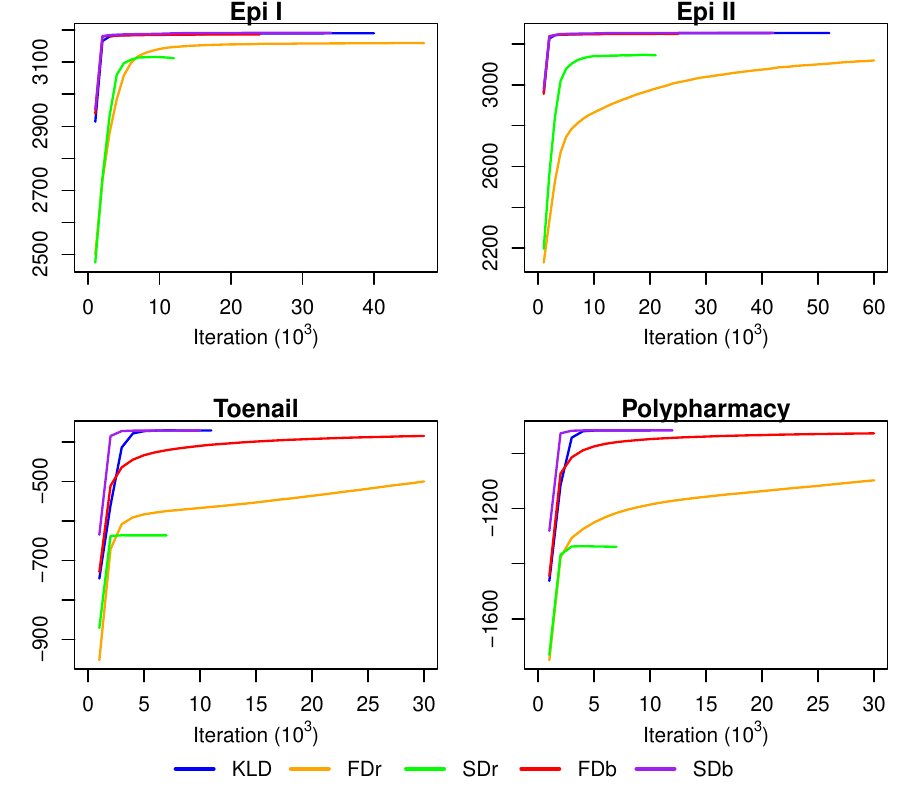}
\caption{GLMM. Lower bound averaged over every 1000 iterations.} \label{GLMMs_LB}
\end{figure}

Fig \ref{GLMMs_LB} shows that SDb is among the fastest to converge among SGD methods, achieving a higher lower bound than KLD for Epi I, Epi II and polypharmacy, and comparable to KLD for toenail. FDb converges rapidly for Epi I and Epi II,  but fails to converge by the maximum number of iterations for toenail and polypharmacy.

\begin{figure}[tb!]
\centering
\includegraphics[width=3.518in]{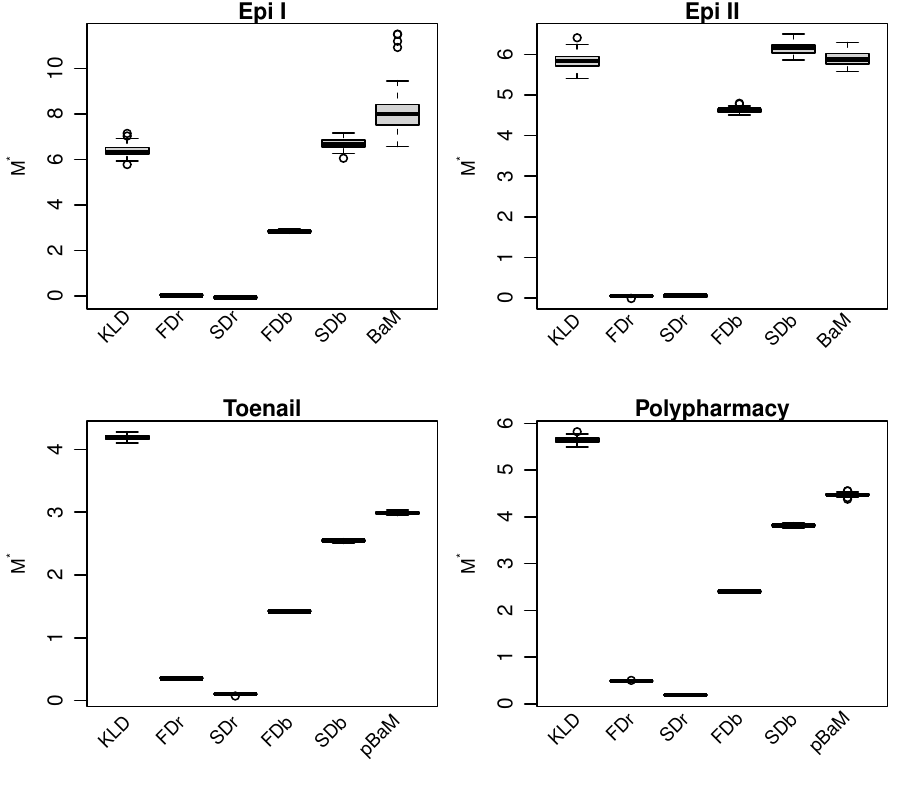}
\caption{GLMM. Boxplots of $M^*$ values.} \label{GLMMs_MMD}
\end{figure}

From Fig \ref{GLMMs_MMD}, FDr and SDr have the lowest $M^*$, while FDb and SDb yield substantial improvements over their reparameterization trick based counterparts. Among SGD methods based on the weighted Fisher divergence, SDb has the highest $M^*$, even surpassing KLD for Epi I and Epi II. BaM also outperforms KLD for Epi I and Epi II. While pBaM outperforms SDb for toenail and polypharmacy, it still falls short of KLD.

From Table \ref{table_GLMMs}, KLD is often able to capture the marginal posterior mean and mode most accurately, with comparable performance from SDb and (p)BaM. While BaM captured the marginal posterior variance most accurately for Epi I and Epi II, pBaM falls behind KLD for toenail and polypharmacy. SDr underestimates the marginal posterior variance most severely, which is reminiscent of the ``variational collapse" problem it faces in the mean-field setting.
(p)BaM is able to converge using the least number of iterations, by leveraging closed-form updates and larger batchsizes. However, the computation cost per iteration of (p)BaM is much higher than SGD methods, which can exploit the sparse structure of the precision matrix. This issue becomes more apparent as the dimension of $\theta$ increases. Among SGD methods, KLD is the fastest. Methods based on FD tend to require more iterations to converge than those based on SD, resulting in longer runtime. SDb converges in about the same number of iterations as KLD, but each iteration takes longer.

\begin{figure*}[htb!]
\centering
\includegraphics[width=6.597in]{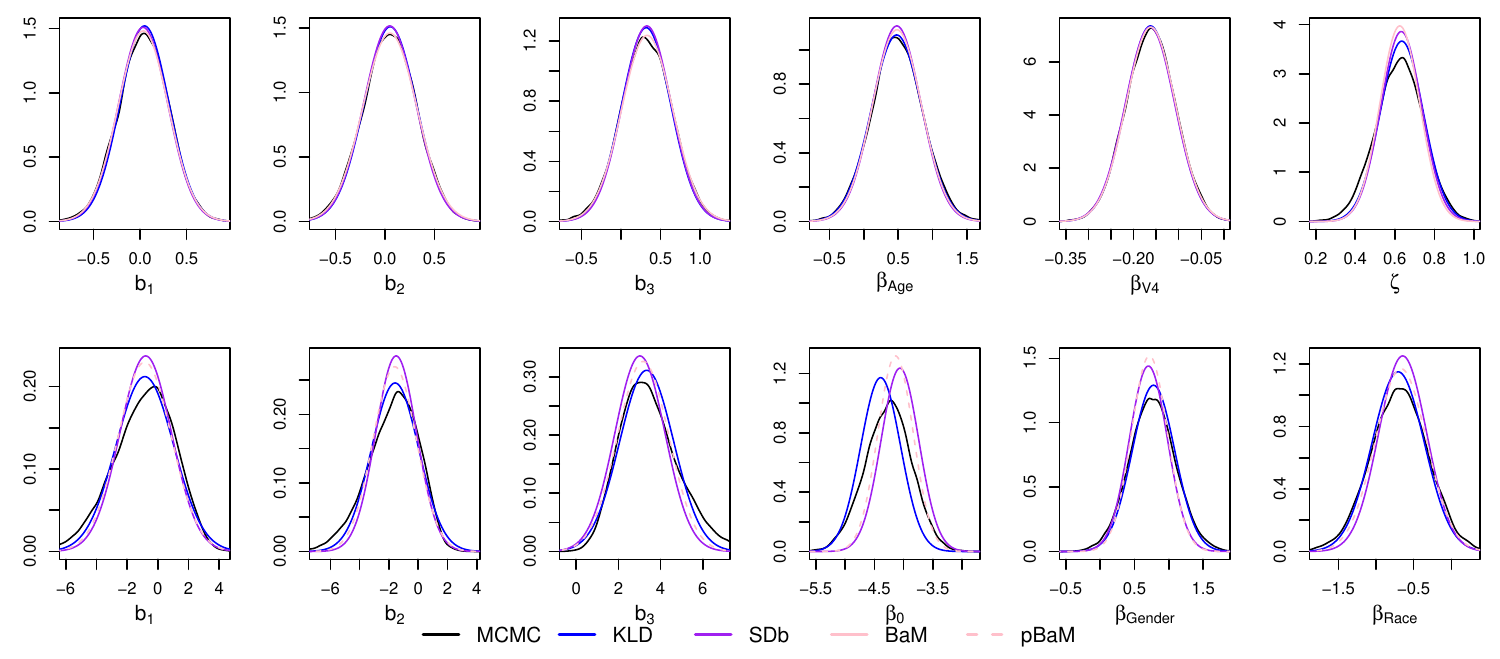}
\caption{Marginal density estimates for some local and global variables in Epi I (first row) and polypharmacy (second row).} \label{EpiPoly_den}
\end{figure*}

Fig \ref{EpiPoly_den} compares the marginal densities estimated using MCMC with variational approximations from KLD, SDb and (p)BaM for some variables in Epi I and polypharmacy. For Epi I, all variational approaches match the MCMC results very closely except for $\zeta$, where SDb underestimated the marginal posterior variance more severely than BaM and KLD. For the higher-dimensional polypharmacy data, KLD matches MCMC results most closely, while SDb and pBaM tend to underestimate the marginal posterior variance although the mode was captured more accurately in some cases.

The results in this section are mixed, although KLD and (p)BaM tend to perform better than other methods. The superior performance of KLD may be related to the findings in Sections \ref{sec:ord_div_Gaussian} and \ref{sec:ord_div_non_Gaussian}, which showed that the mean and variance underestimation for KLD is often less severe than FD or SD, leading to an overall higher accuracy. While SDb and BaM both minimize the batch approximated SD, BaM relies on closed form updates that lead to higher accuracy and faster convergence in low-dimensional problems, compared to SDb which is based on SGD. The advantages of SDb are more apparent for the higher-dimensional stochastic volatility models in Section \ref{sec_SV models}, where (p)BaM fails to converge, and SDb surpasses KLD in overall accuracy and marginal mean and mode estimation. We hypothesize that SDb may perform better than KLD for skewed heavy-tailed posteriors, although this conjecture remains to be verified.

\subsection{Stochastic volatility model} \label{sec_SV models}
The stochastic volatility model is widely used to capture the dynamic nature of financial time series. It provides an attractive alternative to constant volatility models like the Black-Scholes model \citep{black1973pricing}, as the volatility of asset returns evolves over time according to a stochastic process. The response at time $t$ is 
\begin{equation*}
y_t \sim \N(0, \exp(\lambda+\sigma b_t))
\quad \text{for}
\quad t=1,\dots,n,
\end{equation*}
where $\lambda\in \mathbb{R}$, $\sigma>0$, and the latent volatility process $b_t$ follows an autoregressive model of order one such that 
\begin{align*}
b_{t} &\sim \N(\phi b_{t-1},1)\;\text{for}\;t=2,\dots,n, \\
b_1 &\sim \N(0,1/(1-\phi^2)),
\end{align*}
where $\phi\in(0,1)$. To allow unconstrained updates, we apply the transformations, $\alpha = \log \sigma$ and $\psi = \logit(\phi)$. The set of local variables is $\theta_L = (b_1,\dots,b_n)^{\top}$ and global variables are $\theta_G = (\alpha,\lambda,\psi)^{\top}$. We consider the prior $\theta_G \sim \N(0,\sigma_0^2I)$, where $\sigma_0^2=10$. For this model, $b_i$ is independent of $b_j$ given $\theta_G$ a posteriori if $|i-j|>1$. Thus, the Hessian of $\log h(\theta)$ has the same sparsity structure as $\Omega$ in the variational approximation. Both $\nabla_\theta \log(h(\theta))$ and $\nabla_\theta^2 \log h(\theta)$ are derived in the supplement.

\begin{table*}[ht]
\centering
\begin{tabular}{@{}ccccccc@{}}
\hline
& & KLD & FDr & SDr & FDb & SDb \\ 
\hline
\multirow{2}{*}{$\frac{|\mu - m^*|}{\sigma^*}$} & \multicolumn{1}{c}{GBP} & 0.13$\pm$0.09 & 1.03$\pm$0.81 & 0.92$\pm$0.69 & 0.79$\pm$0.60 & \textbf{0.07$\pm$0.05}  \\ 
& \multicolumn{1}{c}{DEM} & 0.11$\pm$0.08 & 1.13$\pm$0.88 & 1.17$\pm$0.71 & 0.96$\pm$0.74 & \textbf{0.07$\pm$0.05}  \\ 
\hline
\multirow{2}{*}{$\frac{|\mu - \mu^*|}{\sigma^*}$} & \multicolumn{1}{c}{GBP} & 0.10$\pm$0.02 & 1.10$\pm$0.86 & 0.98$\pm$0.70 & 0.86$\pm$0.65 & \textbf{0.06$\pm$0.05}  \\ 
& \multicolumn{1}{c}{DEM} & 0.10$\pm$0.03 & 1.17$\pm$0.91 & 1.21$\pm$0.70 & 1.00$\pm$0.77 & \textbf{0.03$\pm$0.02}  \\ 
\hline
\multirow{2}{*}{$\frac{\sigma}{\sigma^*}$} & \multicolumn{1}{c}{GBP} & \textbf{0.92 $\pm$ 0.05} & 0.68 $\pm$ 0.12 & 0.85 $\pm$ 0.19 & 0.48 $\pm$ 0.06 & 0.88 $\pm$ 0.05  \\ 
& \multicolumn{1}{c}{DEM} & \textbf{0.95 $\pm$ 0.03} & 0.60 $\pm$ 0.08 & 0.99 $\pm$ 0.19 & 0.47 $\pm$ 0.04 & 0.91 $\pm$ 0.03  \\ 
\hline
\multirow{2}{*}{time} & \multicolumn{1}{c}{GBP} & 9.2 (20) & 18.9 (30) & 8.3 (13) & 1452.9 (30) & 898.5 (14)  \\ 
& \multicolumn{1}{c}{DEM} & 9.7 (19) & 27.3 (30) & 23.7 (25) & 2748.0 (30) & 1460.3 (13)  \\ 
\hline
\end{tabular}
\caption{Stochastic volatility model. Mean and standard deviation of normalized absolute difference in mode and mean and standard deviation ratio (best values highlighted in bold). Runtime is in seconds and number of iterations (in thousands) is given in brackets.} \label{table_SSMs}
\end{table*}

We analyze two datasets from \texttt{Garch} in the R package \texttt{Ecdat}. The first contains $n = 1323$ observations of the weekday exchange rates of the U.S. Dollar against the British Pound (GBP) from 1 Aug 1980 to 28 Oct 1985. The second contains $n = 1866$ observations of the weekday exchange rates for the U.S. Dollar against the German Deutschemark (DEM) from 2 Jan 1980 to 21 May 1987. For both datasets, the mean-corrected log-return series $\{y_t\}$ is derived from exchange rates $\{r_t\}$ using 
\begin{equation*}
y_t = 100 \times \left\{ \log\left(\frac{r_t}{r_{t-1}}\right) - \frac{1}{n} \sum_{i=1}^n \log\left(\frac{r_i}{r_{i-1}}\right) \right\}.
\end{equation*}

We set $B=10$ for FDb and SDb. We have tried various batchsizes for BaM, but the updates were severely ill-conditioned and BaM failed to converge. The challenge of inferring a full covariance matrix of dimension exceeding 1000 for BaM here is immense, further complicated by the high computational cost of matrix inversion. Although pBaM scales better to high dimensions than BaM, it also fails to converge for $B\in\{32,64,128\}$ and ranks $K \in \{8, 16, 32, 64, 128\}$. This may be because the posterior conditional independence structure of the stochastic volatility model is difficult to capture via a factor covariance matrix. The maximum number of iterations is set as $30,000$.

From Table~\ref{table_SSMs}, SDb provides the best approximations of the mean and mode, while KLD yields the most accurate estimates of the standard deviation. Note that SDr achieves a higher standard deviation ratio of $0.99$ for DEM, but does so with a much larger standard deviation of $0.19$, making KLD more reliable. In terms of runtime, KLD is the most efficient.

Fig \ref{SSMs_MMD} illustrates the impact of varying the batchsize for SDb in terms of convergence rate and approximation accuracy measured by $M^*$. Increasing the batchsize clearly leads to faster convergence and improved accuracy. The total runtime (shown in the legends of the first row) tends to decrease as fewer iterations are required for convergence. This suggests that larger batchsizes can enhance the stability and accuracy of SDb. Notably, the $M^*$ values of SDb exceed those of KLD even with a small $B=3$.

\begin{figure}[tb!]
\centering
\includegraphics[width=3.518in]{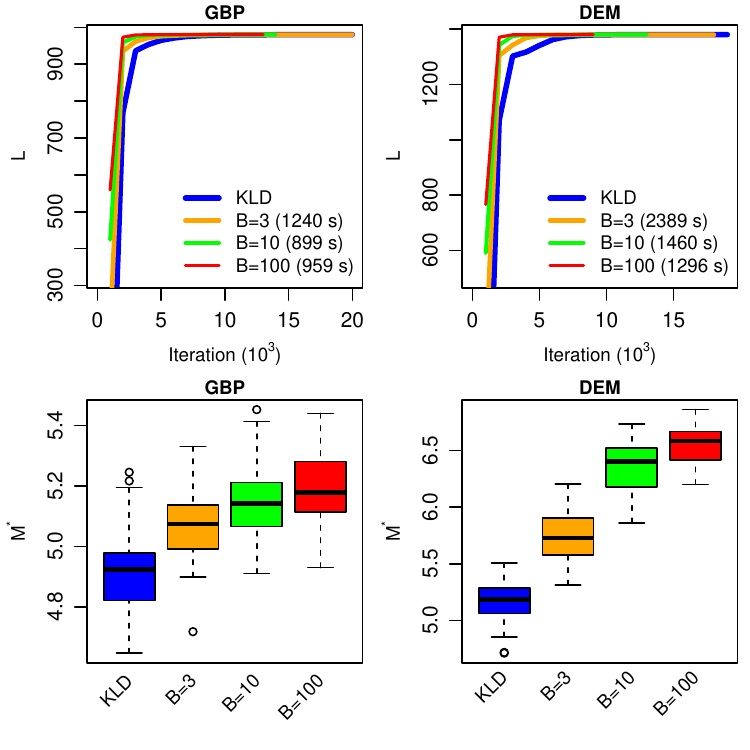}
\caption{Stochastic volatility model. Lower bound averaged over every 1000 iterations and boxplots of $M^*$ values for different batchsizes for SDb.} \label{SSMs_MMD}
\end{figure}

Fig \ref{dm_den} compares the marginal posterior density estimates from MCMC, KLD and SDb ($B=10, 100$) for some local variables and all global variables in DEM. SDb can capture the marginal posterior mode more accurately than KLD, especially for each of the global variables, but has a higher tendency to underestimate the posterior variance. Increasing the batchsize from 10 to 100 helps in reducing underestimation of the posterior variance.

\begin{figure*}[tb!]
\centering
\includegraphics[width=7in]{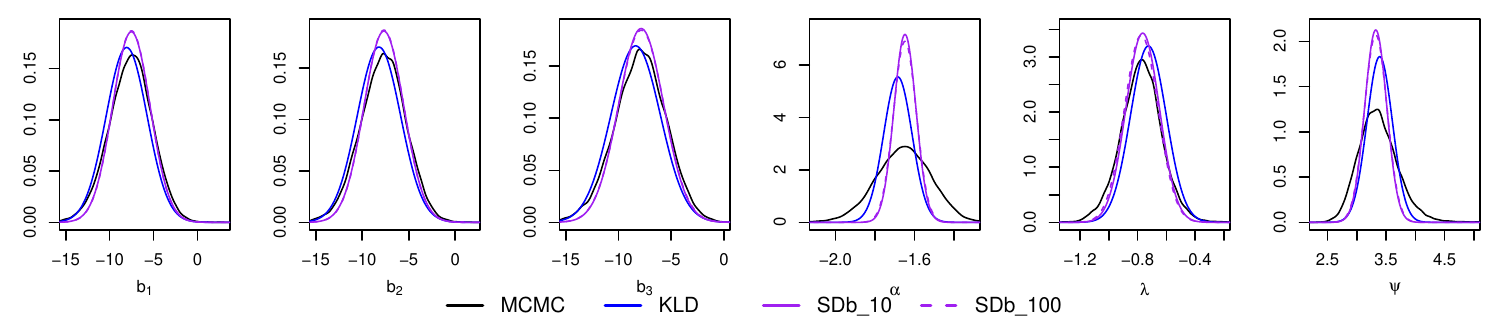}
\caption{Stochastic volatility model. Marginal density estimates for some local and global variables in DEM.} \label{dm_den}
\end{figure*}

\section{Conclusion} \label{sec_conclusion} 
In this article, we evaluate the performance of Gaussian variational inference based on the weighted Fisher divergence by focusing on the FD and SD. First, we consider the mean-field assumption for Gaussian and non-Gaussian targets. We demonstrate that FD and SD tend to underestimate the posterior variance more severely than KLD, and SD can capture the posterior mode more accurately than FD and KLD for skewed targets.

Next, we consider high-dimensional hierarchical models whose posterior conditional independence structure can be captured using a sparse precision matrix in the Gaussian variational approximation. To impose sparsity on the Cholesky factor of the precision matrix, we consider optimization based on SGD and propose two approaches based on the reparametrization trick and a batch approximation of the objective.

The reparametrization trick yields unbiased gradient estimates but involves a Hessian matrix, which is computationally expensive and increases variability in the gradients, leading to reduced stability and slow convergence. To address these issues, we introduce an alternative that minimizes a biased estimate of the FD and SD computed using a random batch of samples at each iteration. This eliminates reliance on the Hessian and improves stability. This approach can also be interpreted as optimizing a new objective, that iteratively improves the match between gradients of the posterior and variational density, at sample points that shift gradually towards regions of high posterior probability. While the general convergence of FDb and SDb remains as an open problem, we make some contribution in this direction by proving the convergence of SDb in the special case of a Gaussian target with infinite batch size, using natural gradients updates with a constant stepsize. We also evaluate the behavior of this new objective under the mean-field assumption for Gaussian targets and show that it alleviates the variational collapse issue faced previously by SD.

The proposed methods are compared to KLD, (p)BaM in applications involving logistic regression, GLMMs and stochastic volatility models. Extensive experiments reveal that FDr and SDr converge very slowly, often to suboptimal variational approximations. FDb and SDb provide substantial improvements over FDr and SDr, with SDb having superior performance in terms of convergence rate and accuracy. (p)BaM, which relies on closed-form updates and hence requires fewer iterations to converge, is very effective for logistic regression. However, it is less efficient than KLD for GLMMs and stochastic volatility models, and its performance gradually worsens as the dimension increases, eventually failing to converge. SDb has an advantage over (p)BaM in high dimensions as it can impose sparsity on the precision matrix, remains feasible computationally and is more stable and less sensitive to poor initialization. SDb can also capture posterior modes more accurately than KLD but is more prone to variance underestimation.

There are several avenues for future research. While this work has focused primarily on two variants of the weighted Fisher divergence (FD and SD), it will be valuable to investigate other variants. Besides Gaussian variational approximations, it is also of interest to investigate the performance of FD and SD under more flexible variational families. While we have used SGD for optimization, the choice of optimizer and associated hyperparameters significantly influences convergence behavior, and it is useful to explore alternative optimization techniques based on natural gradients or which do not rely on SGD. Our findings highlight the potential of the batch approximated SD, and its properties can be  investigated further in other contexts. Finally, proving the convergence of our batch-approximated methods in the practical setting with non-Gaussian targets and finite batches remains as an open problem.

\section{Acknowledgment}
We would like to thank the Editor, Associate Editor and three referees for their comments and helpful suggestions which have improved this manuscript greatly.

%
%

\begin{funding}
Linda Tan's research is supported by the Ministry of Education, Singapore, under its Academic Research Fund Tier 2 (Award MOE-T2EP20222-0002).  David Nott's research is supported by the Ministry of Education, Singapore, under the Academic Research Fund Tier 2 (MOE-T2EP20123-0009).
\end{funding}


\begin{supplement}
\stitle{Zip file}
\sdescription{Julia code and a pdf file containing derivations, proofs of all lemmas and theorems that are not provided in the maniscript, and additional details.} 
\end{supplement}



\bibliographystyle{chicago}
\bibliography{ref}

\newpage

\noindent
\begin{center}
{ \bf \Large Supplementary Material}
\end{center}

\setcounter{section}{0} \renewcommand{\thesection}{S\arabic{section}}
\setcounter{figure}{0} \renewcommand{\thefigure}{S\arabic{figure}}
\setcounter{table}{0} \renewcommand{\thetable}{S\arabic{table}}
\setcounter{equation}{0} \renewcommand{\theequation}{S\arabic{equation}}
\setcounter{lemma}{0} \renewcommand{\thelemma}{S\arabic{lemma}}

\section{Proof of Lemma \ref{Lem - weighted Fisher div betw Gaussians}}
The $M$-weighted Fisher divergence is 
\begin{align*}
S_M(q \| p) &= \E_q\{\|  \nabla_\theta \log q(\theta)  - \nabla_\theta \log p(\theta|y) \|_M^2\}\\
&=\E_q\{ \| \Sigma^{-1}(\theta-\mu) - \Lambda(\theta-\nu) \|_M^2 \}\\
&=\E_q\{  \| (\Sigma^{-1} - \Lambda)(\theta-\mu) - \Lambda(\mu - \nu) \|_M^2 \}\\
&=\tr\{ (\Sigma^{-1} - \Lambda)M (\Sigma^{-1} - \Lambda) \Sigma \}\\
&\quad +(\mu - \nu)^{\top} \Lambda M \Lambda (\mu - \nu)\\
&= \tr(\Sigma^{-1}M) + \tr(\Lambda M \Lambda \Sigma) - 2 \tr(M\Lambda)\\
&\quad +(\mu - \nu)^{\top}\Lambda M \Lambda (\mu - \nu).
\end{align*}
The final result arises from $\tr(AB) = \sum_{i=1}^d A_{ii} B_{ii}$ if $A$ is a diagonal matrix.

\section{Proof of Theorem \ref{thm-t}}
First we present Lemma \ref{lem_symmetry}, which is required in the proof of Theorem \ref{thm-t}.
\begin{lemma} \label{lem_symmetry}
Let $\theta\sim\N(\mu,\Sigma)$. If $f:\mathbb{R}^d\to\mathbb{R}^k$ is integrable and is an odd function of $(\theta-\mu)$ in that $f(\mu - \theta)=-f(\theta-\mu)$, then $\E_{\theta \sim \N(\mu, \Sigma)} [f(\theta-\mu)]=0$.
\end{lemma}
\begin{proof}
Let $\theta'=\theta-\mu$ so that $\theta'\sim \N(0,\Sigma)$. Then $f(\theta')=f(\theta-\mu)=-f(\mu-\theta)=-f(-\theta')$. 
\begin{align*}
\E_{\theta \sim \N(\mu, \Sigma)} & [f(\theta-\mu)]
=\E_{\theta' \sim \N(0, \Sigma)} [f(\theta')]\\
&=\int_{\mathbb{R}^d}f(\theta')\phi(\theta' \mid 0, \Sigma) d\theta'\\
&=\int_{\mathbb{R}^d}f(-\theta')\phi(-\theta' \mid 0, \Sigma) d\theta'\\
&=\int_{\mathbb{R}^d} \{- f(\theta')\} \phi(\theta' \mid 0, \Sigma) d\theta'\\
&=-\E_{\theta' \sim \N(0, \Sigma)} [f(\theta')] \\
&=-\E_{\theta \sim \N(\mu, \Sigma)} [f(\theta-\mu)]
\end{align*}
Thus, $\E_{\theta \sim \N(\mu, \Sigma)} [f(\theta-\mu)]=-\E_{\theta \sim \N(\mu, \Sigma)} [f(\theta-\mu)]=0$.    
\end{proof}

For $q(\theta)=\N(\mu,\Sigma)$, $\nabla_\theta\log q(\theta)=-\Sigma^{-1}(\theta-\mu)$. For the multivariate Student's $t$ distribution, let
\begin{equation*}
\delta(\theta)=(\theta-m)^\top S^{-1}(\theta-m),
\quad
w(\theta)=\frac{\nu+d}{\nu+\delta(\theta)}.
\end{equation*}
Then 
\begin{align*}
\nabla_\theta\log p(y, \theta)
&= -w(\theta)S^{-1}(\theta-m), \\
H_p(\theta) 
& =\nabla_\theta^2\log p(y, \theta) = -w(\theta)S^{-1} \\
& \quad + \frac{2w(\theta)}{\nu+\delta(\theta)}S^{-1}(\theta-m)(\theta-m)^\top S^{-1}.
\end{align*}

For the KLD, the evidence lower bound $\L=\E_q[\log p(y, \theta)-\log q(\theta)]$. Since $\E_q[\log q(\theta)]$ is independent of $\mu$, we only need to focus on the first term. By applying the reparametrization trick described in Section \ref{sec_sgd_reparametrization trick} of the manuscript,
\begin{align*}
\nabla_\mu \mathcal L
&= \E_q[\nabla_\theta \log p(y, \theta)] 
= -\E_q[w(\theta)S^{-1}(\theta-m)].
\end{align*}
Here $w(\theta)$ is even in $(\theta-m)$, while $(\theta-m)$ is odd, so the integrand is odd in $(\theta-m)$. Lemma \ref{lem_symmetry} implies that $\nabla_\mu \mathcal L = - \E_q[w(\theta)S^{-1}(\theta-m)]=0$ at $\mu=m$.

For the FD, from Section \ref{sec_sgd_reparametrization trick} of the manuscript, 
\[
\nabla_\mu F(q\|p)=2\E_q[H_p(\theta)(\nabla_\theta \log p(y, \theta)-\nabla_\theta \log q(\theta))].
\]
Note that $H_p(\theta)$ is even in $(\theta-m)$. Moreover, 
\begin{multline*}
\nabla_\theta \log p(y, \theta)-\nabla_\theta \log q(\theta) \\
= -w(\theta)S^{-1}(\theta-m)+\Sigma^{-1}(\theta-m),
\end{multline*}
which is odd in $(\theta-m)$. Thus the integrand in $\nabla_\mu F(q\|p)$ is odd, and by Lemma \ref{lem_symmetry}, $\nabla_\mu F(q\|p) = 0$ at $\mu=m$.

For the SD, from Section \ref{sec_sgd_reparametrization trick} of the manuscript, 
\[
\nabla_\mu S(q\|p) = 2 \E_q [H_p(\theta)\Sigma(\nabla_\theta \log p(y, \theta)-\nabla_\theta \log q(\theta))].
\]
The same argument as for FD shows that the integrand is odd in $(\theta-m)$, and hence $\nabla_\mu S(q\|p) = 0$ at $\mu=m$. 

Hence $\mu=m$ is a stationary point for all three divergences.

\section{Univariate non-Gaussian target}
We begin by deriving some key expressions that are used throughout our analysis of non-Gaussian target distributions. Suppose the variational approximation $q(\theta) = \phi(\theta|\mu, \sigma^2)$. We have 
\begin{align*}
\log q(\theta) 
&= -\frac{1}{2}\log(2\pi) - \frac{1}{2}\log(\sigma^2) - \frac{(\theta-\mu)^2}{2\sigma^2},
\\
\nabla_{\theta}\log q(\theta)
&=-\frac{\theta-\mu}{\sigma^2}, 
\\
\E_q\left\{\log q(\theta)\right\}
&=-\frac{1}{2}\log(2\pi \sigma^2)-\frac{1}{2}.
\end{align*}
For univariate densities, $S(q\|p) = \sigma^2 F(q\|p)$, so only the expression of FD is presented. 

\subsection{Student's $t$}
When the target is the univariate Student's $t$, such that $\theta \sim t(\nu)$, where $\nu$ is the degree of freedom,
\begin{align*}\label{t_grad}
\log p(y, \theta) 
&= \log\left\{\Gamma\left(\frac{\nu+1}{2}\right)\right\} - \frac{1}{2}\log(\pi\nu)\\
&\quad -\log\left\{\Gamma\left(\frac{\nu}{2}\right)\right\} - \frac{\nu+1}{2}\log\left(1+\frac{\theta^2}{\nu}\right),\\
\nabla \log p(y, \theta)
&= -\frac{(\nu+1)\theta}{\nu+\theta^2}.
\end{align*}
The evidence lower bound for KLD is 
\begin{align*}
\L &=\E_q\left\{\log p(y, \theta)-\log q(\theta)\right\}\\
&=\log\left\{\Gamma\left(\frac{\nu+1}{2}\right)\right\}-\frac{1}{2}\log\left\{\Gamma^2\left(\frac{\nu}{2}\right)\nu\right\}\\
&\quad -\frac{\nu+1}{2}\E_q\log\left(1+\frac{\theta^2}{\nu}\right)+\frac{1}{2}\log(2 \sigma^2)+\frac{1}{2}.
\end{align*}
The Fisher divergence is
\begin{align*}
F(q\|p)&=\E_q\{\| \nabla_\theta \log p(y, \theta) -\nabla_\theta \log q(\theta)\|^2\}\\
&=\E_q\left\{\left\|-(\nu+1)\frac{\theta}{\nu+\theta^2}+\frac{\theta-\mu}{\sigma^2}  \right\|^2 \right\}\\
&= (\nu+1)^2 \E_q\left\{\frac{\theta^2}{(\nu+\theta^2)^2}\right\}+\E_q\left\{\frac{(\theta-\mu)^2}{\sigma^4}\right\} \\
&\quad - 2 \frac{\nu+1}{\sigma^2} \E_q\left\{ \frac{\theta(\theta-\mu)}{\nu+\theta^2}\right\}\\
&=(\nu+1)^2\E_q\left\{\frac{\theta^2}{(\nu+\theta^2)^2}\right\}+\frac{1}{\sigma^2}\\
&\quad -\frac{2(\nu+1)}{\sigma^2}\E_q\left\{\frac{\theta(\theta-\mu)}{\nu+\theta^2}\right\}.
\end{align*}

\subsection{Log transformed inverse Gamma}
For the log transformed inverse gamma target, recall that $a_1 = a_0 + n/2 > 1/2$ since $n \geq 1$ and $b_1 = b_0 + \sum_{i=1}^n y_i^2/2$. Then 
\begin{align*}
\log p(y,\theta)&=-\frac{n}{2}\log(2\pi)+a_0\log b_0-\log\Gamma(a_0)\\
& \quad  - a_1 \theta - b_1 \exp(-\theta),\\
\nabla_{\theta}\log p(y,\theta)&= -a_1 + b_1 \exp(-\theta).
\end{align*}
Setting $\nabla_{\theta}\log p(y,\theta) = 0$, the true posterior mode $m_* = \log(b_1/a_1)$. Since $\exp(-\theta) \mid y \sim \text{Gamma}(a_1, b_1)$, the true posterior mean and variance are given by $\mu_* = \E(\theta \mid y) = \log(b_1) - \psi(a_1)$ and $\sigma_*^2 = \Var(\theta \mid y) = \psi_1(a_1)$ respectively \citep[pg. 33,][]{Hall2024}. As $\psi(a_1) < \log(a_1)$ \citep{Alzer1997}, $\mu_* > m_*$ and the true posterior is right skewed.

First, we find the optimal variational parameters $(\hat{\mu}_{\text{KL}}, \hat{\sigma}^2_{\text{KL}})$ that maximize the evidence lower bound for the KLD. We have
\begin{align*}
\E_q\left\{\log p(y, \theta)\right\}&=-\frac{n}{2}\log(2\pi)+a_0\log b_0-\log\Gamma(a_0)\\
&\quad -a_1\mu-\exp\left(\frac{\sigma^2}{2}-\mu\right)b_1.
\end{align*}
Hence, 
\begin{align*}
\L &=\E_q\left\{\log p(y,\theta)-\log q(\theta)\right\}\\
&= \frac{1-n}{2}\log(2\pi)+a_0\log b_0-\log\Gamma(a_0) -a_1 \mu \\
&\quad - b_1 \exp\left( \frac{\sigma^2}{2} -\mu\right) + \frac{1}{2}\log(\sigma^2)+\frac{1}{2}.
\end{align*}
Setting 
\begin{align*}
\nabla_\mu\L&= -a_1+b_1 \exp(\sigma^2/2-\mu)=0,\\
\nabla_{\sigma^2}\L&= -\frac{b_1}{2} \exp(\sigma^2/2-\mu)+\frac{1}{2\sigma^2}=0.
\end{align*}
and solving simultaneously gives the global maximum at
\begin{align*}
\hat{\mu}_{\KL} & =\log\!\Big(\frac{b_1}{a_1}\Big)+\frac{1}{2a_1}, 
\quad
\hat{\sigma}^2_{\KL} = \frac{1}{a_1}.
\end{align*}

Now, we find the optimal variational parameters $(\hat{\mu}_{\text{F}}, \hat{\sigma}^2_{\text{F}})$ and $(\hat{\mu}_{\text{S}}, \hat{\sigma}^2_{\text{S}})$ that minimize the FD and SD respectively. For the FD, 
\begin{align*}
F(q\|p)&=\E_q\{\| \nabla_\theta \log p(y,\theta) -\nabla_\theta \log q(\theta)\|^2\}\\
&=\E_q\left\{\left\|-a_1 + \exp(-\theta)b_1 + \frac{\theta-\mu}{\sigma^2} \right\|^2 \right\}\\
&=a_1^2 + b_1^2 \E_q\left\{ \exp(-2\theta)  \right\}+\E_q\left\{\frac{(\theta-\mu)^2}{\sigma^4}\right\} \\
&\quad -2 a_1 b_1 \E_q\left\{\exp(-\theta) \right\}  - 2a_1 \E_q\left\{ \frac{\theta-\mu}{\sigma^2}\right\} \\
&\quad + 2 b_1 \E_q\left\{\exp(-\theta) \frac{\theta-\mu}{\sigma^2}\right\} 
\\
&= a_1^2 + b_1^2\exp (2\sigma^2 - 2\mu ) \\
&\quad - 2 b_1 (a_1 + 1) \exp (\sigma^2/2-\mu ) + 1/\sigma^2 ,
\end{align*}
since $\E_q\{ \exp(-a\theta) \} = \exp(a^2\sigma^2/2 - a\mu)$ for any constant $a \in \mathbb{R}$. It follows that 
\begin{align*}
S(q\|p) 
&= \sigma^2 F(q\|p) \\
&= \sigma^2 \{ a_1^2 + b_1^2\exp (2\sigma^2 - 2\mu ) \\
&\quad - 2 b_1 (a_1 + 1) \exp (\sigma^2/2-\mu )\} + 1, \\
\nabla_\mu S(q\|p) &= 2 b_1 \sigma^2 \exp (\sigma^2/2-\mu ) \{ a_1 + 1  \\
&\quad - b_1\exp (3\sigma^2/2 - \mu ) \}. 
\end{align*}
Note that $\nabla_\mu S(q \| p) = \sigma^2\nabla_\mu F(q\|p)$. Therefore, setting $\nabla_\mu S(q\|p) = 0$ and $\nabla_\mu F(q\|p) = 0$ both lead to the same condition,
\[
 \mu = \log \frac{ b_1}{a_1 + 1} + \frac{3\sigma^2}{2}.
\]
At this value of $\mu$,
\begin{align*}
F(\sigma^2) &= F(q\|p)|_{\mu = \hat{\mu}_F} \\
&= a_1^2 - (a_1+1)^2  \exp(-\sigma^2) + \frac{1}{\sigma^2}. \\
S(\sigma^2) &= S(q\|p)|_{\mu = \hat{\mu}_S} \\
&= a_1^2 \sigma^2 - (a_1+1)^2 \sigma^2 \exp(-\sigma^2) + 1. 
\end{align*}
Setting
\begin{align*}
F'(\sigma^2)  &= (a_1+1)^2 \exp(-\sigma^2) -\frac{1}{\sigma^4}=0, 
\\
S'(\sigma^2) &= a_1^2  + (a_1+1)^2 (\sigma^2 -1) \exp(-\sigma^2) = 0,
\end{align*}
we obtain
\begin{align*}
\hat{\sigma}^2_{\F} &=-2W_0\left(-\frac{1}{2(a_1+1)}\right),\\
\hat{\sigma}^2_{\S} &=1 - W_0\left(\frac{\e a_1^2}{(a_1+1)^2}\right),
\end{align*}
where $W_0$ is the principal branch of the Lambert W function \citep{Corless1996}. The Lambert W function yields the solution to the equation $z \exp(z) = a$, such that $z = W_0(a) \geq 0$ if $a \geq 0$, and either $z=W_0(a)\in [-1,0)$ or $z = W_{-1}(a) \leq -1$ if $-\e^{-1} \leq a < 0$. For SD, the argument $0 < {\e a_1^2}/{(a_1+1)^2} < \e$ and hence $S(\sigma^2)$ has a global minimum at $\hat{\sigma}^2_{\S} \in (0,1)$. For FD, it can be verified that $-\e^{-1}<-1/\{2(a_1 + 1)\} <0$ and hence $F(\sigma^2)$ has two stationary points, one in $(0,2)$ and the other in $(2, \infty)$. As $\lim_{\sigma^2 \rightarrow 0^+} F(\sigma^2) = +\infty$ and $\lim_{\sigma^2 \rightarrow +\infty} F(\sigma^2) = a_1^2$, the global minimum occurs in $(0,2)$ and is given by the principal branch $W_0(\cdot)$. Plots of $F(\sigma^2)$ and $S(\sigma^2)$ are given in Fig \ref{fig_loggammaFDSD}. It follows that 
\begin{align} \label{loggamma means FD SD}
\hat{\mu}_{\F} &=\log \frac{ b_1}{a_1 + 1} + \frac{3\hat{\sigma}^2_{\F}}{2}, 
\\
\hat{\mu}_{\S} &= \log \frac{ b_1}{a_1 + 1} + \frac{3\hat{\sigma}^2_{\S}}{2}. \nonumber
\end{align}

\begin{figure}[htb!]
\centering
\includegraphics[width=\linewidth]{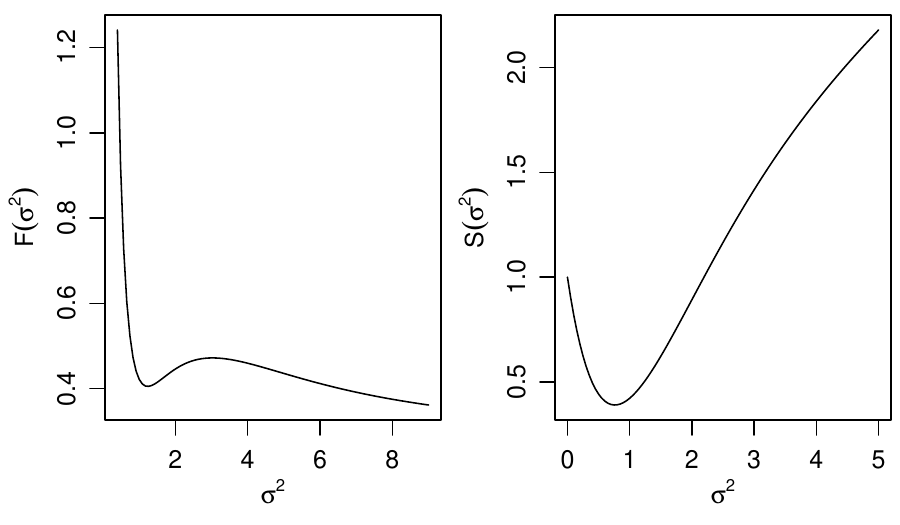}
\caption{Plots of $F(\sigma^2)$ and $S(\sigma^2)$ when $a_1=0.501$.}
\label{fig_loggammaFDSD}
\end{figure}

Before proving Theorem \ref{thm-loggamma}, we require some intermediate results, which are summarized in Lemma \ref{lem_loggamma1}.
\begin{lemma} \label{lem_loggamma1}
\begin{enumerate}[i.]
\item[]
\item Let $h(x) = x/2 - \log(x) - \log(a_1+1)$. Then $h(1/a_1) < 0$ for $a_1 > 1/2$ .
\item $\psi_1(x) > 1/x$ for $x>0$.
\item $S'(c) < 0$ where $c = (2/3) \log (1 + 1/a_1)$.
\item $(2/3) \log(1+1/a_1) + 1/(3a_1) > b$ where $b= (2/3)/(a_1+1/2) + 1/(3a_1)$.
\item $f(a_1) = \log(a_1+1) + \log(b) - b/2 > 0$ for $a_1 > 1/2$ where $b$ is defined in iv.
\end{enumerate}
\end{lemma}
\begin{proof}
For i, we want to show that $h(1/a_1) = \log(a_1) + 1/(2a_1) - \log(a_1+1)<0$. Let 
$g(x)= \log x + 1/(2x) - \log(x+1)$. 
Then $g'(x)= (x-1)/ \{2x^2(x+1) \}$. Hence $g$ is decreasing on $(1/2,1]$ and increasing on $[1,\infty)$, with $g(1/2) = -2/3$ and $\lim_{x \to +\infty} g(x)=0^-$. Thus $g(x)<0$ for all $x>\frac{1}{2}$, and hence $h(1/a_1) <0$.

For ii, we can write $\psi_1(x) = \sum_{n=0}^\infty 1/(n+x)^2$ and $1/x = \sum_{n=0}^\infty \{1/(n+x) - 1/(n+x+1)\}$. For $x>0$, 
\[
\psi_1(x) - \frac{1}{x} = \sum_{n=0}^\infty 
\frac{1}{n+x} \left( \frac{1}{n+x}  - \frac{1}{n+x+1} \right) > 0.
\]

For iii, $S'(c) = a_1^2  + (a_1+1)^2 (c -1) \exp(-c) < 0$ is equivalent to $\log(1-c) + 2c > 0$ since $\log\{a_1/(a_1+1)\} = -3c/2$. Let $k(x) = \log(1-x) + 2x$ for $0 < x < 1$. Then $k'(x) = (1-2x)/(1-x)$. Thus $k(x)$ increases from $(0,0)$, reaches a maximum at $(1/2, \log(1/2) + 1)$ and decreases to $-\infty$ with an asymptote at $x=1$. Since $0 <c<(2/3) \log(3)$, and $k(2/3\log3) = 0.147> 0$, we conclude that $k(c) > 0$ and hence $S'(c) < 0$.

For iv, we use the inequality $\log(1+x) > x/(1+x/2)$ for $x>0$ from \cite{Love1980}, which implies that $\log(1+1/a_1) > 1/(a_1+1/2)$.

For v, as $b = (a_1 + 1/6)/\{a_1(a_1+1/2)\}$
\begin{multline*}
f'(a_1) = \frac{1}{a_1+1} + \frac{1}{a_1+1/6} - \frac{1}{a_1} - \frac{1}{a_1+1/2} + \frac{1}{6a_1^2} \\
- \frac{1}{3(a_1 + 1/2)^2} = -\frac{120a_1^4 + 100a_1^3 + 6a_1^2 -5a_1 -1}{6a_1^2(a_1+1)(2a_1+1)^2 (6a_1+1)}.
\end{multline*}
It can be verified that $f'(a_1) < 0$ and hence $f(a_1)$ is strictly decreasing for $a_1 > 1/2$. Moreover $f(0.5) =\log(2) - 2/3 > 0$ and $\lim_{a_1 \to +\infty} f(a_1) = 0^+$. Hence $f(a_1) > 0$ for $a_1>1/2$.
\end{proof}

\begin{proof}[Proof of Theorem \ref{thm-loggamma}]
First, we establish the ordering for the variance parameters. At the global minimum of FD, we have $F'(\hat{\sigma}^2_{\F})=0$. As $S(\sigma^2) = \sigma^2 F(\sigma^2)$, 
\begin{equation*}
S'(\hat{\sigma}^2_{\F})=F(\hat{\sigma}^2_{\F})+\hat{\sigma}^2_{\F}F'(\hat{\sigma}^2_{\F})=F(\hat{\sigma}^2_{\F})>0.
\end{equation*}
The strict inequality above holds because $F(q\|p)=0$ if and only if $q(\theta) = p(\theta \mid y)$ almost everywhere. However, this is not true as $q(\theta)$ is symmetric while $p(\theta \mid y)$ is right skewed. Since $S(\sigma^2)$ has a global minimum and $S'(\hat{\sigma}^2_{\F})>0$, its minimum must occur strictly before $\hat{\sigma}^2_{\F}$, and hence
\[
\hat{\sigma}^2_{\S}<\hat{\sigma}^2_{\F}.
\]
Next, we want to show that $0 < \hat{\sigma}^2_{\F}<\hat{\sigma}^2_{\KL} = 1/a_1$ < 2. By setting $F'(\sigma^2) = 0$, observe that $h(\sigma^2) = \sigma^2/2 - \log(\sigma^2) - \log(a_1+1)$ is strictly decreasing and has a single root $\hat{\sigma}^2_{\F}$ in $(0,2)$. Moreover, $h(1/a_1) < 0$ from Lemma \ref{lem_loggamma1}i. Hence $\hat{\sigma}^2_{\F}<\hat{\sigma}^2_{\KL}$. Finally, $\hat{\sigma}^2_{\KL} = 1/a_1 < \psi_1(a_1) = \sigma_*^2$ from Lemma \ref{lem_loggamma1}ii. Hence,  
\[
\hat{\sigma}^2_{\S} < \hat{\sigma}^2_{\F} < \hat{\sigma}^2_{\KL} < \sigma_*^2.
\]

Next, we establish the ordering for the mean parameters. First, $\hat{\sigma}^2_{\S}<\hat{\sigma}^2_{\F}$ implies that $\hat\mu_{\S}<\hat\mu_{\F}$ from \eqref{loggamma means FD SD}. Next, 
\begin{align*}
\mu_*-\hat\mu_{\KL}
&= \{\log b_1 - \psi(a_1)\} - \left(\log \frac{b_1}{a_1} +\frac{1}{2a_1}\right)\\
&= \log a_1 - \psi(a_1) - \frac{1}{2a_1}>0
\end{align*}
for $a_1>0$, based on a result from \cite[pg 374,][]{Alzer1997}. Hence $\hat\mu_{\KL}<\mu_*$. We want to show that 
\begin{align*}
\hat\mu_{\S}- m_*
&= \left( \log \frac{ b_1}{a_1 + 1} + \frac{3\hat{\sigma}^2_{\S}}{2} \right) - \log \frac{b_1}{a_1} \\
&= \frac{3}{2}\hat{\sigma}^2_{\S} -\log\left(1+\frac1{a_1}\right) > 0.
\end{align*}
This is equivalent to showing that $\hat{\sigma}^2_{\S} > c$, where $0 < c = (2/3) \log (1 + 1/a_1) < (2/3) \log(3) \approx 0.732$. Recall that $S(\sigma^2)$ has a global minimum at $\hat{\sigma}^2_{\S}$. Thus, it suffices to show that $S'(c) < 0$, which holds from Lemma \ref{lem_loggamma1}iii.  

It remains to show that $\hat\mu_{\F}<\hat\mu_{\KL}$. We have 
\begin{align*}
\hat\mu_{\KL}-\hat\mu_{\F}
&= \left( \frac{1}{2a_1} + \log \frac{b_1}{a_1} \right) 
- \left( \log \frac{ b_1}{a_1 + 1} + \frac{3\hat{\sigma}^2_{\F}}{2} \right) \\
&= \log\left(\frac{a_1+1}{a_1}\right)+\frac{1}{2a_1} -\frac{3}{2}\hat{\sigma}^2_{\F}.
\end{align*}
Hence, our goal is to show that $\hat{\sigma}^2_{\F} < (2/3) \log(1+1/a_1) + 1/(3a_1)$. We split the proof into two parts by showing that $\hat{\sigma}^2_{\F} < b$ and $b \leq (2/3) \log(1+1/a_1) + 1/(3a_1)$, where $0 < b=(2/3)/(a_1+1/2) + 1/(3a_1) < 4/3$. The second part of the proof is given in Lemma \ref{lem_loggamma1}iv. For the first part of the proof, recall that $F(\sigma^2)$ has a single global minimum in $(0,2)$. Hence it suffices to show that $F'(b) >0$ or equivalently that $\log(a_1+1) + \log(b) - b/2 > 0$. This is true from Lemma \ref{lem_loggamma1}v.

Therefore, we have $m_*<\hat\mu_{\S}<\hat\mu_{\F}<\hat\mu_{\KL}<\mu_*$.
\end{proof}

\subsection{Skew normal}
The probability density function (pdf) for $\theta \sim \SN(m, t, \lambda)$ is
\begin{equation*} 
p(y, \theta) = 2 \, \phi(\theta|m, t^2) \, \Phi \{\lambda(\theta - m) \}.
\end{equation*}
The log-density and gradient for the skew normal are
\begin{align*} \label{skew_log}
\log p(y, \theta)&=\log 2 -\frac{1}{2}\log(2\pi t^2) -\frac{(\theta-m)^2}{2t^2}\\
&\quad +\log[\Phi \{\lambda(\theta - m) \}],\\
\nabla_{\theta}\log p(y, \theta)&=-\frac{\theta-m}{t^2}+\frac{\lambda\phi\{\lambda(\theta-m)\}}{\Phi \{\lambda(\theta - m) \}}.
\end{align*}
Taking the expectation of $\log p(y, \theta)$ with respect to $q(\theta)$, 
\begin{align*}
\E_q\left\{\log p(y, \theta)\right\}&=\log 2 -\frac{1}{2}\log(2\pi t^2) -\frac{\sigma^2+(\mu -m)^2}{2t^2}\\
&\quad +\E_q\log[\Phi \{\lambda(\theta - m) \}].
\end{align*}
For KLD, we maximize the evidence lower bound
\begin{align*}
\L &=\E_q\left\{\log p(y, \theta)-\log q(\theta)\right\}\\
&= \log 2 - \log(t) -\frac{\sigma^2+(\mu - m)^2}{2t^2} \\
&\quad + \E_q\log[\Phi \{\lambda(\theta - m) \}] +\log(\sigma)+\frac{1}{2}.
\end{align*}

The FD is given by
\begin{align*}
& F(q\|p)=\E_q\{\| \nabla_\theta \log p(y, \theta) -\nabla_\theta \log q(\theta)\|^2\}\\
&=\E_q\left\{\left\|-\frac{\theta-m}{t^2}+\frac{\lambda\phi\{\lambda(\theta-m)\}}{\Phi \{\lambda(\theta - m) \}}+\frac{\theta-\mu}{\sigma^2} \right\|^2\right\}\\
&=\E_q\left\{\frac{(\theta-m)^2}{t^4}\right\}+2\,\E_q\left[\frac{\lambda\phi\{\lambda(\theta-m)\}}{\Phi \{\lambda(\theta - m) \}}\frac{(\theta-\mu)}{\sigma^2}\right]\\
&\quad +\E_q\left\{\frac{(\theta-\mu)^2}{\sigma^4}\right\}-2\,\E_q\left[\frac{(\theta-m)}{t^2}\frac{\lambda\phi\{\lambda(\theta-m)\}}{\Phi\{\lambda(\theta-m)\}}\right]\\
&\quad -2\,\E_q\left\{\frac{(\theta-m)(\theta-\mu)}{t^2\sigma^2}\right\} 
+\E_q\left[\frac{\lambda^2\phi^2\{\lambda(\theta-m)\}}{\Phi^2 \{\lambda(\theta - m) \}}\right].
\end{align*}
After computing the 1st, 3rd and 5th terms in the final expression exactly, we obtain 
\begin{align*}
F(q\|p)&= \frac{\sigma^2+(\mu - m)^2}{t^4} + \lambda^2\E_q\left[\frac{\phi^2\{\lambda(\theta-m)\}}{\Phi^2 \{\lambda(\theta - m) \}}\right] \\
&\quad + \frac{1}{\sigma^2} -\frac{2\lambda}{t^2}\,\E_q\left[ (\theta-m)\frac{\phi\{\lambda(\theta-m)\}}{\Phi\{\lambda(\theta-m)\}}\right]-\frac{2}{t^2}\\
&\quad + \frac{2\lambda}{\sigma^2} \E_q\left[ (\theta-\mu) \frac{\phi\{\lambda(\theta-m)\}}{\Phi\{\lambda(\theta-m)\}}\right].
\end{align*}

\section{SGD based on reparametrization trick}
As $\theta=\mu+T^{-\top}z$, we have 
\begin{equation*}
\begin{aligned}
d\theta = d\mu
\end{aligned}
\quad \text{and} \quad 
\begin{aligned}
d\theta = -T^{-\top} (dT^\top) T^{-\top} z.
\end{aligned}
\end{equation*}
Recall that
\begin{align*}
g(\lambda, \theta) &= \nabla_\theta \log h(\theta) + TT^\top (\theta-\mu),\\
f(\lambda, \theta) &= T^{-1}\nabla_\theta \log h(\theta) + T^\top (\theta-\mu).
\end{align*}

Let $\text{vec}(\cdot)$ be the operator that stacks all elements of a matrix into a vector columnwise from left to right. In addition, let $K$ be the commutation matrix such that $K \vec(A) = \vec(A^\top)$, and $L$ be the elimination matrix such that $L \vec(A) = \vech(A)$ for any $d \times d$ matrix $A$, and $L^\top \vech(A) = \vec(A)$ if $A$ is lower triangular.

Differentiating $g(\lambda, \theta) $ w.r.t. $\mu$,
\begin{align*}
dg(\lambda, \theta) &= \{\nabla^2_\theta \log h(\theta)\}^\top d\theta + TT^\top (d\theta - d\mu) \\
&= \{\nabla^2_\theta \log h(\theta)\}^\top d\mu \\
\therefore \nabla_\mu  g(\lambda, \theta) &= \nabla^2_\theta \log h(\theta).
\end{align*}

Differentiating $g(\lambda, \theta) $ w.r.t. $\vech(T)$,
\begin{align*}
dg(\lambda, \theta) 
&= \{\nabla^2_\theta \log h(\theta)\}^\top d\theta + (dT)T^\top (\theta - \mu) \\
&\quad +  T(dT^\top) (\theta-\mu) + TT^\top (d\theta) \\
&= - \{\nabla^2_\theta \log h(\theta)\}^\top T^{-\top} (dT^\top) T^{-\top} z \\
& \quad + (dT)T^\top (\theta - \mu) +  T(dT^\top) (\theta-\mu) \\
&\quad  - T (dT^\top) T^{-\top} z \\
&= - \{\nabla^2_\theta \log h(\theta)\}^\top T^{-\top} (dT^\top) T^{-\top} z + (dT) z \\
&= \{ - ( z^\top T^{-1} \otimes \nabla^2_\theta \log h(\theta) ^\top  T^{-\top} ) K \\
&\quad  +(z^\top \otimes I_d) \}L^\top d\vech(T) \\
&= \{  -  ( T^{-1} \nabla^2_\theta \log h(\theta) \otimes  T^{-\top} z ) \\
&\quad +(z \otimes I_d)\}^\top L^\top d\vech(T).
\end{align*}
\begin{align*}
\therefore \nabla_{\vech(T)} g(\lambda, \theta) &= L \{ (z \otimes I_d) \\
&\quad - ( T^{-1} \nabla^2_\theta \log h(\theta)  \otimes T^{-\top} z) \}.
\end{align*}

Differentiating $f(\lambda, \theta) $ w.r.t. $\mu$,
\begin{align*}
df(\lambda, \theta) &= T^{-1} \{\nabla^2_\theta \log h(\theta)\}^\top d\theta + T^\top (d\theta - d\mu)\\
&= T^{-1} \{\nabla^2_\theta \log h(\theta)\}^\top d\mu. \\
\therefore \nabla_\mu  f(\lambda, \theta) &=  \nabla^2_\theta \log h(\theta) T^{-\top}.
\end{align*}

Differentiating $f(\lambda, \theta) $ w.r.t. $\vech(T)$,
\begin{align*}
df(\lambda, \theta) &= -T^{-1} (dT) T^{-1} \nabla_\theta \log h(\theta) +  T^\top d\theta  \\
&\quad+ (dT^\top) (\theta - \mu)  + T^{-1} \{\nabla^2_\theta \log h(\theta)\}^\top d\theta \\
&= -T^{-1} (dT) T^{-1} \nabla_\theta \log h(\theta) + (dT^\top) (\theta - \mu)  \\
&\quad - T^{-1} \{\nabla^2_\theta \log h(\theta)\}^\top T^{-\top} dT^\top T^{-\top} z \\
& \quad -  (dT^\top) T^{-\top} z \\
&= - \{ (z^\top T^{-1} \otimes T^{-1} \nabla^2_\theta \log h(\theta)^\top T^{-\top} ) K \\
&\quad + (\nabla_\theta \log h(\theta)^\top T^{-\top} \otimes T^{-1} )  \}  L^\top d\vech(T) \\
&= - \{  (T^{-1} \nabla^2_\theta \log h(\theta) T^{-\top}  \otimes T^{-\top} z ) \\
&\quad + (T^{-1} \nabla_\theta \log h(\theta) \otimes T^{-\top} )  \}^\top  L^\top d\vech(T).
\end{align*}
\begin{align*}
    \therefore \nabla_{\vech(T)} f(\lambda, \theta) &= - L\{  (T^{-1} \nabla^2_\theta \log h(\theta) T^{-\top}  \otimes T^{-\top} z ) \\
    &\quad + (T^{-1} \nabla_\theta \log h(\theta) \otimes T^{-\top} ) \}.
\end{align*}

Differentiating 
$$F(\lambda) = \E_\phi \left\{ g(\lambda, \mu+T^{-\top}z)^\top g(\lambda, \mu+T^{-\top}z) \right\}$$ with respect to $\mu$,
\begin{align*}
d F(\lambda) &= \E_\phi \left[ 2 g(\lambda, \theta)^\top dg(\lambda, \theta)  \right] \\
& = \E_\phi \left[ 2 g(\lambda, \theta)^\top  \{\nabla^2_\theta \log h(\theta)\}^\top d\mu \right]. \\ 
\therefore \nabla_\mu F(\lambda) & = 2\E_\phi \left[ \{\nabla^2_\theta \log h(\theta)\} g(\lambda, \theta) \right] .
\end{align*}

Next we differentiate $F(\lambda)$ with respect to $\vech(T)$.  
\begin{align*}
d F(\lambda) &= \E_\phi \left[ 2 g(\lambda, \theta)^\top dg(\lambda, \theta)  \right] \\
&= 2 \E_\phi \left[  g(\lambda, \theta)^\top \{ -  ( T^{-1} \nabla^2_\theta \log h(\theta) \otimes  T^{-\top} z )  \right.\\
& \quad \left. + (z \otimes I_d)  \}^\top L^\top d\vech(T)  \right],
\end{align*}
\begin{align*}
\nabla_{\vech(T)} & F(\lambda) 
= 2 L \E_\phi \left[ \{ -  ( T^{-1} \nabla^2_\theta \log h(\theta) \otimes  T^{-\top} z ) \right.\\ 
&\quad \left. + (z \otimes I_d) \} g(\lambda, \theta) \right] \\
& =  2  \E_\phi \vech\{  - T^{-\top} z g(\lambda, \theta)^\top \nabla_\theta^2 \log h(\theta)  T^{-\top}\\
&\quad + g(\lambda, \theta) z^\top\}.
\end{align*}

Differentiating $S(\lambda)$ with respect to $\mu$,
\begin{align*}
d S(\lambda) &= \E_\phi \left[ \{ 2 f(\lambda, \theta)\}^\top df(\lambda, \theta) \right]
\\ 
&= \E_\phi \left[ \{ 2 f(\lambda, \theta)\}^\top T^{-1} \{\nabla^2_\theta \log h(\theta)\}^\top d\mu \right],
\\ 
\therefore \nabla_\mu S(\lambda)& = 2 \E_\phi  \{\nabla_\theta^2 \log h(\theta) T^{-T} f(\lambda, \theta)\}.
\end{align*}

Differentiating $S(\lambda)$ with respect to $\vech(T)$, 
\begin{align*}
d &S(\lambda) = \E_\phi \left[ 2 f(\lambda, \theta)^\top df(\lambda, \theta) \right]
\\ 
&= - 2 \E_\phi [ f(\lambda, \theta)^\top  \{  (T^{-1} \nabla^2_\theta \log h(\theta) T^{-\top}  \otimes T^{-\top} z ) \\
&\quad + (T^{-1} \nabla_\theta \log h(\theta) \otimes T^{-\top} )  \}^\top L^\top d\vech(T) ],
\end{align*}
\begin{align*}
& \nabla_{\vech(T)} S(\lambda)  \\
&= - 2 L \E_\phi \left[ \{  (T^{-1} \nabla^2_\theta \log h(\theta) T^{-\top}  \otimes T^{-\top} z ) \right.\\
&\quad \left. + (T^{-1} \nabla_\theta \log h(\theta) \otimes T^{-\top} )  \} f(\lambda, \theta) \right] \\
& = - 2  \E_\phi \vech\{ T^{-\top}f(\lambda, \theta) \nabla_\theta \log h(\theta)^{\top}T^{-\top} \\
&\quad + T^{-\top} z f(\lambda, \theta)^\top T^{-1}\nabla_\theta^2 \log h(\theta)  T^{-\top} \}.
\end{align*}

\subsection{Variance of gradient estimates}
We have
\begin{align*}
g_{\mu}^{\KL} &= \nabla_{\theta}\log h(\theta) + Tz = -\Lambda (\theta - \nu ) + Tz \\
&= -\Lambda (T^{-\top} z + \mu - \nu ) + Tz \\
&= (T-\Lambda T^{-\top})z - \Lambda(\mu-\nu),\\
g_{\mu}^{\text{F}} 
&= - 2\nabla_{\theta}^2\log h(\theta)g_{\mu}^{\KL} 
= - 2 (-\Lambda) g_{\mu}^{\KL} 
= 2\Lambda g_{\mu}^{\KL},\\
g_{\mu}^{\text{S}} &= - 2\nabla_{\theta}^2\log h(\theta)T^{-\top}T^{-1}g_{\mu}^{\KL} \\
&= 2 \Lambda T^{-\top}T^{-1} g_{\mu}^{\KL}, \\
g_T^{\KL} &= T^{-\top} z (\mu -\nu)^\top \Lambda T^{-\top} \\
& \quad -  T^{-\top} zz^\top (T^\top - T^{-1} \Lambda) T^{-\top} ,
\\
g_T^{\text{F}} &= 2\{ \Lambda  (\mu -\nu) z^\top + T^{-\top} z (\mu -\nu)^\top \Lambda^2 T^{-\top}   \\
& \quad - (T-\Lambda T^{-\top})zz^\top \\
& \quad - T^{-\top} zz^\top (T^\top - T^{-1} \Lambda)\Lambda T^{-\top} \},
\end{align*}
\begin{align*}
g_T^{\text{S}} 
&= 2[ -\Sigma  (T-\Lambda T^{-\top}) \{z z^\top T^{-1}  + z   (\mu -\nu)^\top\} \\
& \times  \Lambda T^{-\top} + \Sigma \Lambda (\mu -\nu) \{z^\top T^{-1}  +  (\mu -\nu) ^\top\}  \Lambda T^{-\top} \\
& + T^{-\top} \{z  (\mu -\nu)^\top \Lambda - zz^\top (T^\top - T^{-1} \Lambda) \} \Sigma \Lambda T^{-\top} ],
\end{align*}
and 
\begin{align*}
\Var(g_{\mu}^{\KL})&= (T-\Lambda T^{-\top})(T^{\top} - T^{-1}\Lambda) \\
&= \Sigma^{-1} - 2\Lambda + \Lambda \Sigma \Lambda,\\
\Var(g_{\mu}^{\text{F}})&= 4\Lambda \Var(g_{\mu}^{\KL}) \Lambda,\\
\Var(g_{\mu}^{\text{S}})& = 4\Lambda T^{-\top}T^{-1} \Var(g_{\mu}^{\KL}) T^{-\top}T^{-1}\Lambda.
\end{align*}

To simplify the derivation of the variance with respect to $T_{ii}$, we further assume that both $\Lambda$ and $T$ are diagonal matrices. Under this assumption, the gradient terms can be expressed as
\begin{align*}
g_{T_{ii}}^{\KL} &= \frac{\Lambda_{ii}(\mu_i-\nu_i)}{T_{ii}^2}z_i+\left(-\frac{1}{T_{ii}} + \frac{\Lambda_{ii}}{T_{ii}^3}\right)z_i^2,\\
g_{T_{ii}}^{\text{F}} &= 2\left(\Lambda_{ii} + \frac{\Lambda_{ii}^2}{T_{ii}^2}\right)(\mu_i-\nu_i)z_i + 2\left(-T_{ii} + \frac{\Lambda_{ii}^2}{T_{ii}^3}\right)z_i^2\\
&= 2 (T_{ii}^2 + \Lambda_{ii}) g_{T_{ii}}^{\KL},\\
g_{T_{ii}}^{\text{S}} &= 2 \frac{\Lambda_{ii}^2(\mu_i - \nu_i)^2}{T_{ii}^3} + 2 z_i (\mu_i-\nu_i) \left( -\frac{\Lambda_{ii}}{T_{ii}^2} + 3\frac{\Lambda_{ii}^2}{T_{ii}^4} \right) \\
&\quad + 4z_i^2 \left( - \frac{\Lambda_{ii}}{T_{ii}^3} + \frac{\Lambda_{ii}^2}{T_{ii}^5}\right).
\end{align*}

Utilizing the properties $\Var(z_i) = 1$, $\Var(z_i^2) = 2$ and $\cov(z_i,z_i^2)=0$, we obtain 
\begin{align*}
\Var(g_{T_{ii}}^{\KL}) &= \frac{1}{T_{ii}^4} \left\{ \Lambda_{ii}^2 (\mu_i - \nu_i)^2 + 2 \left( T_{ii} - \frac{\Lambda_{ii}}{T_{ii}} \right)^2 \right\},\\
\Var(g_{T_{ii}}^{\text{F}}) &= 4 (T_{ii}^2 + \Lambda_{ii})^2\Var(g_{T_{ii}}^{\KL}),\\
\Var(g_{T_{ii}}^{\text{S}})
&= \frac{4\Lambda_{ii}^2}{T_{ii}^8}  \left\{ \left( 3\Lambda_{ii}  - T_{ii}^2 \right)^2  (\mu_i - \nu_i)^2 \right. \\
& \quad \left. + 8 \left( T_{ii} - \frac{\Lambda_{ii}}{T_{ii}} \right)^2 \right\}.
\end{align*}

\section{SGD based on batch approximation} 
We have
\begin{align*}
& \hat{S}_{q_t}(\lambda) = \frac{1}{B} \sum_{b=1}^B \{g_h(\theta_i)^\top \Sigma g_h(\theta_i) + 2 g_h(\theta_i)^\top (\theta_i - \mu) \\
&\quad + (\theta_i - \mu) ^\top \Sigma^{-1} (\theta_i - \mu) \} \\
&= \frac{1}{B} \sum_{b=1}^B  [ \{g_h(\theta_i) - \overline{g}_h  + \overline{g}_h\}^\top \Sigma \{g_h(\theta_i) - \overline{g}_h  + \overline{g}_h\} \\
&\quad + 2 \{g_h(\theta_i) - \overline{g}_h  + \overline{g}_h\}^\top (\theta_i - \overline{\theta} + \overline{\theta} - \mu) \\
& \quad + (\theta_i - \overline{\theta} + \overline{\theta}- \mu) ^\top \Sigma^{-1} (\theta_i - \overline{\theta} + \overline{\theta}- \mu) ] \\
&= \tr\{ (C_g + \overline{g}_h \overline{g}_h^\top) \Sigma \} + \tr(C_\theta \Sigma^{-1}) \\
&\quad +  (\mu - \overline{\theta})^\top \Sigma^{-1} (\mu - \overline{\theta}) - 2 \overline{g}_h^\top (\mu - \overline{\theta}) \\
&\quad  +  \frac{2}{B} \sum_{b=1}^B \{g_h(\theta_i) - \overline{g}_h \}^\top (\theta_i - \overline{\theta})\\
&= \tr(V \Sigma) + \tr(U \Sigma^{-1}) +  2 \tr(W), 
\end{align*} 
\begin{align*}
&\hat{F}_{q_t}(\lambda) 
= \frac{1}{B} \sum_{b=1}^B \{  2 g_h(\theta_i)^\top \Sigma^{-1}(\theta_i - \mu)\\
&\quad +(\theta_i - \mu) ^\top \Sigma^{-2} (\theta_i - \mu) \} + g_h(\theta_i)^\top g_h(\theta_i) \\
&= \frac{1}{B} \sum_{b=1}^B  [ \{g_h(\theta_i)- \overline{g}_h + \overline{g}_h\} ^\top \{g_h(\theta_i) - \overline{g}_h + \overline{g}_h\}\\
&\quad + 2 \{g_h(\theta_i) - \overline{g}_h  + \overline{g}_h\}^\top \Sigma^{-1} (\theta_i - \overline{\theta} + \overline{\theta}- \mu)\\
&\quad +(\theta_i - \overline{\theta} + \overline{\theta}- \mu) ^\top \Sigma^{-2} (\theta_i - \overline{\theta} + \overline{\theta}- \mu) ] \\
&= \tr(C_g + \overline{g}_h \overline{g}_h^\top) + 2\tr(C_{\theta g}\Sigma^{-1})+\tr(C_\theta\Sigma^{-2}) \\
&\quad + (\mu - \overline{\theta})^\top \Sigma^{-2} (\mu - \overline{\theta}) - 2 \overline{g}_h^\top\Sigma^{-1} (\mu - \overline{\theta}) \\
&= \tr(V) + \tr(U \Sigma^{-2}) + 2\tr(W\Sigma^{-1}),
\end{align*}
where $U = C_\theta + (\mu - \overline{\theta}) (\mu - \overline{\theta})^\top$, $V= C_g + \overline{g}_h \overline{g}_h^\top$ and $W= C_{\theta g}  -(\mu - \overline{\theta})\overline{g}_h^T$. Note that $U$ and $V$ are symmetric but $W$ is not. Differentiating with respect to $\mu$ and $T$, 
\begin{align*}
\nabla_\mu \hat{S}_{q_t}(\lambda) &= 2\Sigma^{-1} (\mu - \overline{\theta}) - 2 \overline{g}_h. 
\\
d\hat{S}_{q_t}(\lambda) &=  d\{ \tr(VT^{-\top} T^{-1}) + \tr(U T T^\top) \} \\
&= - \tr(V T^{-\top} dT^\top \Sigma) - \tr (V \Sigma dT T^{-1}) \\
&\quad + \tr(U dT T^\top) + \tr(U T dT^\top) \\
&= \tr\{ (UT - \Sigma VT^{-\top}) dT^\top \} \\
&\quad + \tr \{ (T^\top U - T^{-1} V \Sigma) dT \} \\
&= 2 \vec(UT - \Sigma VT^{-\top})^\top L^\top d\vech(T). \\
\nabla_{\vech(T)} & \hat{S}_{q_t}(\lambda) = 2\vech( U T  - \Sigma V T^{-\top} ).\\
\nabla_\mu \hat{F}_{q_t}(\lambda) &= \Sigma^{-1} \nabla_\mu \hat{S}_{q_t}(\lambda) . \\
d \hat{F}_{q_t}(\lambda) &=  d\{ \tr(U T T^\top T T^\top) + 2\tr(W T T^\top) \} \\
&= \tr(U dT T^\top \Sigma^{-1}) +  \tr(U T dT^\top \Sigma^{-1}) \\
&\quad +  \tr(U \Sigma^{-1} dT T^\top) +  \tr(U \Sigma^{-1} T dT^\top) \\
& \quad + 2\tr(W dT T^\top) + 2\tr(W T dT^\top) \\
&= \tr\{ (T^\top \Sigma^{-1}U  + T^\top U \Sigma^{-1})dT\} \\
&\quad +  \tr\{ (\Sigma^{-1}U T + U \Sigma^{-1} T)  dT^\top\}  \\
& \quad + 2\tr(T^\top W dT ) + 2\tr(W T dT^\top) \\
&= 2 \vec (\Sigma^{-1}U T + U \Sigma^{-1} T + W^\top T \\
&\quad + WT)^\top  L^\top d\vech(T). \\
\nabla_{\vech(T)} \hat{F}_{q_t}(\lambda) &= 2\vech \{( W + W^\top  + \Sigma^{-1} U  \\
& \quad +U \Sigma^{-1}) T\}.
\end{align*}

\section{Proof of Theorem \ref{thm-convSDb}}
Let $\|x\| = \sqrt{x^\top x}$ for $x \in \mathbb{R}^d$ and $\|A\|$ denote the spectral norm of a matrix $A \in \mathbb{R}^{d \times d}$, which is evaluated as the square root of the largest eigenvalue of $A^\top A$. Let $A\succ0$ and $A \succeq 0$ denote that $A$ is positive definite and positive semidefinite respectively, and $A \succeq B$ denote that the matrix $A-B$ is positive semidefinite. In addition, let $\tau_k(\cdot)$, $\tau_{\min}(\cdot)$ and $\tau_{\max}(\cdot)$ denote the $k$th, minimum and maximum eigenvalue of a given matrix respectively. 

First, differentiating $\hat{S}_{q_t}(\lambda)$ with respect to $\vec(\Sigma)$,
\begin{align*}
& d\hat{S}_{q_t}(\lambda) = \tr(V d\Sigma) - \tr(U \Sigma^{-1} d\Sigma \Sigma^{-1}) 
\\
&\implies \nabla_{\vec(\Sigma)} \hat{S}_{q_t}(\lambda) = \vec(V - \Sigma^{-1} U \Sigma^{-1}).
\end{align*}

Suppose the target is $p(\theta|y) = \N(\nu, \Lambda^{-1})$ and the variational density at iteration $t$ is $q_t (\theta) = \N(\theta \mid \mu_t, \Sigma_t)$. Assuming the batchsize $B \to \infty$, from Lemma \ref{Lem3}, 
\begin{equation*} 
\begin{gathered}
\overline{\theta} \xrightarrow{\text{a.s.}} \mu_t, \quad
C_\theta  \xrightarrow{\text{a.s.}} \Sigma_t, \quad 
\overline{g}_h  \xrightarrow{\text{a.s.}} \Lambda(\nu-\mu_t), 
\\
C_g  \xrightarrow{\text{a.s.}} \Lambda \Sigma_t \Lambda, \quad
C_{\theta g}  \xrightarrow{\text{a.s.}} -\Sigma_t \Lambda,
\end{gathered} 
\end{equation*}
which implies that $U \to \Sigma_t$, $V \to \Lambda \{\Sigma_t + (\nu-\mu_t)(\nu-\mu_t)^\top\} \Lambda$ and $W \to - \Sigma_t \Lambda$.

Consider the updates for $(\mu_t, \Sigma_t)$ at iteration $t$ based on natural gradients as given in Table 1 of \cite{Tan2024a}, and let $B \to \infty$. Note the change in signs below, as the updates in \cite{Tan2021} are for maximizing the lower bound, while we are minimizing $\hat{S}_{q_t} (\lambda_t)$ here. Let $0 < \rho_t < 1/4$ denote the stepsize at iteration $t$. We assume that the stepsize is decreasing, so that $\rho_{t+1} \leq \rho_t$ $\forall t$. We have 
\begin{align*}
& \Sigma_{t+1}^{-1} = \Sigma_t^{-1} + 2 \rho_t \nabla_{\Sigma} \hat{S}_{q_t} (\lambda_t) \\
&= \Sigma_t^{-1} + 2\rho_t (V - \Sigma_t^{-1} U \Sigma_t^{-1}) \\
&\to (1-2 \rho_t) \Sigma_t^{-1} + 2\rho_t \Lambda\{ \Sigma_t + (\nu-\mu_t)(\nu-\mu_t)^\top \} \Lambda,
\\
&\mu_{t+1} = \mu_t - \rho_t \Sigma_{t+1} \nabla_\mu \hat{S}_{q_t} (\lambda_t) \\
&= \mu_t - 2\rho_t  \Sigma_{t+1} \{ \Sigma_t^{-1} (\mu_t - \overline{\theta}) - \overline{g}_h\} \\
& \to \mu_t - 2 \rho_t \Sigma_{t+1} \Lambda(\mu_t - \nu).
\end{align*}
Let $1/2 < \beta_t = 1 - 2\rho_t < 1$ and introduce
\begin{align*}
J_t &= \Lambda^{-1/2} \Sigma_t^{-1} \Lambda^{-1/2},
\\
\epsilon_t &= \Lambda^{1/2}(\mu_t - \nu), 
\\
\Delta_t &= J_t - I_d.
\end{align*}
Note that $\beta_{t+1} \geq \beta_t$ $\forall t$ since $\{\rho_t\}$ is decreasing. Next, we multiply the update of $\Sigma_{t+1}^{-1}$ by $\Lambda^{-1/2}$ on the left and right. As for the update of $\mu_{t+1}$, we first subtract $\nu$ from both sides and then multiply by $\Lambda^{1/2}$ on the left. This gives
\begin{align*}
J_{t+1} &= \beta_t J_t + (1 - \beta_t) (J_t^{-1} + \epsilon_t \epsilon_t ^\top),
\\
\epsilon_{t+1} &= \{I_d - (1-\beta_t) J_{t+1}^{-1}\} \epsilon_t.
\end{align*}
Our goal is to show that $\| \Delta_t\| \to 0$ and $\|\epsilon_t\| \to 0$ as $t \to \infty$ as this will imply that $\mu_t \to \nu$ and $\Sigma_t^{-1} \to \Lambda$.

As the eigenvalues of $J_{t+1}$ are not computable directly, we introduce
\begin{align*}
K_{t+1} &= \beta_t J_t + (1-\beta_t) J_t^{-1}, 
\\
H_{t+1} &= \beta_t J_t + (1-\beta_t) (J_t^{-1} + \| \epsilon_t \|^2 I_d),
\end{align*}
to bound them. Note that $K_{t+1} \preceq J_{t+1} \preceq H_{t+1}$, since  
\begin{align*}
x^\top (J_{t+1} - K_{t+1})x &= (1-\beta_t) (x^\top \epsilon_t)^2 \geq 0,
\\
x^\top (H_{t+1} - J_{t+1})x &= (1-\beta_t) \{ \| \epsilon_t \|^2 \|x \|^2  \\
&\quad  - (x^\top \epsilon_t)^2 \} \geq 0, 
\quad\forall x \in \mathbb{R}^d.
\end{align*}

We assume that the initial $\Sigma_0^{-1}$ and hence $J_0$ to be positive definite. Given that $J_t \succ 0$,
\begin{multline*}
x^\top J_{t+1} x = \beta_t x^\top J_t x 
\\
+ (1 - \beta_t) \{x^\top J_t^{-1} x +  (\epsilon_t ^\top x)^2 \} > 0.
\end{multline*}
Hence $\{J_t\}_{t=0}^\infty$ is positive definite. By a similar reasoning, $\{H_t\}_{t=1}^\infty$ and $\{K_t\}_{t=1}^\infty$ are also positive definite. Let $J_t = Q D_{J_t} Q^\top$ be an eigendecomposition of $J_t$, where $Q$ is an orthogonal matrix containing the normalized eigenvectors of $J_t$, and $D_{J_t}$ is a diagonal matrix containing the eigenvalues of $J_t$ in increasing order. Since
\begin{align} \label{eg of H and K}
Q^\top K_{t+1} Q &= \beta_t D_{J_t} + (1-\beta_t) D_{J_t}^{-1}, 
\\
Q^\top H_{t+1} Q &= \beta_t D_{J_t} + (1-\beta_t) (D_{J_t}^{-1} + \| \epsilon_t \|^2 I_d), \nonumber
\end{align}
it follows that $K_{t+1}$ and $H_{t+1}$ have the same eigenvectors as $J_t$ and their eigenvalues are contained in the diagonal elements of the matrices on the RHS. Specifically,
\begin{align*}
\tau_k(K_{t+1}) &= \beta_t \tau_k(J_t) +  \frac{1-\beta_t}{\tau_k(J_t)} \;\; \forall k, 
\\
\tau_k(H_{t+1}) &= \beta_t \tau_k(J_t) +  (1-\beta_t) \left( \frac{1}{\tau_k(J_t)} + \| \epsilon_t \|^2 \right) \;\; \forall k.
\end{align*}
Next, we study the properties of the eigenvalues of $K_{t+1}$ and $H_{t+1}$ more closely through Lemma \ref{lem-fgcurve} and \ref{lem-SDhat}.

\begin{figure}[tb!]
\centering
\includegraphics[width=\linewidth]{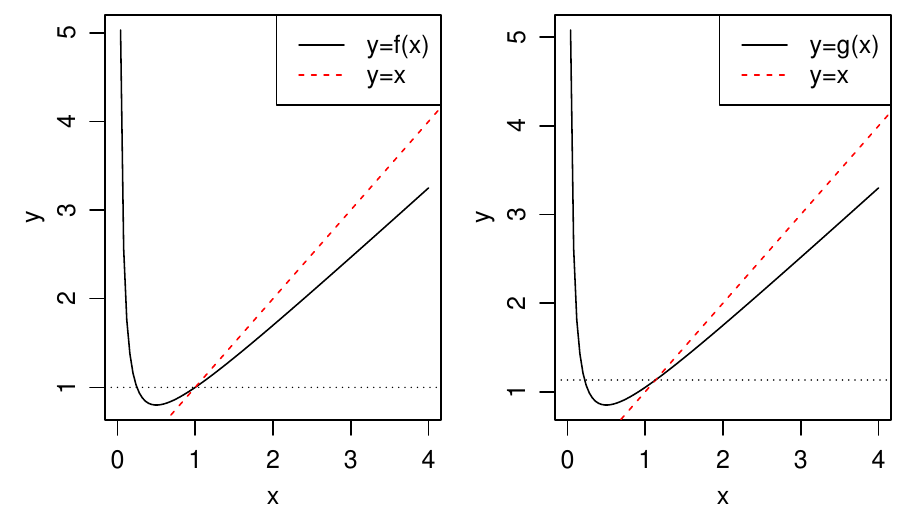}
\caption{Plots of $y=f(x)$ and $y=g(x)$ from which eigenvalues of $K_{t+1}$ and $H_{t+1}$ are derived respectively.}
\label{fig_egHK}
\end{figure}

\begin{lemma} \label{lem-fgcurve}
Let $f(x)=\beta_tx + (1-\beta_t)/x$ and $g(x)=f(x) + (1-\beta_t)\|\epsilon_t\|^2$ for $x>0$ and $1/2 < \beta_t <1$. 
\begin{enumerate}[i.] 
\item Then $y=f(x)$ has a global minimum at $(\sqrt{(1-\beta_t)/\beta_t}, 2\sqrt{\beta_t(1-\beta_t)})$ and $f(x)$ is strictly increasing on $(\sqrt{(1-\beta_t)/\beta_t}, \infty)$. The line $y=1$ cuts this curve at $x=(1 - \beta_t)/\beta_t < 1$ and $x=1$.
\item Then $y=g(x)$ has a global minimum at $(\sqrt{(1-\beta_t)/\beta_t}, 2\sqrt{\beta_t(1-\beta_t)} + (1-\beta_t)\|\epsilon_t\|^2)$ and $g(x)$ is strictly increasing on $(\sqrt{(1-\beta_t)/\beta_t}, \infty)$. The line $y=x$ cuts this curve at $x=\tilde{\epsilon}_t$, where $\tilde{\epsilon}_t = \{ \|\epsilon_t\|^2 + \sqrt{\|\epsilon_t\|^4 + 4} \}/2$.
\end{enumerate}
\end{lemma}

\begin{proof}
For i, setting 
\[
dy/dx = \beta_t - (1-\beta_t)/x^2 = 0
\]
leads to  $x=\sqrt{(1-\beta_t)/\beta_t}$ and $d^2y/dx^2 = 2(1-\beta_t)/x^3>0$. Hence there is a minimum point at $x=\sqrt{(1-\beta_t)/\beta_t}$ and $y=2\sqrt{\beta_t(1-\beta_t)}$. Solving $\beta_tx + (1-\beta_t)/x=1$ leads to the  equation $\beta_t x^2 - x +1-\beta_t=0$, which has two roots $x=(1-\beta_t)/\beta_t < 1$ and $x=1$.

The result in ii follows directly from i as $y=g(x)$ is just $y=f(x)$ translated vertically upwards by $(1-\beta_t)\|\epsilon_t\|^2$. Solving $\beta_t x + (1-\beta_t)/x  + (1-\beta_t)\|\epsilon_t\|^2 = x$ leads to the equation $x^2 - \|\epsilon_t\|^2 x - 1 = 0$, which has only one positive root at $x=\tilde{\epsilon}_t$.
\end{proof}

\begin{lemma} \label{lem-SDhat}
Suppose $1/2 < \beta_t <1$.
\begin{enumerate}[i.] 
\item $\tau_k(J_t) > \sqrt{(1-\beta_t)/\beta_t}$ $\forall$ k and $t\geq 1$.
\item $\tau_k(H_{t+1}) \leq \max\{ \tilde{\epsilon}_t, \tau_k(J_t) \}$ $\forall k$ and $t \geq 1$.
\item $| \tau_k(K_{t+1}) - 1 | \leq \beta_t| \tau_k(J_t) - 1 |$ $\forall k$ and $t \geq 1$.
\item $\|K_{t+1} - I_d\| \leq \beta_t \|\Delta_t\|$ for $t \geq 1$.
\item $\tau_{\max}(J_{t+1}^{-1}) \leq \frac{1}{2\sqrt{\beta_t(1-\beta_t)}}$ $\forall t \geq 0$.
\end{enumerate}
\end{lemma}

\begin{proof}
For i, since $J_{t+1} \succeq K_{t+1}$, 
\begin{align*}
\tau_k(J_{t+1}) &\geq \tau_{\min}(K_{t+1}) \\
&\geq 2\sqrt{\beta_{t} (1-\beta_{t})} \\
&\geq 2\sqrt{\beta_{t+1} (1-\beta_{t+1})}, 
\quad \forall k, t\geq0.
\end{align*}
The second line follows from Lemma \ref{lem-fgcurve} i, since the minimum value of $y=f(x)$ is $2\sqrt{\beta_{t} (1-\beta_{t})}$. The third line is because $y=2\sqrt{x(1-x)}$ is decreasing on $[1/2,1]$ and $\{\beta_t\}$ is increasing. Thus $\tau_k(J_{t}) \geq 2\sqrt{\beta_t (1-\beta_t)}$ $\forall k$ and $t\geq 1$. It suffices to show that $2\sqrt{\beta_t (1-\beta_t)} > \sqrt{(1-\beta_t)/\beta_t}$, which is just equivalent to $\beta_t > 1/2$.

For ii, we have $\tau_k(J_t) > \sqrt{(1-\beta_t)/\beta_t}$ $\forall$ k and $t\geq 1$ from i. Hence, to obtain the eigenvalues of $H_{t+1}$ for $t\geq 1$, we only need to consider the curve of $y=g(x)$ for $x>\sqrt{(1-\beta_t)/\beta_t}$ in Figure \ref{fig_egHK}, which is strictly increasing. Moreover, from Lemma \ref{lem-fgcurve} ii, the line $y=x$ cuts $y=g(x)$ at $x=\tilde{\epsilon}_t$. This implies that if $\tau_k(J_t)\leq \tilde{\epsilon}_t$, then $\tau_k(H_{t+1}) \leq \tilde{\epsilon}_t$. If $\tau_k(J_t) > \tilde{\epsilon}_t$, then $\tau_k(H_{t+1}) < \tau_k(J_t)$ because $y=g(x)$ lies below $y=x$. Hence $\tau_k(H_{t+1}) \leq \max\{ \tilde{\epsilon}_t, \tau_k(J_t) \}$ $\forall k$ and $t \geq 1$.

For iii, if $\tau_k(J_t) = 1$, then $\tau_k(K_{t+1}) = 1$ and the inequality is trivially satisfied. Hence it suffices to consider $\tau_k(J_t) \neq 1$ and only $\tau_k(J_t) > \sqrt{(1-\beta_t)/\beta_t}$ from i. First, suppose $\sqrt{(1-\beta_t)/\beta_t} < \tau_k(J_t) < 1$. Then $\tau_k(K_{t+1}) < 1$ from Lemma \ref{lem-fgcurve} i and
\begin{align*}
&\beta_t| \tau_k(J_t) - 1 | - | \tau_k(K_{t+1}) - 1 | \\
& = \beta_t(1 - \tau_k(J_t)) - (1 - \tau_k(K_{t+1})) \\
&= \beta_t - \beta_t\tau_k(J_t) - 1 + \beta_t \tau_k(J_t) + \frac{1-\beta_t}{\tau_k(J_t)} \\
&= (1-\beta_t) \left( \frac{1}{\tau_k(J_t)}-1 \right) > 0.
\end{align*}
Next, suppose $\tau_k(J_t) > 1$. Then $\tau_k(K_{t+1}) >1$ and
\begin{align*}
&\beta_t| \tau_k(J_t) - 1 | - | \tau_k(K_{t+1}) - 1 | \\
& = \beta_t(\tau_k(J_t) - 1) - (\tau_k(K_{t+1}) - 1) \\
&= \beta_t\tau_k(J_t)  - \beta_t - \beta_t \tau_k(J_t) -\frac{1-\beta_t}{\tau_k(J_t)} +1\\
&= (1- \beta_t)\left( 1 -\frac{1}{\tau_k(J_t)}\right) > 0.
\end{align*}

For iv, we have from iii,
\begin{align*}
\|K_{t+1} - I_d\| 
& = \max_k | \tau_k(K_{t+1}) - 1 | \\
& \leq \beta_t \max_k  | \tau_k(J_t) - 1 | \\
& = \beta_t \|J_t - I_d\| = \beta_t \|\Delta_t\|.
\end{align*}

For v, from Lemma \ref{lem-fgcurve} i, $\forall k$ and $t \geq0$,
\begin{equation*}
\begin{aligned}
\tau_k(J_{t+1}) \geq \tau_{\min} (K_{t+1}) &\geq 2\sqrt{\beta_t(1-\beta_t)}
\\
\implies \tau_k(J_{t+1}^{-1}) &\leq \frac{1}{2\sqrt{\beta_t(1-\beta_t)}}
\\
\implies \tau_{\max}(J_{t+1}^{-1}) &\leq \frac{1}{2\sqrt{\beta_t(1-\beta_t)}}.
\end{aligned}
\end{equation*}
\end{proof}

From Lemma \ref{lem-SDhat} i, $\tau_k(J_t) > \sqrt{(1-\beta_t)/\beta_t}$ $\forall k$, $t\geq 1$. Hence we only need to consider the curves of $y=f(x)$ and $y=g(x)$ for $x> \sqrt{(1-\beta_t)/\beta_t}$ to obtain the eigenvalues of $K_{t+1}$ and $H_{t+1}$ for $t \geq 1$, which are strictly increasing from Lemma \ref{lem-fgcurve} i and ii. Hence the eigenvalues of $K_{t+1}$ and $H_{t+1}$ are also arranged in increasing order in \eqref{eg of H and K}. Let the eigenvalues of $J_{t+1}$ be arranged in increasing order as well. Since $K_{t+1} \preceq J_{t+1} \preceq H_{t+1}$, we have 
\begin{align} \label{ineq for all k}
\tau_k (K_{t+1}) \leq \tau_k(J_{t+1})  \leq \tau_k (H_{t+1})
\quad \forall k, t\geq 1.
\end{align}

Next, we will establish upper bounds for $\|\epsilon_t\|$ and $\|\Delta_t\|$. As $\epsilon_{t+1} = \{I_d - (1-\beta_t) J_{t+1}^{-1}\} \epsilon_t$, from the sub-multiplicative property of the spectral norm, we have
\begin{align*}
\| \epsilon_{t+1} \| \leq \|I_d - (1-\beta_t)  J_{t+1}^{-1} \| \|\epsilon_t\|.
\end{align*}
From Lemma \ref{lem-SDhat} v, $J_{t+1}^{-1} \preceq I_d/\{2\sqrt{\beta_t(1-\beta_t)}\}$. Hence 
\begin{align*}
I_d - (1-\beta_t)  J_{t+1}^{-1} \succeq \left(1 - \frac{(1-\beta_t) }{2\sqrt{\beta_t(1-\beta_t)}} \right) I_d \succ 0.
\end{align*}
Thus $\|I_d - (1-\beta_t)  J_{t+1}^{-1}\| = 1 - (1-\beta_t) \tau_{\min}(J_{t+1}^{-1})$ and 
\begin{align} \label{epsilon bd}
\| \epsilon_{t+1} \| &\leq \{ 1 - (1-\beta_t) \tau_{\min}(J_{t+1}^{-1}) \}  \|\epsilon_t\| \quad \forall t \geq 0.
\end{align}
As for $\|\Delta_t\|$, applying the triangle inequalities and Lemma \ref{lem-SDhat} iv,
\begin{align} \label{Delta bd}
\|\Delta_{t+1}\| &= \| J_{t+1} - I_d \| \\
& \leq \| J_{t+1} - K_{t+1} \| + \| K_{t+1} - I_d \|  \nonumber \\
& \leq \| H_{t+1} - K_{t+1} \| + \| K_{t+1} - I_d \| 
\nonumber \\
& = (1-\beta_t) \|\epsilon_t\|^2 + \| K_{t+1} - I_d \|
\quad \forall t \geq 0 
\nonumber \\
& \leq (1-\beta_t) \|\epsilon_t\|^2 + \beta_t \|\Delta_t\| \quad \forall t \geq 1. 
\nonumber
\end{align}

Next, we present Lemma \ref{lem-taumin}, which is useful in bounding $\|\epsilon_t\|$ and proving the convergence of $\| \epsilon_t\|$ and $\|\Delta_t\|$.

\begin{lemma} \label{lem-taumin}
\begin{enumerate}[i.]
\item[]
\item $\tau_{\min}(J_{t+1}^{-1}) \geq \min\{ \tilde{\epsilon_0}^{-1}, \tau_{\min}(J_t^{-1}) \}$ $\forall t \geq 1$.
\item $\tau_{\min}(J_t^{-1}) \geq \xi$, where $\xi = \min\{\tau_{\min}(J_1^{-1}), \tilde{\epsilon_0}^{-1}\}$ $\forall t \geq 1$.
\end{enumerate}
\end{lemma}

\begin{proof}
For i, note that $\| \epsilon_{t+1} \| < \|\epsilon_t\|$ $\forall t \geq 0$ from \eqref{epsilon bd}, which implies $\tilde{\epsilon}_t \leq \tilde{\epsilon}_{t-1} \leq \dots \leq  \tilde{\epsilon}_0$ and $\tilde{\epsilon}_t^{-1} \geq \tilde{\epsilon}_0^{-1}$. From \eqref{ineq for all k} and Lemma \ref{lem-SDhat} ii, for $t \geq 1$,
\begin{gather*}
\tau_k(J_{t+1}) \leq \tau_k(H_{t+1}) \leq \max(\tilde{\epsilon_t}, \tau_k(J_t)) \quad \forall k 
\\
\implies \tau_k(J_{t+1}) \leq \tilde{\epsilon_t} 
\quad \text{or} \quad 
\tau_k(J_{t+1})  \leq \tau_k(J_t) \quad \forall k
\\
\implies \tau_k(J_{t+1}^{-1}) \geq \tilde{\epsilon_t}^{-1}
\quad \text{or} \quad 
\tau_k(J_{t+1}^{-1})  \geq \tau_k(J_t^{-1}) \quad \forall k
\\
\implies \tau_k(J_{t+1}^{-1}) \geq \tilde{\epsilon_0}^{-1}
\quad \text{or} \quad 
\tau_k(J_{t+1}^{-1})  \geq \tau_k(J_t^{-1}) \quad \forall k.
\end{gather*}
Hence $\tau_{\min}(J_{t+1}^{-1}) \geq \min\{ \tilde{\epsilon_0}^{-1}, \tau_{\min}(J_t^{-1}) \}$ $\forall t \geq 1$.

For ii, consider a proof by induction. If $t=1$, then the statement holds trivially. Now, assume $\tau_{\min}(J_t^{-1}) \geq \xi$ for some $t \geq 1$. Then from i,
\begin{align*}
\tau_{\min}(J_{t+1}^{-1}) &\geq \min\{ \tilde{\epsilon_0}^{-1}, \tau_{\min}(J_t^{-1}) \} \\
& \geq \min\{ \tilde{\epsilon_0}^{-1}, \min\{\tau_{\min}(J_1^{-1}), \tilde{\epsilon_0}^{-1}\} \} \\
& \geq \min\{\tau_{\min}(J_1^{-1}), \tilde{\epsilon_0}^{-1}\} = \xi.
\end{align*}
\end{proof}

Now, we will prove the convergence to zero of $\| \epsilon_t\|$ and $\|\Delta_t\|$ by assuming a constant stepsize $\beta_t=\beta$ $\forall t$. Let $\delta = 1-(1-\beta)\xi \in (0,1)$. From \eqref{epsilon bd} and Lemma \ref{lem-taumin} ii, 
\begin{align*}
\| \epsilon_{t+1} \| &\leq \{ 1 - (1-\beta) \xi \}  \|
\epsilon_t\| \quad \forall t \geq 0 \\
& = \delta  \|\epsilon_t\| \\
& \leq ... \\
& \leq \delta^{t+1}  \|\epsilon_0\|.
\end{align*}
Thus $\| \epsilon_{t+1}\| \to 0$ as $t \to \infty$. From the above result and \eqref{Delta bd}, for $t \geq 1$,
\begin{align*}
|\Delta_{t+1}\| &\leq (1-\beta) \|\epsilon_t\|^2 + \beta\|\Delta_t\|
\\
&\leq \beta\|\Delta_t\| + (1-\beta) \delta^{2t} \|\epsilon_0\|^2 \\
& \leq \beta \{ \beta\|\Delta_{t-1}\| + (1-\beta) \delta^{2(t-1)} \|\epsilon_0\|^2  \}\\
& \quad + (1-\beta) \delta^{2t} \|\epsilon_0\|^2 \\
& \leq ... \\
& \leq \beta^t \| \Delta_1 \| + (1-\beta) \|\epsilon_0\|^2 \sum_{j=0}^{t-1} \beta^j \delta^{2(t-j)} 
\\ 
&= \beta^t \| \Delta_1 \| + \frac{\delta^2(1-\beta) \|\epsilon_0\|^2}{(\delta^2-\beta)} (\delta^{2t} -\beta^t).
\end{align*}
Thus $\|\Delta_{t+1}\| \to 0$ as $t \to \infty$.

\section{Batch approximated objective under mean-field}
In this section, we provide the proofs of Lemma \ref{Lem2} and \ref{Lem3} and Theorem \ref{thm2}.

\subsection{Proof of Lemma \ref{Lem2}}
Differentiating $\hat{S}_{q}(\lambda)$ and $\hat{F}_{q}(\lambda)$ with respect to $\mu$ and $\Sigma_{ii}$, we obtain 
\begin{align*}
\nabla_\mu \hat{S}_{q}(\lambda) &= 2\Sigma^{-1} (\mu - \overline{\theta}) - 2 \overline{g}_h, \\
\nabla_\mu \hat{F}_{q}(\lambda) &= \Sigma^{-1} \nabla_\mu \hat{S}_{q}(\lambda), \\
\nabla_{\Sigma_{ii}} \hat{S}_{q}(\lambda) &= V_{ii} -U_{ii}\Sigma_{ii}^{-2},\\
\nabla_{\Sigma_{ii}} \hat{F}_{q}(\lambda) &= - 2\Sigma_{ii}^{-2} (U_{ii}\Sigma_{ii}^{-1} + W_{ii}).
\end{align*}
Setting these derivatives to zero yields
\begin{equation}\label{mini-mu}
\begin{aligned}
\mu_i^{\hat{S}} &= \overline{\theta}_i + \Sigma_{ii}^{\hat{S}} \overline{g}_{h,i}, 
&V_{ii}(\Sigma^{\hat{S}}_{ii})^2 =  C_{\theta, ii} + (\mu_i^{\hat{S}} -  \overline{\theta}_i )^2, \\
\mu_i^{\hat{F}} &= \overline{\theta}_i + \Sigma_{ii}^{\hat{F}} \overline{g}_{h,i}, 
&\Sigma^{\hat{F}}_{ii} = - \frac{C_{\theta, ii} + (\mu_i^{\hat{F}} -  \overline{\theta}_i )^2}{ C_{\theta g, ii} - \overline{g}_{h, i}(\mu_i^{\hat{F}} -  \overline{\theta}_i )}.
\end{aligned}
\end{equation}
Solving these equations simultaneously, we obtain
\begin{equation*}
\begin{aligned}
V_{ii} (\Sigma^{\hat{S}}_{ii})^2 &= C_{\theta,ii} + (\Sigma^{\hat{S}}_{ii})^2 \overline{g}_{h,i}^2 \implies \Sigma^{\hat{S}}_{ii} = \sqrt{C_{\theta,ii}/C_{g, ii}}, \\
\Sigma^{\hat{F}}_{ii} &= - \{C_{\theta, ii} + (\Sigma_{ii}^{\hat{F}})^2 \overline{g}_{h,i}^2 \} / \{ C_{\theta g, ii} - \Sigma_{ii}^{\hat{F}} \overline{g}_{h, i}^2 \}\\
&\quad \implies \Sigma^{\hat{F}}_{ii} =- C_{\theta,ii} / C_{\theta g,ii}.
\end{aligned}
\end{equation*}
Plugging these values into \eqref{mini-mu} yields corresponding values for $\mu_i^{\hat{S}}$ and $\mu_i^{\hat{F}}$. 

\subsection{Proof of Lemma \ref{Lem3}}
The first two results follow directly from the law of large numbers. For the target, $g_h(\theta_i) = -\Lambda(\theta_i - \nu)$. Thus
$\overline{g}_h = - \Lambda(\overline{\theta} - \nu)$ and $g_h(\theta_i)-\overline{g}_h = -\Lambda(\theta_i-\overline{\theta})$. 
\begin{align*}
\therefore \quad C_g &= \frac{1}{B} \sum_{i=1}^B \Lambda(\theta_i - \overline{\theta})(\theta_i - \overline{\theta})^\top\Lambda=\Lambda C_\theta\Lambda, \\
C_{\theta g} &= -\frac{1}{B} \sum_{i=1}^B(\theta_i-\overline{\theta})(\theta_i-\overline{\theta})^\top\Lambda=-C_\theta\Lambda.
\end{align*}
By the continuous mapping theorem \citep{Durrett2019}, $\overline{g}_h  \xrightarrow{\text{a.s.}} \Lambda(\nu-\hat{\mu})$, $C_g  \xrightarrow{\text{a.s.}} \Lambda\hat{\Sigma}\Lambda$ and $C_{\theta g}  \xrightarrow{\text{a.s.}} -\hat{\Sigma}\Lambda$.

\subsection{Proof of Theorem \ref{thm2}}
Results can be obtained by applying the continuous mapping theorem on Lemma \ref{Lem2} and using the results in Lemma \ref{Lem3}. Note that  
$(\Lambda\hat{\Sigma}\Lambda)_{ii} = \sum_{j=1}^d \hat{\Sigma}_{jj} \Lambda_{ij}^2$ and $(\hat{\Sigma}\Lambda)_{ii} = \hat{\Sigma}_{ii} \Lambda_{ii}$.

\section{Gradients for logistic regression}
The log joint density of the model, gradient and Hessian are given by
\begin{align*}
\log h(\theta) &= y^\top X \theta - \sum_{i=1}^n \log\{ 1 + \exp(X_i^\top \theta)\}  \\
&\quad  -\tfrac{d}{2} \log (2\pi \sigma_0^2) - \theta^\top\theta/(2\sigma_0^2), \\
\nabla_\theta \log h(\theta) &= X^\top (y - w)  -\theta/\sigma_0^2, \\ \nabla_\theta^2 \log h(\theta) &= -X^\top WX -I_d/\sigma_0^2,
\end{align*}
where $w = (w_1, \dots, w_n)^\top$, $w_i = \{1 + \exp(-X_i^\top \theta)\}^{-1}$ for $i=1, \dots, n$, $W$ is an $n \times n$ diagonal matrix with diagonal entries $w_i(1-w_i)$ and $X = (X_1, \dots, X_n)^\top$.

\section{Gradients for GLMMs}
The log joint density of the model can be written as 
\begin{align*}
\log h(\theta) & = \sum_{i=1}^n \sum_{j=1}^{n_i} \log p(y_{ij}|\beta, b_i) + \sum_{i=1}^n \log p(b_i|\zeta) \\
&\quad + \log p(\beta) + \log p(\zeta)  \\
& = \sum_{i,j} \{y_{ij}\eta_{ij}- A(\eta_{ij})\} + n \log|W| \\
&\quad -\frac{1}{2} \sum_{i=1}^n b_i^\top  WW^{\top}b_i - \frac{\beta^\top \beta}{2\sigma_\beta^2} -\frac{\zeta^\top \zeta}{2\sigma_\zeta^2} +C,
\end{align*}
where $A(\cdot)$ is the log-partition function and $C$ is a constant independent of $\theta$. For instance, $A(x)=\log(1+e^x)$ for Bernoulli-distributed binary responses and $A(x)=\exp(x)$ for Poisson-distributed count responses.

Let $X_i=(X_{i1},\dots,X_{in_i})^\top$ and $Z_i=(Z_{i1},\dots,Z_{in_i})^\top$ be design matrices for the $i$th subject. Recall that $b_i\sim \N(0,G^{-1})$, $G=WW^\top$, $W^*$ is such that $W_{ii}^* = \log(W_{ii})$ and $W_{ij}^* = W_{ij}$ if $i \neq j$, and $\zeta = \text{vech}(W^*)$. Let $J^W$ be an $r \times r$ matrix with diagonal given by $\diag(W)$ and all off-diagonal entries being 1, and $D^W = \diag\{ \vech(J^W) \}$. Then $d \vech(W) = D^W d\vech(W^*)$. We have 
\begin{multline*}
\nabla_\theta \log h(\theta) = [\nabla_{b_1} \log h(\theta)^\top, \dots, \nabla_{b_n} \log h(\theta)^\top, \\
\nabla_{\beta} \log h(\theta)^\top, \nabla_{\zeta} \log h(\theta)^\top]^\top,
\end{multline*}
where
\begin{align*}\label{grad_GLMM}
\nabla_{b_i} \log h(\theta) &= \sum_{j=1}^{n_i} \{y_{ij} - A'(\eta_{ij})  \}Z_{ij} - G b_i,\\
\nabla_{\beta} \log h(\theta) &= \sum_{i=1}^n \sum_{j=1}^{n_i} \{y_{ij} -A'(\eta_{ij}) \}X_{ij} -\frac{\beta}{\sigma_{\beta}^2}, \\
\nabla_{\zeta} \log h(\theta) &=  -D^W \vech(\widetilde{W}) +n\vech(I_r) -\frac{\zeta}{\sigma_\zeta^2}, 
\end{align*}
and $\widetilde{W}=\sum_{i=1}^{n} b_i b_i^{\top}W$.

Let $H_{\theta_i,\theta_j} = \nabla^2_{\theta_i,\theta_j} \log h(\theta)$. The Hessian takes the block form
\begin{equation*}\label{Has_mat_glmm}
\begin{aligned}
H = \begin{bmatrix}
H_{b_1, b_1}  & \ldots & 0 & H_{b_1,\theta_G} \\
\vdots & \ddots & \vdots & \vdots \\
0 & \ldots & H_{b_n, b_n} & H_{b_n,\theta_G}\\
H_{\theta_G,b_1} & \ldots & H_{\theta_G,b_n} & H_{\theta_G} 
\end{bmatrix}.
\end{aligned}
\end{equation*}
We have
\begin{align*}
\nabla_{b_i,\eta}^2 \log h(\theta) &=\begin{bmatrix}
\nabla_{b_i,\beta}^2 \log h(\theta)\\
\nabla_{b_i,\zeta}^2 \log h(\theta)\\
\end{bmatrix}, \\
\nabla_{\eta}^2 \log h(\theta) &=\begin{bmatrix}
\nabla_{\beta}^2 \log h(\theta) & 0\\
0 & \nabla_{\zeta}^2 \log h(\theta)\\
\end{bmatrix}.
\end{align*}

Let $B_i = \diag([A''(\eta_{i1}),\dots,A''(\eta_{in_i})]^\top)$. The second order derivatives of $\log h(\theta)$ are
\begin{align*}
\nabla_{b_i}^2 \log h(\theta) &=-(Z_i^{\top}B_iZ_i+G),\text{ for },\\
\nabla_{\beta}^2 \log h(\theta) &= - \left( \sum_{i=1}^{n} X_i^{\top} B_i X_i + \frac{1}{\sigma^2_{\beta}} I_p \right),\\
\nabla_{\zeta}^2 \log h(\theta) &= -S -D^WL \sum_{i=1}^{n} (I_r\otimes b_ib_i^{\top}) L^\top D^W \\
&\quad - \frac{1}{\sigma^2_{\zeta}}I_{r(r+1)/2}, \\
\nabla_{\beta, b_i}^2 \log h(\theta) &=- X_i^\top B_i Z_i, \\
\nabla_{\zeta, b_i}^2 \log h(\theta) &= -D^WL(W^{\top}b_i\otimes I_r + W^{\top}\otimes b_i),
\end{align*}
where $S = \diag[ \vech\{ \dg(W) \dg(\widetilde{W}) \} ]$ and $\dg(A)$ is a copy of $A$ with all off-diagonal entries set to 0. 

The derivations for $\nabla_{b_i}^2 \log h(\theta)$, $\nabla_{\beta}^2 \log h(\theta)$ and $\nabla_{b_i,\beta}^2 \log h(\theta)$ are straightforward. More 
details for $\nabla_{\zeta^2}^2 \log h(\theta)$ and $\nabla_{b_i,\zeta}^2 \log h(\theta)$ are given below.
Differentiating $\nabla_{\zeta} \log h(\theta)$ w.r.t. $b_i$, we have
\begin{align*}
d \nabla_{\zeta} & \log h(\theta) =-D^W\sum_{i=1}^{n} \vech \{ (d b_i) b_i^{\top}W+b_i (db_i^{\top}) W \} \\
&=-D^WL\sum_{i=1}^{n}[(W^{\top}b_i\otimes I_r)+(W^{\top}\otimes b_i)]d  b_i.
\end{align*}
Differentiating $\nabla_{\zeta} \log h(\theta)$ w.r.t. $\zeta$, we have
\begin{align*}
& d \nabla_{\zeta} \log h(\theta)= - (d D^W) \sum_{i=1}^{n}\vech(b_ib_i^{\top}W) \\
&\quad - D^W\sum_{i=1}^{n}\vech \{b_ib_i^{\top} (d W) \} - \frac{1}{\sigma^2_{\zeta}} d \zeta \\
&= -D^WL\sum_{i=1}^{n} (I_r \otimes b_i b_i^{\top}) d \vec(W) \\
&\quad -Sd \zeta - \frac{1}{\sigma^2_{\zeta}} d \zeta \\
&=-D^WL\sum_{i=1}^{n} (I_r \otimes b_i b_i^{\top}) L^T D^W d \zeta \\
&\quad -Sd \zeta - \frac{1}{\sigma^2_{\zeta}} d \zeta.
\end{align*}

\section{Gradients for stochastic volatility model}
For this model, the log joint density is
\begin{align*}
\log h(\theta) &= -\frac{n\lambda}{2} - \frac{\sigma}{2}\sum_{t=1}^{n}b_t 
- \frac{1}{2}\sum_{t=1}^{n}y_t^2 \exp\{-\lambda - \sigma b_t \}\\
&\quad- \frac{1}{2}\sum_{t=2}^{n}(b_t-\phi b_{t-1})^2 +\frac{1}{2}\log (1-\phi^2) \\
&\quad  -\frac{1}{2}b_1^2(1-\phi^2) - \frac{\alpha^2}{2\sigma_0^2} - \frac{\lambda^2}{2\sigma_0^2} - \frac{\psi^2}{2\sigma_0^2} + C,
\end{align*}
where $C$ is a constant independent of $\theta$. The gradients of $ \log h(\theta)$ are,
\begin{align*}
\nabla_{b_1} \log h(\theta) &= - (1-\phi^2) b_1 + \phi (b_2 - \phi b_1)  -\frac{e^\alpha}{2} \\
&\quad +  \frac{ e^\alpha y_1^2}{2} \exp(-\lambda - e^\alpha b_1), \\
\nabla_{b_t} \log h(\theta) &= \phi (b_{t+1} - \phi b_t) - (b_t - \phi b_{t-1}) -\frac{e^\alpha}{2}  \\
&\quad+  \frac{ e^\alpha}{2} y_t^2\exp(-\lambda - e^\alpha b_t) \;\; \text{for} \; 1 < t < n,  \\
\nabla_{b_n} \log h(\theta) &= - (b_n - \phi b_{n-1}) -\frac{e^\alpha}{2}  \\
&\quad +  \frac{ e^\alpha}{2} y_n^2\exp(-\lambda - e^\alpha b_n), \\
\nabla_{\alpha} \log h(\theta) &=  \frac{1}{2} \sum_{t=1}^n y_t^ 2b_t \exp(\alpha -\lambda - e^\alpha b_t)  \\
&\quad - \frac{e^\alpha}{2} \sum_{t=1}^n b_t  - \frac{\alpha}{\sigma_0^2} , \\
\nabla_{\lambda} \log h(\theta) &=  -\frac{n}{2}  +  \frac{1}{2} \sum_{t=1}^n y_t^2\exp(-\lambda - e^\alpha b_t)  - \frac{\lambda}{\sigma_0^2}, \\
\nabla_{\psi} \log h(\theta) &= \Big\{ \phi b_1^2- \frac{\phi}{(1-\phi^2)} +  \sum_{t=1}^{n-1} (b_{t+1} - \phi b_t)b_t  \Big\}\\
&\quad \times\frac{e^\psi}{(e^\psi + 1)^2}  - \frac{\psi}{\sigma_0^2}. \\
\end{align*}

The Hessian has a sparse block structure,
\begin{equation*}
\nabla_\theta^2 \log h(\theta) =
\begin{bmatrix}
H_{b_1, b_1} & H_{b_1,b_2} & \ldots & 0  & H_{b_1,\theta_G} \\
H_{b_2,b_1} & H_{b_2, b_2} & \ldots & 0 & H_{b_2,\theta_G} \\
\vdots & \vdots & \ddots & \vdots & \vdots \\
0 &  0 &  \dots & H_{b_n, b_n} & H_{b_n,\theta_G} \\
H_{\theta_G,b_1} & H_{\theta_G,b_2} &\ldots & H_{\theta_G,b_n} & H_{\theta_G, \theta_G} \\
\end{bmatrix}.
\end{equation*}
The second order derivatives of $\log h(\theta)$ are,
\begin{align*}
\nabla_{b_1}^2\log h(\theta) &= -1-\frac{y_1^2}{2} \exp\{2\alpha -\lambda-e^{\alpha}b_1\},\\
\nabla_{b_t}^2\log h(\theta) &= -\phi^2 - 1 - y_t^2 \exp\{2\alpha-\lambda-e^{\alpha}b_t\}/2,\\
\nabla_{b_n}^2\log h(\theta) &= -1 - y_n^2 \exp\{2\alpha -\lambda-e^{\alpha}b_n\}/2,\\
\nabla_{b_i,b_j}^2\log h(\theta) &= \phi \mathbbm{1}_{|i-j|=1}, \\
\nabla_{b_t,\alpha}^2\log h(\theta) &= \frac{y_t^2}{2} \exp\{\alpha-\lambda-e^{\alpha}b_t\}(1-b_te^{\alpha})-\frac{e^{\alpha}}{2},
\\
\nabla_{b_t,\lambda}^2 \log h(\theta)&= - y_t^2 \exp\{\alpha -\lambda-e^{\alpha}b_t\}/2 
\\
\nabla_{b_1,\psi}^2\log h(\theta) &= \frac{b_2 e^\psi}{(e^\psi + 1)^2},
\\
\nabla_{b_t,\psi}^2\log h(\theta) &= \frac{e^\psi (b_{t+1} - 2\phi b_t + b_{t-1})}{(e^\psi + 1)^2},
\\
\nabla_{b_n,\psi}^2\log h(\theta)&= \frac{e^\psi b_{n-1}}{(e^\psi + 1)^2},
\\
\nabla_{\alpha}^2 \log h(\theta) &= \frac{1}{2} \sum_{t=1}^{n} y_t^2 b_t \exp\{\alpha-\lambda-e^{\alpha}b_t\}(1-e^{\alpha}b_t)\\
&\quad -\frac{e^{\alpha}}{2}\sum_{t=1}^{n}b_t-\frac{1}{\sigma_0^2},
\\
\nabla_{\lambda}^2\log h(\theta) &= -\frac{1}{2}\sum_{t=1}^{n}y_t^2\exp(-\lambda-e^{\alpha}b_t)-\frac{1}{\sigma^2_0},
\\
\nabla_{\psi}^2\log h(\theta) &= \bigg\{ b_1^2-\sum_{t=1}^{n-1}b_t^2-\frac{1+\phi^2}{(1-\phi^2)^2} \bigg\} \frac{e^{2\psi}}{(e^\psi + 1)^4} 
\\
& + \bigg\{ \phi b_1^2- \frac{\phi}{(1-\phi^2)} +  \sum_{t=1}^{n-1} (b_{t+1} - \phi b_t)b_t  \bigg\}\\
& \times\frac{e^{\psi}(1 - e^{\psi})}{(e^{\psi}+1)^3}- \frac{1}{\sigma_0^2},
\\
\nabla_{\alpha,\lambda}^2\log h(\theta) &= -\frac{1}{2}\sum_{t=1}^{n} y_t^2 b_t \exp\{\alpha - \lambda-e^{\alpha}b_t\},
\\
\nabla_{\psi, \lambda}^2 \log h(\theta) &= \nabla_{\psi, \alpha}^2 \log h(\theta) = 0.
\end{align*}

\end{document}